\documentclass{article}
\usepackage{color, xcolor, colortbl}
\usepackage{graphicx}
\usepackage{geometry}
\usepackage{amsthm,amsmath,amssymb,amsfonts}
\usepackage{algorithm}
\usepackage{algorithmic}
\usepackage{dcolumn}
\usepackage{epstopdf}
\usepackage{bm}
\usepackage[caption=false]{subfig}
\usepackage{appendix}
\usepackage{multirow}
\usepackage{braket}
\usepackage[english]{babel}
\usepackage{hyperref}
\usepackage[capitalize]{cleveref}
\usepackage{xpatch}
\usepackage{tikz}
\usepackage{adjustbox}
\usepackage{xspace}

\newcommand{\mc}[1]{\mathcal{#1}}

\newcommand{\norm}[1]{\left\lVert#1\right\rVert}

\theoremstyle{definition}
\newtheorem{thm}{Theorem}
\newtheorem{lem}[thm]{\protect\lemmaname}
\newtheorem{rem}{\protect\remarkname}

\newtheorem{cor}[thm]{\protect\corollaryname}
\newtheorem{assumption}[thm]{\protect\assumptionname}
\newtheorem{defn}[thm]{\protect\definitionname}
\newtheorem{example}[thm]{\protect\examplename}

\newtheorem{proposition}[thm]{Proposition}
\newtheorem{corollary}[thm]{Corollary}

\numberwithin{equation}{section}
\numberwithin{thm}{section}
\numberwithin{rem}{section}

\providecommand{\definitionname}{Definition}
\providecommand{\assumptionname}{Assumption}
\providecommand{\corollaryname}{Corollary}
\providecommand{\lemmaname}{Lemma}
\providecommand{\propositionname}{Proposition}
\providecommand{\remarkname}{Remark}

\providecommand{\examplename}{Example}

\usetikzlibrary{fit}
\tikzset{%
  highlight/.style={rectangle,rounded corners,fill=blue!15,draw,fill opacity=0.3,thick,inner sep=0pt}
}

\usepackage{authblk}

\usepackage[normalem]{ulem}
\usepackage{amsthm}

\usetikzlibrary{decorations.pathreplacing,calc} 
\allowdisplaybreaks
\begin{document}

\title{Unified error bounds for perturbations of non‑Markovian open quantum systems in Gaussian environments}
\author[1]{Zhen Huang}
\author[2]{Yuanran Zhu}
\author[3]{Gunhee Park}
\author[1,2]{Lin Lin}
\affil[1]{Department of Mathematics, University of California, Berkeley, California, 94720, USA}
\affil[2]{Applied Mathematics and Computational Research Division, Lawrence Berkeley National Laboratory, Berkeley, California, 94720, USA}
\affil[3]{Division of Engineering and Applied Science, California Institute of Technology, Pasadena, California, 91125, USA}

\maketitle
\begin{abstract}
We present perturbative error bounds for the non-Markovian dynamics of observables  in  open quantum systems interacting with Gaussian environments, governed by a general Liouville dynamics. This extends the work of [Mascherpa et al., Phys. Rev. Lett. 118, 100401, 2017], which demonstrated qualitatively tighter bounds over the standard Gr\"onwall-type inequality for unitary system-bath evolution. Our results apply to systems with both bosonic and fermionic environments.  Our approach utilizes a superoperator formalism, which avoids the need for formal coherent state path integral calculations, or the dilation of Lindblad dynamics into an equivalent unitary framework with infinitely many degrees of freedom. This enables a unified treatment of a wide range of open quantum systems.  These findings provide a solid theoretical basis for various recently developed pseudomode methods in simulating open quantum system dynamics.
\end{abstract}

\vspace{1em}
\section{Introduction}
\label{sec:intro}
Quantum systems that interact with environments characterized by a large number of degrees of freedom are known as open quantum systems \cite{breuer2007open}. Under certain conditions, such as weak coupling or the presence of well-separated time scales, a Markovian approximation may be effectively applied, allowing for the neglect of environmental memory effects.  However, in general cases where information exchange occurs between the system and environment, these non-Markovian memory effects become significant, rendering the Markovian approximation insufficient. Such non-Markovian effects play a critical role in various disciplines, including  condensed matter physics \cite{LeggettRevModPhys.59.1, Hewson1993}, solid-state quantum devices~\cite{ChirolliBurkard2008review,Rotter_2015}, chemical physics \cite{Ishizaki2009, HuelgaPlenio2013review}, quantum thermodynamics~\cite{EspositoOchoa2015, TalknerHanggi2020} and quantum computation \cite{li2023succinct}.

Consider the dynamics of a quantum system (denoted by S) coupled with an environment (denoted by E), with $\mathcal H_{\text S}$ and $\mathcal H_{\text E}$ being the Hilbert spaces for the system and the environment, respectively. The total Hilbert space is $\mathcal H = \mathcal H_{\text E} \otimes \mathcal H_{\text S}$. Given an observable $\hat{O}_S$ acting on $\mathcal H_{\text S}$, the goal is to compute its time-dependent expectation value denoted by $O_{\text S}(t)$ by tracing out the contribution from $\mathcal H_{\text E}$ (see \cref{eq:rho_partial_trace}, \cref{eqn:expectation_O}). The most common example is that the system and environment together are described by a unitary dynamics on $\mathcal H$. By tracing out $\mathcal H_{\text E}$, the effective dynamics (called an open quantum dynamics) on $\mathcal H_{\text S}$ is in general non-Markovian. To make this problem analytically tractable, the environment as well as the system-environment coupling must be sufficiently simple. For instance, in a spin-boson model, the spin degree of freedom is viewed as the system, and the bosonic modes are viewed as the environment. In a fermionic impurity model, a few fermionic modes describing the fermionic impurities are viewed as the system, and the remaining fermionic modes describing the conduction electrons are viewed as the environment. Many numerical methods for simulating open quantum systems~\cite{PriorChinHuelgaEtAl2010, Vega2015, DeVegaAlsonso2017, TamascelliSmirneLimEtAl2019} are formulated in this unitary setting.

However, when $\mathcal H$ is a finite-dimensional space, the unitary dynamics on $\mathcal H$ only has a finite number of modes. Therefore, a large number of degrees of freedom in the environment may be needed to describe the long-time behavior of the system dynamics correctly.  An alternative way to tackle this problem is to consider a more general, non-unitary Liouville dynamics on $\mc{H}$, such as the Lindblad dynamics. This Liouville description  encompasses the unitary dynamics as a special case. It is also compatible with the unitary description, in the sense that a Lindblad dynamics can always be described by a unitary dynamics using an infinitely sized environment via a dilation process~\cite{HudsonParthasarathy1984,TamascelliSmirneHuelga2018}.  The advantage is that the number of degrees needed for the environment using the Lindblad dynamics can be significantly smaller.  This is the basis of various ``pseudomode'' theories used in practice \cite{Garraway1997,Schwarz2016,TamascelliSmirneHuelga2018,Mascherpa2020,Lotem2020, Brenes2020, Zwolak2020,Trivedi2021,lednev2024lindblad}. 

In this paper, we consider the open quantum system dynamics undergoing a general Liouville dynamics (see \cref{sec:preliminaries} for the setup in an abstract setting, and \cref{sec:applications} for applications to spin-boson and fermionic models). We assume that in the interaction picture formulation, each system-environment coupling operator satisfies Wick's conditions (see \cref{defn:Wick_condition}). Such an environment is called Gaussian, since the multi-point correlation function is characterized by the two-point bath correlation functions (BCFs), denoted as $C_{\alpha,\alpha'}(t-t')$ ($\alpha,\alpha' = 1,\cdots, N$)\footnote{This is the most common scenario in practice. Our results can also be stated with respect to a more general bath correlation function of the form $C_{\alpha,\alpha'}(t,t')$ (see \cref{rmk:single_variable_Corr}).}. Here, $N$ denotes the number of system-environment coupling operators. In such systems, the environment's influence on the reduced system dynamics is completely described by the BCF~\cite{breuer2007open}. However, we may only have approximate knowledge of the BCF, e.g., its values may be obtained by fitting to experimental results. Our main question is:

\vspace{1em}
\emph{How do errors in the bath correlation functions affect the expectation values of system observables?}
\vspace{0.5em}

Specifically, for a given system observable $\hat{O}_S$, if the two-point BCF is perturbed by a small quantity $\Delta C(t-t')$, how does the time-dependent expectation value $O_{\text S}(t)$ depend on the perturbation?

Since the Liouville dynamics is linear, we can derive the following  error estimate for $O_{\text S}(t)$, which we refer to as a Gr\"{o}nwall-type bound due to its resemblance to bounds obtained using the classical Gr\"{o}nwall inequality:
\begin{thm}[Gr\"{o}nwall-type error bound for system observables]
    Let $C(t-t')$ and $C'(t-t') = C(t-t')+\Delta C(t-t')$ be two-point BCFs corresponding to two different environments. For a bounded system operator $\hat O_{\text S}$, let $O_{\text S}(t)$, $O_{\text S}'(t)$ be the expectation value of system observables corresponding to the two environments. Then, we have
    $$|O_{\text S}(t) - O_{\text S}'(t)| \leq \|\hat O_{\text S}\| {\epsilon_1 t}\mathrm e^{\mathcal M_1 t}, \quad \text { for } t \in[0, T],$$
    where $\|\hat O_{\text S}\|$ is the operator norm of $\hat O_{\text S}$, $\mathcal M_1=  \max\{ \sum_{\alpha,\alpha'=1}^N   \|  C_{\alpha,\alpha'}  \|_{L^{1}([0,T])}, \sum_{\alpha,\alpha'=1}^N   \|  C'_{\alpha,\alpha'} \|_{L^{1}([0,T])}\}$, $\epsilon_1 = \sum_{\alpha,\alpha'=1}^N   \|  C_{\alpha,\alpha'} - C'_{\alpha,\alpha'} \|_{L^{1}([0,T])} $,  and $\|\cdot\|_{L^1[0,T]}$ is the $L^1$ norm of a function on $[0,T]$.
\label{lem:Gronwall_informal}
\end{thm}

According to \cref{lem:Gronwall_informal}, achieving a precision of $\epsilon$ for $O_{\text S}(t)$ up to time $T$ requires that the BCF be computed with a precision of $\epsilon_1 = O\left(\epsilon e^{-\mathcal{M}_1 T}\right)$. This stringent requirement becomes impractical for any time $T$ beyond $\log(1/\epsilon)$. In this work, we establish a substantially improved error bound that relaxes the required precision to $\epsilon_1 = O\left(\epsilon / T\right)$. Such an improvement allows simulations to be conducted over significantly longer time intervals while maintaining controlled precision on system observables.

\begin{thm}[Main result, for system observables, also see {\cref{thm:main_error_bound}}]
    Let $C(t-t')$, $C'(t-t')$, $O_{\text S}(t)$, $O_{\text S}'(t)$ and $\epsilon_1$ be the same as in Lemma \ref{lem:Gronwall_informal}. Then, we have
    $$ |O_{\text S}(t) - O_{\text S}'(t)| \leq \|\hat O_{\text S}\|(\mathrm{e}^{ \epsilon_1 t}-1), \quad \text { for } t \in[0, T]. $$
    \label{thm:main}
\end{thm}

This error bound can be directly applied to simulating non-Markovian open quantum systems using pseudomode theories based on the Lindblad dynamics \cite{Garraway1997,Schwarz2016,TamascelliSmirneHuelga2018,Mascherpa2020,Lotem2020, Brenes2020, Zwolak2020,Trivedi2021,lednev2024lindblad}. To further reduce the number of degrees of freedom in the environment, a number of ``quasi-Lindblad'' pseudomode theories have also been proposed in recent years \cite{Lambert2019, Pleasance2020, Cirio2023, menczel2024nonhermitian,ParkHuangZhuetal2024,ThoennissVilkoviskiyAbanin2024}, which make use of non-Hermitian dynamics, or dynamics that may violate the complete positivity (CP) condition. We establish an error bound analogous to \cref{thm:main} in \cref{cor:quasi_lind_dynamics} under the setting of \cite{ParkHuangZhuetal2024}. Our result is also relevant beyond pseudomode theories, such as in the context of hierarchical equations of motion~\cite{Tanimura1989, Tanimura2020, Stockburger2022, DanShi2023, Nori2023qutip, xu2023universal, ivander2024unifiedframeworkopenquantum}, which can be interpreted as a certain quasi-Lindblad dynamics~\cite{xu2023universal, ivander2024unifiedframeworkopenquantum}.  

To our knowledge, an improved error bound similar to \cref{thm:main} was first proved in \cite{MascherpaSmirneHuelgaetal2017} for spin-boson systems undergoing a unitary dynamics. Their proof uses coherent state path integrals. However, the spin coherent path integrals, and other continuous-time  coherent state path integrals, may exhibit subtleties and unexpected breakdowns \cite{WilsonGalitski2011, Kochetov2019}. As a result, derivations based on coherent state path integrals should not be identified with rigorous proofs in general (see also \cite{KordasMistakidisKaranikas2014, KORDAS2016226}). Recently,  \cite{LiuLu2024} provides a rigorous proof under the settings of \cite{MascherpaSmirneHuelgaetal2017} for a unitary spin-boson dynamics, using a diagrammatic and combinatorial argument similar to that used in the analysis of inchworm algorithms \cite{CaiLuYang2020}. It is not yet clear whether the framework in \cite{LiuLu2024} could be easily generalized to non-unitary dynamics without much complication, especially with Lindblad-like jump operators. Furthermore, though all such arguments could be generalized to fermionic systems in principle, the fermionic case involves the Grassmann number (see, e.g., \cite{NegeleOrland1998}) or equivalent algebraic structures, which can be technically more complicated than their bosonic counterpart. To our knowledge, an error bound of the type of \cref{thm:main} has not been established in the literature for fermionic systems.

In this paper, we provide a mathematically rigorous proof of the error bound in \cref{thm:main} using the superoperator formalism \cite{Aurell_2020,cirio2022canonical}. This method is simple and rigorous and provides a unified analysis framework for various physical settings. By \textit{unified}, we mean that it treats unitary and non-unitary dynamics, as well as bosonic and fermionic environments, on the same footing. In particular, in the case of the fermionic environment, this approach avoids direct manipulation of Grassmann algebras.

This article is organized as follows. In \cref{sec:preliminaries}, we first describe our problem setup, introduce the general Liouville dynamics for open quantum systems, and introduce Wick's conditions. In \cref{sec:gronwall}, we give a proof for the Gr\"{o}nwall-type error bound (see \cref{thm:Gronwall_error_bound}). \cref{sec:main} presents the main results of the paper, where we also provide an informal yet intuitive argument why this improved error bound holds. We give a rigorous proof for our improved error bound (see \cref{thm:main_error_bound}) in \cref{sec:proof_main}, which is valid for general Liouville dynamics. The applications of the main theorem to different open quantum system dynamics are considered in \cref{sec:applications}, where we use the spin-boson model (\cref{sec:spin_boson}) and fermionic impurity model (\cref{sec:fermion}) as examples for bosonic and fermionic environments, respectively. For both models, we discuss the original unitary dynamics and  Lindblad pseudomode dynamics, demonstrating our results for both unitary and non-unitary cases. Finally, we discuss the quasi-Lindblad dynamics in \cref{sec:quasi-Lind}, which violates the CP condition and contractive property.

\subsection*{Acknowledgements}
The work is partially supported by the Simons Targeted Grants in Mathematics and Physical Sciences on Moir\'e Materials Magic (Z.H., L.L.), by the U.S. Department of Energy, Office of Science, Office of Advanced Scientific Computing Research and Office of Basic Energy Sciences, Scientific Discovery through Advanced Computing (SciDAC) program under Award Number DE–SC0022088 (G.P.), DE–SC0022198 (Y.Z.). This material is also based upon work supported by the U.S. Department of Energy, Office of Science, Accelerated Research in Quantum Computing Centers, Quantum Utility through Advanced Computational Quantum Algorithms under Grant Number DE-SC0025572 (L.L.). The authors would like to thank Garnet Kin-Lic Chan, Zhiyan Ding, Jianfeng Lu, Michael Ragone and Chao Yang for helpful discussions. 
\subsection*{Data Availability}
No new data were created or analyzed in this study. Data sharing is not applicable to this article.

\section{Error bounds of open quantum systems under  Liouville dynamics}

We consider the time evolution of  open quantum systems under a Liouville dynamics. In this section, after introducing  the problem setup, we obtain the Dyson series for a reduced system density operator (see \cref{thm:rhoS_Dyson_short}). A Gr\"onwall-type error bound can be derived directly from the series expansion of this Dyson series. Achieving a sharper bound, however, requires a resummation of the infinite series. Formally, this improved error bound can be obtained via (anti-)commuting operators under time-ordering. This formal proof is made rigorous by  employing a rigorous combinatorial argument.

This section provides a unified framework for various physical settings of interest to this work, while in \cref{sec:applications} and \cref{sec:quasi-Lind}, we introduce concrete examples for both bosonic and fermionic environments, and for both unitary and non-unitary dynamics, demonstrating the applications of \cref{thm:rhoS_Dyson_short} and \cref{thm:main_error_bound} in different scenarios.

\subsection{Preliminaries}
\label{sec:preliminaries}

Let \( \mathcal{H}_{\text S} \) and \( \mathcal{H}_{\text E} \) denote the Hilbert spaces of the (sub)system and the environment, respectively. The total Hilbert space of the combined system is then given by $\mathcal{H} = \mathcal{H}_{\text E} \otimes \mathcal{H}_{\text S}$. We assume the subsystem is finite-dimensional. The environment $\mathcal{H}_{\text E}$ is assumed to be a separable Hilbert space, potentially infinite-dimensional (e.g., a bosonic or fermionic Fock space).  We denote by $B(\mathcal{H})$ the Banach space of bounded linear operators on $\mathcal{H}$, equipped with the operator norm $\|\cdot\|$. We denote by $B_1(\mathcal{H})$ the Banach space of trace-class operators on $\mathcal{H}$, equipped with the trace norm $\|\cdot\|_{\text{tr}}$. The state of the system is described by a density operator $\hat\rho \in B_1(\mathcal{H})$, which is positive semi-definite  and has unit trace.

The time evolution of the density operator $\hat{\rho}(t)$ for the composite system (S + E) is governed by a strongly continuous one-parameter semigroup $(\Phi_t)_{t\geq 0} = (\mathrm{e}^{\boldsymbol{L}t})_{t\geq 0}$ acting on $B_1(\mathcal{H})$, such that
\begin{equation}
    \hat{\rho}(t) = \mathrm{e}^{\boldsymbol{L}t} \hat{\rho}(0).
\end{equation}
The generator $\boldsymbol{L}$, or the Liouvillian, is a densely defined, closed superoperator (linear map) on $B_1(\mathcal{H})$. We assume the Liouvillian $\boldsymbol{L}$ decomposes according to the system-environment structure:
\begin{equation}
    \boldsymbol{L} = \boldsymbol{L}_{\text{0}} + \boldsymbol{L}_{\text{\text{SE}}}.
    \label{eq:L_decomp}
\end{equation}
To define these components, we utilize the tensor product of superoperators. Let $\boldsymbol{S}$ and $\boldsymbol{E}$ act on $B_1(\mathcal{H}_{\text S})$ and $B_1(\mathcal{H}_{\text E})$, respectively. The tensor product $\boldsymbol{E} \otimes \boldsymbol{S}$ is defined on elementary tensors by
\begin{equation}
   (\boldsymbol{E} \otimes \boldsymbol{S}) (\hat{A} \otimes \hat{B}) := \boldsymbol{E}(\hat{A}) \otimes \boldsymbol{S}(\hat{B}),
\end{equation}
and extended by linearity and continuity (or by taking the closure, if the operators are unbounded).

Let $\boldsymbol{L}_{\text s}$ and $\boldsymbol{L}_{\text e}$ be the superoperators acting on $B_1(\mathcal{H}_{\text S})$ and $B_1(\mathcal{H}_{\text E})$, respectively. The free part of the Liouvillian is
\begin{equation}
    \boldsymbol{L}_{0} =  (\boldsymbol{1}_{\text e} \otimes \boldsymbol{L}_{\text s})+(\boldsymbol{L}_{\text e} \otimes \boldsymbol{1}_{\text s}),
\end{equation}
where $\boldsymbol{1}_{\text s}$ and $\boldsymbol{1}_{\text e}$ are the respective identity superoperators.  The interaction Liouvillian $\boldsymbol{L}_{\text{SE}}$ takes a separable form:
\begin{equation}
    \boldsymbol{L}_{\text{SE}} = \sum_{\alpha=1}^N \boldsymbol{E}_{\alpha} \otimes \boldsymbol{S}_{\alpha}.
    \label{eq:decomp_detail}
\end{equation}

Since $\mathcal{H}_{\text S}$ is assumed to be finite dimensional, $\boldsymbol{L}_{\text s}$ is bounded. In contrast, when $\mathcal{H}_{\text E}$ is infinite-dimensional, the Liouvillian $\boldsymbol{L}$ is typically unbounded. We assume that $\boldsymbol{L}_{\text{SE}}$ satisfies conditions relative to $\boldsymbol{L}_0$ (e.g., relative boundedness) such that the closure of their sum is the generator of a strongly continuous semigroup on $B_1(\mathcal{H})$~\cite{engel2000one,chebotarev1998sufficient}.

An important case is when the composite system is closed and evolves unitarily. The dynamics is generated by a total Hamiltonian $\hat H$, a self-adjoint operator on $\mathcal{H}$. The Liouvillian is given by the commutator: $ \boldsymbol{L}(\cdot) = -\mathrm{i}[\hat H, \cdot]$.

The Hamiltonian is structured as:
\begin{align*}
\hat H = \hat H_0 + \hat H_{\text{SE}},\quad \hat H_0 = \hat 1_{\text e}\otimes \hat H_{\text s}+\hat H_{\text e}\otimes \hat 1_{\text{s}},\quad \hat H_{\text{SE}} = \sum_{k=1}^{N_{\text U}} \hat E_k\otimes\hat S_k,
\end{align*}
where $\hat H_0$ represents the non-interacting part of the Hamiltonian, and \(\hat H_{\text{SE}}\)  describes the interaction between the subsystem and the environment. Here $\hat 1_{\text e/\text s}$  are identity operators on $\mathcal H_{\text e/\text s}$, and $\hat H_{\text s}, \hat S_k \in B(\mathcal{H}_{\text s})$. The environment operators $\hat H_{\text e}, \hat E_k$ are generally unbounded, densely defined operators on $\mathcal{H}_{\text e}$. The system and environment Liouvillians are 
$
\boldsymbol L_{\text s } = -\mathrm i[\hat H_{\text s }, \cdot]$, $\boldsymbol L_{ \text e} = -\mathrm i[\hat H_{ \text e}, \cdot]$, and the superoperators $\boldsymbol E_{\alpha},\boldsymbol{S}_{\alpha}$ are defined  from the operators $\hat E_k,\hat S_k$  using left- and right-multiplication superoperators (see \cref{sec:applications} for details). 

We utilize the induced trace norm for a superoperator $\boldsymbol{O}$ acting on $B_1(\mathcal{H})$. For a densely defined superoperator $\boldsymbol{O}$ with domain $\operatorname{Dom}(\boldsymbol{O})\subseteq B_1(\mathcal{H})$, the norm is defined as:
\begin{equation}
\|\boldsymbol{O}\|_{\text{tr}}:=\sup \{\|\boldsymbol{O}(\hat A)\|_{\text{tr}} : \hat A\in \operatorname{Dom}(\boldsymbol{O}), \|\hat A\|_{\text{tr}}=1 \}.
\end{equation}
We will repeatedly use the H\"older inequality for operator norms
\begin{equation}\label{eqn:optrace_inequality}
\|\hat O_1\hat O_2\|_{\text{tr}}\leq\|\hat O_1\|\|\hat O_2\|_{\text{tr}}, \quad \hat O_1\in B(\mathcal H), \quad \hat O_2\in B_1(\mathcal H).
\end{equation}
For instance, for $\boldsymbol{O}(\hat A)=[\hat{O},\hat{A}]$ with $\hat O\in B(\mathcal H)$, we have $\|\boldsymbol O\|_{\text{tr}}\le 2\norm{\hat O}$. Because $\mathcal{H}_{\text S}$ is finite-dimensional, each superoperator $\boldsymbol{S}_\alpha$ in \cref{eq:decomp_detail} is bounded. Without loss of generality, we may normalize them so that $\|\boldsymbol{S}_\alpha\|_{\text{tr}} = 1$ for all $\alpha = 1,\cdots, N$.

We start from an uncorrelated initial state: $\hat{\rho}(0) = \hat\rho_{\text S}(0) \otimes \hat\rho_{\text E}(0)$. The object of interest is the reduced dynamics of the subsystem, obtained by taking the partial trace over the environment:
\begin{equation}
    \hat\rho_{\text S}(t) = \operatorname{tr}_{\text E} (\hat\rho(t)).
    \label{eq:rho_partial_trace}
\end{equation}
For any system observable $\hat O_{\text S} \in B(\mathcal H_{\text S})$, the time-dependent expectation value of $\hat O_{\text S}$ is given by
\begin{equation}\label{eqn:expectation_O}
    O_{\text S}(t) = \operatorname{tr}(\hat O_{\text S} \hat\rho_{\text S}(t)).
\end{equation}
Errors in $O_{\text S}(t)$ are controlled by the errors in $\hat\rho_{\text S}(t)$ and by the operator norm of $\hat O_{\text S}$. The following lemma (a direct consequence of \eqref{eqn:optrace_inequality} and the bound $|\operatorname{tr}(\hat O_1)|\leq\|\hat O_1\|_{\text{tr}}$ for trace-class $\hat O_1$) makes this precise.

\begin{lem}
    Let $\hat O_{\text S} \in B(\mathcal H_{\text S})$ and $O_{\text S}(t)$, $O_{\text S}'(t)$ be the expectation value of $\hat O_{\text S}$ corresponding to two different density operators $\hat\rho_{\text S}(t)$ and $\hat\rho_{\text S}'(t)$, respectively. Then we have
    \begin{equation}
        |O_{\text S}(t) - O_{\text S}'(t)| \leq \|\hat O_{\text S}\| \|\hat\rho_{\text S}(t) - \hat\rho_{\text S}'(t)\|_{\text{tr}}.
    \end{equation}
    \label{lem:rho_to_O}
\end{lem}

If there is no system-environment coupling, i.e., $\boldsymbol{L}_{\text{SE}} = 0$,  the total density operator evolves as $\mathrm e^{\boldsymbol{L}_0 t}\hat\rho(0)$, where $\mathrm e^{\boldsymbol{L}_0 t} = \mathrm e^{\boldsymbol{L}_{\text e} t}\otimes \mathrm e^{\boldsymbol{L}_{\text s} t}$. In general, $\hat\rho(t)$ has the following Dyson's series expansion for finite $t>0$:
\begin{equation}
    \begin{aligned}
     &   \hat\rho(t) = \mathrm e^{\boldsymbol{L}_0t}\hat\rho(0) + \int_0^t\mathrm dt_1 \mathrm e^{\boldsymbol{L}_0(t-t_1)}\boldsymbol{L}_{\text{SE}}\mathrm e^{\boldsymbol{L}_0t_1}\hat\rho(0) + \int_0^t\mathrm dt_1\int_0^{t_1}\mathrm dt_2 \mathrm e^{\boldsymbol{L}_0(t-t_1)}\boldsymbol{L}_{\text{SE}}\mathrm e^{\boldsymbol{L}_0(t_1-t_2)}\boldsymbol{L}_{\text{SE}}\mathrm e^{\boldsymbol{L}_0t_2}\hat\rho(0) + \cdots\\
    & =   \sum_{n=0}^{\infty}\iint\cdots\int_{0<t_n<\cdots<t_1<t} \mathrm e^{\boldsymbol{L}_0(t-t_1)}\boldsymbol{L}_{\text{SE}}\mathrm e^{\boldsymbol{L}_0(t_1-t_2)}\boldsymbol{L}_{\text{SE}}\cdots\mathrm e^{\boldsymbol{L}_0(t_{n-1}-t_n)}\boldsymbol{L}_{\text{SE}}\mathrm e^{\boldsymbol{L}_0t_n}\hat\rho(0)\mathrm dt_1\cdots\mathrm dt_n.
    \end{aligned}
    \label{eq:rho_perturbation_series}
   \end{equation}

Since $\boldsymbol{L}_{\text{SE}} = \sum_{\alpha=1}^N \boldsymbol{E}_{\alpha}\otimes\boldsymbol{S}_{\alpha}$ and $\hat\rho(0) = \hat \rho_{\text E}(0)\otimes \hat \rho_{\text S}(0)$, \cref{eq:rho_perturbation_series} is equivalent to the following:
\begin{equation}
    \begin{aligned}
        \hat\rho (t) =  \sum_{n=0}^{\infty}\sum_{\alpha_1,\cdots,\alpha_n=1}^N \iint\cdots&\int_{0<t_n<\cdots<t_1<t} \left(\mathrm e^{\boldsymbol{L}_{\text e}(t-t_1)}\boldsymbol{E}_{\alpha_1}\mathrm e^{\boldsymbol{L}_{\text e}(t_1-t_2)}\cdots \boldsymbol{E}_{\alpha_n}\mathrm e^{\boldsymbol{L}_{\text e}t_n}\right)\hat \rho_{\text E}(0) \\
        &\otimes \left(\mathrm e^{\boldsymbol{L}_{\text s}(t-t_1)}\boldsymbol{S}_{\alpha_1}\mathrm e^{\boldsymbol{L}_{\text s}(t_1-t_2)}\cdots \boldsymbol{S}_{\alpha_n}\mathrm e^{\boldsymbol{L}_{\text s}t_n}\right)\hat \rho_{\text S}(0)\mathrm dt_1\cdots\mathrm dt_n.
    \end{aligned}
\label{eq:rho_perturbation_series_SE}
\end{equation}

\cref{eq:rho_perturbation_series} could be formally verified by taking derivatives on both sides. We will later establish the absolute convergence of the above Dyson series for the specific cases under consideration. Specifically,  we verify that the bound given in \cref{eq:estimate} holds for the scenarios considered in this work. This bound, in turn, guarantees the absolute convergence of the series in \cref{eq:rho_perturbation_series_SE} and \cref{eq:expansion_longest}.

To take the partial trace with respect to the environment, it is natural to adopt the following trace-class assumption:

\begin{assumption}[Environment trace-class assumption]
   We assume that for all $n$  and for any $t\geq t_1\geq \cdots \geq t_n$, the following operator $\hat{\mathcal E}_{\alpha_1,\cdots,\alpha_n}(t,t_1,\cdots,t_n)$, known as the $n$-point correlator, is in the trace class:
\begin{equation}
    \hat{\mathcal E}_{\alpha_1,\cdots,\alpha_n}(t,t_1,\cdots,t_n)= \mathrm e^{\boldsymbol{L}_{\text e}(t-t_1)}\boldsymbol{E}_{\alpha_1}\mathrm e^{\boldsymbol{L}_{\text e}(t_1-t_2)}\cdots \boldsymbol{E}_{\alpha_n}\mathrm e^{\boldsymbol{L}_{\text e}t_n}\hat\rho_{\text E}(0) .
    \label{eq:env_correlator}
\end{equation}
\label{assume:trace}
\end{assumption}
As discussed in detail later, the trace-class assumption is satisfied in all cases studied in this paper, since the initial bath state $\hat\rho_{\text E}(0)$ is Gaussian and the system-environment coupling is linear. With the trace-class assumption \ref{assume:trace}, we can define the $n$-point environment superoperator correlation functions $C^{(n)}_{\alpha_1,\cdots,\alpha_n}(t_1,\cdots,t_n)$ (for $t\geq t_1\geq \cdots\geq t_n \geq 0$) as the trace of the above operators:
\begin{defn}[$n$-point environment superoperator correlation functions]
The $n$-point environment superoperator correlation functions $C^{(n)}_{\alpha_1,\cdots,\alpha_n}(t_1,\cdots,t_n)$ (for $t\geq t_1\geq \cdots\geq t_n\geq 0$) is defined as the trace of the above operator:
    \begin{equation}
    C^{(n)}_{\alpha_1,\cdots,\alpha_n}(t_1,\cdots,t_n) = \operatorname{tr}\left(\hat{\mathcal E}_{\alpha_1,\cdots,\alpha_n}(t,t_1,\cdots,t_n)\right).
    \label{eq:BCF_npoint}
\end{equation}
Note that $C^{(n)}_{\alpha_1,\cdots,\alpha_n}(t_1,\cdots,t_n)$ is independent of $t\in [t_1,\infty)$ since $\mathrm e^{\boldsymbol{L}_{\text e} t}$ is a trace preserving map for $t>0$. 
\end{defn}
\cref{assume:trace} directly indicates that the operator
$$\left(\mathrm e^{\boldsymbol{L}_{\text e}(t-t_1)}\boldsymbol{E}_{\alpha_1}\mathrm e^{\boldsymbol{L}_{\text e}(t_1-t_2)}\cdots \boldsymbol{E}_{\alpha_n}\mathrm e^{\boldsymbol{L}_{\text e}t_n}\right)\hat \rho_{\text E}(0)  \otimes \left(\mathrm e^{\boldsymbol{L}_{\text s}(t-t_1)}\boldsymbol{S}_{\alpha_1}\mathrm e^{\boldsymbol{L}_{\text s}(t_1-t_2)}\cdots \boldsymbol{S}_{\alpha_n}\mathrm e^{\boldsymbol{L}_{\text s}t_n}\right)\hat \rho_{\text S}(0).$$
is also trace-class for all $\alpha_1,\cdots,\alpha_n$ and $t \geq t_1\geq \cdots\geq t_n$, since $\boldsymbol{L}_{\text s}$ and $\boldsymbol{S}_\alpha$ are bounded in the induced trace norm. This leads to the following convergence criterion for the Dyson series in \cref{eq:rho_perturbation_series}:
\begin{lem}
\label{lem:dyson_conv}
Fix $t\ge 0$ and set
\begin{equation*}
    A_n(t) := \sup_{t\ge t_1\ge \cdots \ge t_n\ge 0} \bigl\|\hat{\mathcal E}_{\alpha_1,\ldots,\alpha_n}(t,t_1,\ldots,t_n) \bigr\|_{\text{tr}}, 
\end{equation*}
If
\begin{equation}
    \sum_{n=0}^\infty \frac{( N t)^n}{n!} \, A_n(t) < +\infty,
    \label{eq:bound}
\end{equation}
then the Dyson series in \cref{eq:rho_perturbation_series} (equivalently, \cref{eq:rho_perturbation_series_SE}) converges absolutely. Moreover, the summability condition \eqref{eq:bound} holds whenever there exists a constant $\widetilde A$ and  $\mathfrak{c}(t)>0$ such that 
\begin{equation}
    A_n(t) \le \sqrt{n!}\, (\mathfrak{c}(t))^{n}.
    \label{eq:estimate}
\end{equation}
\end{lem}
\begin{proof}
    Taking the trace norm of the right-hand side of \cref{eq:rho_perturbation_series_SE} and using $\|\boldsymbol{O}_{\text E}\otimes \boldsymbol{O}_{\text S}\|_{\text{tr}} = \|\boldsymbol{O}_{\text E}\|_{\text{tr}} \|\boldsymbol{O}_{\text S}\|_{\text{tr}}$ yields the first claim. For the second claim, observe that  $\lim_{n\rightarrow\infty}\frac{\sqrt{n} Nt\cdot \mathfrak{c}(t)}{n}= 0$. By the ratio test, the summation is therefore finite.
\end{proof}

If $\mathcal H_{\text E}$ is finite-dimensional, and thus $\boldsymbol{L}_{\text e}$ and $\boldsymbol{E}_\alpha$ are bounded, the estimate \cref{eq:estimate} holds naturally. This covers fermionic environments with a finite number of modes. We will show later that this estimate holds for systems of interest in this paper, especially when $\boldsymbol{L}_{\text E}$ and $\boldsymbol{E}_\alpha$ are unbounded, and the scaling in \cref{eq:estimate} arises naturally for Gaussian environments and linear system-environment coupling (see \cref{sec:estimate}). In particular, for bosonic Gaussian environments, the estimate in \cref{eq:estimate} is the direct result of the infrared regularity condition (see \cref{eq:infrared_regular}). (See \cref{eq:boson_growth_constant} for the detailed expression of the constant $\mathfrak c(t)$.) We note that \cref{eq:bound} guaranties that the $n$-point correlation functions $C^{(n)}_{\alpha_1,\cdots,\alpha_n}(t_1,\cdots,t_n)$ are bounded in the $L^1$ sense, namely $\int_{[0,T]^n}\left|C^{(n)}_{\alpha_1,\cdots,\alpha_n}(t_1,\cdots,t_n)\right|\mathrm dt_1\cdots\mathrm dt_n<\infty$ for all finite $T\geq 0$. One could consider more general correlators such as tempered Radon measures \cite{Trivedi2022, TrivediRudner2024}. Such scenarios, however, fall outside the scope of the systems of interest in this paper.

 \subsection{Tracing out the environment}
Now we are ready to take the partial trace over the environment. 
\begin{lem}
    For $\hat\rho(t)$, if the RHS of \cref{eq:rho_perturbation_series}(equivalently, \cref{eq:rho_perturbation_series_SE})  is absolutely convergent, then, with the environment trace-class assumption, we have
\begin{equation}
        \begin{aligned}
            \hat\rho_{\text S}(t) =  \sum_{n=0}^{\infty}\sum_{\alpha_1,\cdots,\alpha_n=1}^N \iint\cdots&\int_{0<t_n<\cdots<t_1<t} \mathrm dt_1\cdots\mathrm d t_n \\
        &C^{(n)}_{\alpha_1,\cdots,\alpha_n}(t_1,\cdots,t_n) \mathrm e^{\boldsymbol{L}_{\text s}(t-t_1)}\boldsymbol{S}_{\alpha_1}\mathrm e^{\boldsymbol{L}_{\text s}(t_1-t_2)}\cdots \boldsymbol{S}_{\alpha_n}\mathrm e^{\boldsymbol{L}_{\text s}t_n}\hat\rho_{\text S}(0).
        \end{aligned}
        \label{eq:rhoS_Dyson}
    \end{equation}
    \label{lem:reduced_dyson}
\end{lem}  
\begin{proof}
    The key here is that we have exchanged the partial trace first with the summation and then with the integration. The partial trace, as a linear mapping  defined on trace-class operators, is contractive in the trace norm; namely, for any trace-class $\hat A$ in $B_1(\mathcal H)$, $\|\operatorname{tr}_{\text E}\hat A\|_{\text{tr}}\leq \|A\|_{\text{tr}}$. The absolute convergence of the infinite series in \cref{eq:rho_perturbation_series_SE} implies the absolute convergence of the RHS of \cref{eq:rhoS_Dyson}.
    The  exchange of the partial trace, summation, and integration is thus justified by the linearity of the partial trace and the dominated convergence theorem.
\end{proof}

Let us introduce the notation $\boldsymbol S_{\alpha}(t)$, $\boldsymbol E_{\alpha}(t)$, and $ \boldsymbol{L}_{\text{SE}}(t)$:
\begin{equation}
    \boldsymbol S_{\alpha}(t) = \mathrm e^{-\boldsymbol L_{\text s}t}\boldsymbol S_{\alpha}\mathrm e^{\boldsymbol L_{\text s}t},\quad \boldsymbol E_{\alpha}(t) = \mathrm e^{-\boldsymbol L_{\text e}t}\boldsymbol E_{\alpha}\mathrm e^{\boldsymbol L_{\text e}t}, \quad 
    \boldsymbol{L}_{\text{SE}}(t)=\mathrm e^{-\boldsymbol L_{0}t}\boldsymbol L_{\text{SE}}\mathrm e^{\boldsymbol L_0t}.
    \label{eq:superO_int}
\end{equation}
We emphasize that this is just a shorthand notation for convenience since $\mathrm e^{\boldsymbol{L}_{\text s}t}$, $\mathrm e^{\boldsymbol{L}_{\text e}t}$, and $\mathrm e^{\boldsymbol{L}_{0}t}$ are not necessarily well-defined for $t<0$. Then  we could rewrite \cref{eq:BCF_npoint} and \cref{eq:rhoS_Dyson} as:
\begin{equation}
    C_{\alpha_1,\cdots,\alpha_n}^{(n)}(t_1,\cdots,t_n) = 
\operatorname{tr}\left(\mathrm e^{\boldsymbol L_{\text e}t}\boldsymbol E_{\alpha_1}(t_1)\cdots \boldsymbol E_{\alpha_n}(t_n) \hat\rho_{\text E}(0) 
    \right)
    \label{eq:BCF_npoint_formal}
\end{equation}
\begin{equation}
\begin{aligned}
      \hat\rho_{\text S}(t) = \mathrm e^{\boldsymbol L_{\text s}t} \sum_{n=0}^{\infty}\sum_{\alpha_1,\cdots,\alpha_n=1}^N \iint\cdots&\int_{0<t_n<\cdots<t_1<t}\mathrm dt_1\cdots\mathrm d t_n \\&C^{(n)}_{\alpha_1,\cdots,\alpha_n}(t_1,\cdots,t_n)  \boldsymbol S_{\alpha_1}(t_1)\cdots \boldsymbol S_{\alpha_n}(t_n)\hat\rho_{\text S}(0).
    \label{eq:rhoS_Dyson_S_alpha}  
\end{aligned}
\end{equation}
We emphasize that although $\boldsymbol E_{\alpha}(t)$ and $\boldsymbol S_{\alpha}(t)$ are formal expressions, \cref{eq:BCF_npoint_formal} and \cref{eq:rhoS_Dyson_S_alpha} are well-defined since $\mathrm e^{\boldsymbol{L}_{\text e}t}$ (resp. $\mathrm e^{\boldsymbol{L}_{\text s}t}$) appears in front of the sum. For unordered $t_1,\cdots,t_n$, let us define $C^{(n)}_{\alpha_1,\cdots,\alpha_n}(t_1,\cdots,t_n) $ and the time-ordering operation $\mathbf T^{\pm}$:
\begin{equation}\begin{aligned}
    C^{(n)}_{\alpha_1,\cdots,\alpha_n}(t_1,\cdots,t_n) &= (\pm)^\sigma C^{(n)}_{\alpha_{\sigma(1)},\cdots,\alpha_{\sigma(n)}}(t_{\sigma(1)},\cdots,t_{\sigma(n)}),\\
    \mathbf T^{\pm}\left( \boldsymbol S_{\alpha_1}(t_1)\cdots \boldsymbol S_{\alpha_n}(t_n)\right) &= (\pm)^\sigma \boldsymbol S_{\alpha_{\sigma(1)}}(t_{\sigma(1)})\cdots \boldsymbol S_{\alpha_{\sigma(n)}}(t_{\sigma(n)}),
\end{aligned}
\label{eq:permutations}
    \end{equation}
where $\sigma\in S_n$ is the permutation of $1,\cdots,n$ such that $t_{\sigma(1)}\geq \cdots\geq t_{\sigma(n)}$, and the plus and the minus sign correspond to bosonic and fermionic environments, respectively. Then we can rewrite \cref{eq:rhoS_Dyson_S_alpha} as:
\begin{equation}
    \hat\rho_{\text S}(t) = \mathrm e^{\boldsymbol L_{\text s}t} \sum_{n=0}^{\infty}\sum_{\alpha_1,\cdots,\alpha_n=1}^N \frac{1}{n!}\int_0^t\cdots\int_0^t C^{(n)}_{\alpha_1,\cdots,\alpha_n}(t_1,\cdots,t_n)  \mathbf T^{\pm}\left(\boldsymbol S_{\alpha_1}(t_1)\cdots \boldsymbol S_{\alpha_n}(t_n)\right)\mathrm dt_1\cdots\mathrm d t_n\hat\rho_{\text S}(0).
    \label{eq:rhoS_Dyson_n!}
\end{equation}
One can see that, operators under the time-ordering operation (anti-)commute. Let us also introduce the time-ordering operation $\boldsymbol{\mathcal T}$ without the sign change, which will be used in the proof of \cref{thm:main_error_bound}, defined as:
\begin{equation}
    \mathcal{\boldsymbol{\mathcal T}}\left(\boldsymbol S_{\alpha_1}(t_1)\cdots \boldsymbol S_{\alpha_n}(t_n)\right) = \boldsymbol S_{\alpha_{\sigma(1)}}(t_{\sigma(1)})\cdots \boldsymbol S_{\alpha_{\sigma(n)}}(t_{\sigma(n)}),
\end{equation}
where $\sigma\in S_n$ is the same permutation as above.
\begin{rem}
    Note that $C_{\alpha_1,\cdots,\alpha_n}^{(n)}(t_1,\cdots, t_n)$ is defined using environment superoperators $\boldsymbol E_{\alpha_i}$ ($i=1,\cdots,n$), while the commonly-used $n$-point BCF is defined using operators \cite{breuer2007open}. Nevertheless, we will refer to $C_{\alpha_1,\cdots,\alpha_n}^{(n)}(t_1,\cdots, t_n)$ as the $n$-point BCF in this section when there are no further confusions.
\end{rem}
\begin{rem}
    For time-ordering operation defined in \cref{eq:permutations}, later we will use the following expression
    $$
    \mathbf{T}^{ \pm}\left(\prod_{i=1}^{n} \boldsymbol{S}_{\alpha_i}\left(t_i\right) \prod_{j=1}^{k} \boldsymbol{S}_{\beta_j}\left(s_j\right)\right),
    $$
    in which the time-ordering is understood to be performing on $t_1,\cdots,t_{n},s_1,\cdots, s_{k}$. In other words, let $w_i = t_i$, $\gamma_i=\alpha_i$ ($i=1,\cdots,n$) and $w_{2n+j}=s_j$, $\gamma_{2n+j} = \beta_j$ ($j=1,\cdots, k$), and let $\sigma\in S_{n+k}$ be the permutation such that $w_1\geq \cdots\geq w_{n+k}$. Then we have $\mathbf{T}^{ \pm}\left(\prod_{i=1}^{n} \boldsymbol{S}_{\alpha_i}\left(t_i\right) \prod_{j=1}^{k} \boldsymbol{S}_{\beta_j}\left(s_j\right)\right)=(\pm)^\sigma \prod_{i=1}^{n+k}\boldsymbol{S}_{\gamma_i}(w_i)$.
    \label{rem:time_ordering}
\end{rem}
In this paper, we consider environment $\boldsymbol L_{\text e}$ and system-environment coupling $\boldsymbol L_{\text{SE}}$ such  that  the $n$-point BCFs satisfy the following Isserlis-Wick's condition:
\begin{defn}[Isserlis-Wick's conditions for BCFs]
    For $t_1,\cdots, t_n>0$, the $n$-point BCFs $C^{(n)}_{\alpha_1,\cdots,\alpha_n}(t_1,\cdots,t_n)$ satisfy the Isserlis-Wick's condition: for $n$ odd, $C^{(n)}_{\alpha_1,\cdots,\alpha_n}(t_1,\cdots,t_n)= 0$, while for $n$ even, i.e., $n=2m$,
    we have 
    \begin{equation}
        C^{(2m)}_{\alpha_1,\cdots,\alpha_{2m}}(t_1,\cdots,t_{2m}) = \sum_{\sigma \in \Pi_{2 m}}( \pm)^\sigma \prod_{i=1}^m C_{\alpha_{\sigma(2 i-1)}, \alpha_{\sigma(2 i)}}\left(t_{\sigma(2 i-1)}, t_{\sigma(2 i)}\right),
    \end{equation}
    where $\Pi_{2m}$ is the set of pairings of $\{1,\cdots,2m\}$, i.e.,
    $$
    \Pi_{2m} = \left\{\sigma \in S_{2m} \mid \sigma(2i-1)<\sigma(2i) \text{ for } i \in \{1,\cdots,m\}, \text{ and } \sigma(1)<\sigma(3)<\cdots<\sigma(2m-1)\right\}.
    $$
    \label{defn:Wick_condition} 
\end{defn}
Such conditions are known to be satisfied for linearly coupled Gaussian environments in both unitary and non-unitary dynamics. We refer to \cref{sec:applications} for examples.
With the Isserlis-Wick's conditions, we have
\begin{thm}[Dyson's series for $\hat\rho_{\text S}(t)$ after resummation]
    With the Isserlis-Wick's conditions in \cref{defn:Wick_condition}, we have
    \begin{equation}
        \begin{aligned}
            \hat\rho_{\text S}(t) =& \sum_{m=0}^{\infty} \frac{1}{2^m m!}\int_{0}^{t}\int_0 ^{t}\cdots\int_0 ^{t}\mathrm d t_{2m}\cdots \mathrm d t_1 \mathrm e^{\boldsymbol{L}_{\text s}t}\\
            &\sum_{\alpha_1, \cdots, \alpha_{2 m}=1}^N \prod_{i=1}^m C_{\alpha_{2 i-1}, \alpha_{2 i}}\left(t_{2 i-1}, t_{2 i}\right) \mathbf{T}^{\pm}\left(\boldsymbol{S}_{\alpha_1}\left(t_1\right) \cdots \boldsymbol{S}_{\alpha_{2 m}}\left(t_{2 m}\right)\right) \hat{\rho}_{\mathrm{S}}(0).      
        \end{aligned}
        \label{eq:rhoS_Dyson_short}
           \end{equation}
           \label{thm:rhoS_Dyson_short}
\end{thm}
\begin{proof}
    With the definition of $\mathbf T^{\pm}$, we have 
    $$
    \mathbf T^{\pm}\left( \boldsymbol S_{\alpha_1}(t_1)\cdots \boldsymbol S_{\alpha_n}(t_n)\right) = (\pm)^\sigma \mathbf T^{\pm}\left( \boldsymbol S_{\alpha_{\sigma(1)}}(t_{\sigma(1)})\cdots \boldsymbol S_{\alpha_{\sigma(n)}}(t_{\sigma(n)})\right).
    $$
    Using \cref{defn:Wick_condition}, 
    we can rewrite \cref{eq:rhoS_Dyson_n!} as:
    \begin{equation}\begin{aligned}
        & \begin{aligned}   \hat{{\rho}}_{\text S}(t)
       = & \sum_{m=0}^{\infty} \frac{1}{(2m)!}\int_{0}^{t}\int_0 ^{t}\cdots\int_0 ^{t}\mathrm d t_{2m}\cdots \mathrm d t_1 \mathrm e^{\boldsymbol{L}_{\text s}t} \\ 
        &\sum_{\alpha_1,\cdots,\alpha_{2m}=1}^N
     \sum_{\sigma\in \Pi_{2m}} (\pm)^{\sigma}\prod_{i=1}^{m} C_{\alpha_{\sigma(2i-1)},\alpha_{\sigma(2i)}}(t_{\sigma(2i-1)}, t_{\sigma(2i)})\mathbf{T}^{\pm} \left(\boldsymbol S_{\alpha_{1}}(t_{1})\cdots \boldsymbol S_{\alpha_{{2m}}}(t_{{2m}})\right)\hat\rho_{\text{S}}(0)
     \end{aligned}\\
    & \begin{aligned}   
        =  \sum_{m=0}^{\infty}&\sum_{\sigma\in \Pi_{2m}} \frac{1}{(2m)!}\int_{0}^{t}\int_0 ^{t}\cdots\int_0 ^{t}\mathrm d t_{2m}\cdots \mathrm d t_1 \mathrm e^{\boldsymbol{L}_{\text s}t} \\ 
         &\sum_{\alpha_{\sigma(1)},\cdots,\alpha_{\sigma(2m)}=1}^N
       \prod_{i=1}^{m} C_{\alpha_{\sigma(2i-1)},\alpha_{\sigma(2i)}}(t_{\sigma(2i-1)}, t_{\sigma(2i)})\mathbf{T}^{\pm} \left(\boldsymbol S_{\alpha_{\sigma(1)}}(t_{\sigma(1)})\cdots \boldsymbol S_{\alpha_{{\sigma(2m)}}}(t_{{\sigma(2m)}})\right)\hat\rho_{\text{S}}(0).
      \end{aligned}
        \end{aligned} 
        \label{eq:rho_S_resum}
        \end{equation}  
    By change of variables, we have 
    \begin{equation}
        \begin{aligned}   
            \hat\rho_{\text S}(t) = \sum_{m=0}^{\infty} & \sum_{\sigma\in \Pi_{2m}}\frac{1}{(2m)!}\int_{0}^{t}\int_0 ^{t}\cdots\int_0 ^{t}\mathrm d t_{2m}\cdots \mathrm d t_1 \mathrm e^{\boldsymbol{L}_{\text s}t} \\ 
         &\sum_{\alpha_{1},\cdots,\alpha_{2m}=1}^N 
       \prod_{i=1}^{m} C_{\alpha_{2i-1},\alpha_{2i}}(t_{2i-1}, t_{2i})\mathbf{T}^{\pm} \left(\boldsymbol S_{\alpha_{1}}(t_{1})\cdots \boldsymbol S_{\alpha_{2m}}(t_{{2m}})\right)\hat\rho_{\text{S}}(0).
      \end{aligned} 
    \end{equation}
Since $|\Pi_{2m}| = (2m-1)!!$, we obtain \cref{eq:rhoS_Dyson_short}.
\end{proof}

\begin{rem}
    A formal but convenient expression for \cref{thm:rhoS_Dyson_short}, which is known in physics literature (for example,  \cite{FeynmanVernon1963, BreuerMaPetruccione2004}), is:
    \begin{equation}
    \begin{aligned}
        \hat\rho_{\text S}(t) &=\mathrm e^{\boldsymbol L_{\text s}t} \mathbf T^{\pm}\left(
            \mathrm{exp}\left( \int_0^t\int_0^{t_1} \sum_{\alpha_1,\alpha_2=1}^N C_{\alpha_1,\alpha_2}(t_1,t_2)\boldsymbol S_{\alpha_1}(t_1)\boldsymbol S_{\alpha_2}(t_2)\mathrm d t_2\mathrm d t_1\right)\right)\hat\rho_{\text S}(0) ,\end{aligned}\label{eq:rhoS_Dyson_short_formal}   
    \end{equation}
    where the exponential is understood as the formal Taylor's expansion, and the time-ordering operation $\mathbf T^{\pm}$ acting on this exponential is understood to be acting on each term of the summation. In other words, we have:
    \begin{equation}
        \begin{aligned}
      &  \mathrm{exp}\left( \int_0^t\int_0^{t_1}\sum_{\alpha_1,\alpha_2=1}^N C_{\alpha_1,\alpha_2}(t_1,t_2)\boldsymbol S_{\alpha_1}(t_1)\boldsymbol S_{\alpha_2}(t_2)\mathrm d t_2\mathrm d t_1\right)   
      \\  : = &\sum_{n=0}^\infty \frac{1}{n!}\left( \int_0^t\int_0^{t_1}\sum_{\alpha_1,\alpha_2=1}^N C_{\alpha_1,\alpha_2}(t_1,t_2)\boldsymbol S_{\alpha_1}(t_1)\boldsymbol S_{\alpha_2}(t_2)\mathrm d t_2\mathrm d t_1\right)^n\\
       =& \sum_{n=0}^\infty \frac{1}{ n!}\int \cdots \int_{\Omega_n(t)}\sum_{\alpha_1,\cdots,\alpha_{2n}=1}^N \left(\prod_{i=1}^n C_{\alpha_{2i-1},\alpha_{2i}}(t_{2i-1},t_{2i})\right)\prod_{j=1}^{2n}\boldsymbol S_{\alpha_j}(t_j)\mathrm d t_1\cdots \mathrm d t_{2n}.
        \end{aligned}
    \end{equation}
    where $\Omega_n(t)$ is the following region:
    $$
    \Omega_n(t) = \{(t_1,\cdots,t_{2n})| 0<t_{2i}<t_{2i-1}<t, \text{ for }i=1,\cdots,n \}.
    $$
    Then we have,
    \begin{equation}
        \begin{aligned}
            &\mathrm e^{\boldsymbol L_{\text s}t}\mathbf T^{\pm}\left(\mathrm{exp}\left( \int_0^t\int_0^{t_1} \sum_{\alpha_1,\alpha_2=1}^N C_{\alpha_1,\alpha_2}(t_1,t_2)\boldsymbol S_{\alpha_1}(t_1)\boldsymbol S_{\alpha_2}(t_2)\mathrm d t_2\mathrm d t_1\right)\right) 
          \\  := & \mathrm e^{\boldsymbol L_{\text s}t}\sum_{n=0}^\infty \frac{1}{  n!}\int\cdots \int_{\Omega_n(t)}\sum_{\alpha_1,\cdots,\alpha_{2n}=1}^N \left(\prod_{i=1}^n C_{\alpha_{2i-1},\alpha_{2i}}(t_{2i-1},t_{2i})\right)\mathbf T^{\pm} \left(\prod_{j=1}^{2n}\boldsymbol S_{\alpha_j}(t_j)\right)\mathrm d t_1\cdots \mathrm d t_{2n}.
        \end{aligned}
    \label{eq:T_exp_defn}
    \end{equation}
    We emphasize that the left-hand side of the above equation should only be understood as the shorthand notation for the right-hand side, which is mathematically well-defined.
   Because of \cref{eq:permutations}, 
   the right-hand side of \cref{eq:T_exp_defn} is equal to the right-hand side of \cref{eq:rhoS_Dyson_short}.
    Similarly, another shorthand notation  for \cref{eq:rhoS_Dyson_short} is
   the following:
   \begin{equation}
    \hat\rho_{\text S}(t)= \mathrm e^{\boldsymbol L_{\text s}t} \mathbf T^{\pm}\left(
            \mathrm{exp}\left( \frac 1 2\int_0^t\int_0^{t} \sum_{\alpha_1,\alpha_2=1}^N C_{\alpha_1,\alpha_2}(t_1,t_2)\boldsymbol S_{\alpha_1}(t_1)\boldsymbol S_{\alpha_2}(t_2)\mathrm d t_2\mathrm d t_1\right)\right)\hat\rho_{\text S}(0), 
   \end{equation}
   where $\mathrm e^{\boldsymbol L_{\text s}t}\mathbf T^{\pm}\left(
            \mathrm{exp}\left( \frac 1 2\int_0^t\int_0^{t} \sum_{\alpha_1,\alpha_2=1}^N C_{\alpha_1,\alpha_2}(t_1,t_2)\boldsymbol S_{\alpha_1}(t_1)\boldsymbol S_{\alpha_2}(t_2)\mathrm d t_2\mathrm d t_1\right)\right)$ is understood as the shorthand notation for the following:
            $$
            \mathrm e^{\boldsymbol L_{\text s}t}\sum_{n=0}^\infty \frac{1}{ 2^n n!}\int_0^t\cdots \int_0^t\sum_{\alpha_1,\cdots,\alpha_{2n}=1}^N \left(\prod_{i=1}^n C_{\alpha_{2i-1},\alpha_{2i}}(t_{2i-1},t_{2i})\right)\mathbf T^{\pm} \left(\prod_{j=1}^{2n}\boldsymbol S_{\alpha_j}(t_j)\right)\mathrm d t_1\cdots \mathrm d t_{2n}.
            $$
   \label{rmk:formal_exponential_formula}
\end{rem} 

\subsection{Gr\"{o}nwall-type error bound}
\label{sec:gronwall}
From \cref{thm:rhoS_Dyson_short}, we observe that the reduced system dynamics $\hat\rho_{\text S}(t)$ is fully determined by the system Liouvillian $\boldsymbol{L}_{\text s}$ and two-point BCFs $C_{\alpha_1,\alpha_2}(t_1,t_2)$. We assume that  $(\mathrm e^{\boldsymbol{L}_{\text s}t})_{t\geq 0}$ is a contractive semigroup, which is a natural condition since $\mathrm e^{\boldsymbol{L}_{\text s}t}$ describes the (sub)system dynamics governed by the von-Neumann or Lindblad equations (see \cref{sec:applications}). 
In what follows, we prove the Gr\"onwall-type error bound for $\hat\rho_{\text S}(t)$ under perturbations of $C_{\alpha_1,\alpha_2}(t_1,t_2)$.

In the case when $C_{\alpha_1,\alpha_2}(t_1,t_2)$ depends only on the time difference $(t_1-t_2)$ (see \cref{rmk:single_variable_Corr}), we introduce the notation $\mathcal C_{\alpha_1,\alpha_2}(t_1-t_2)$ such that $C_{\alpha_1,\alpha_2}(t_1,t_2) = \mathcal C_{\alpha_1,\alpha_2}(t_1-t_2)$. To differ perturbations in the two-time BCFs $C_{\alpha_1,\alpha_2}(t_1,t_2)$ and the one-time function $\mathcal C_{\alpha_1,\alpha_2}(t_1-t_2)$, we denote the norm of the former by $\delta$ for $L^\infty$ and $\delta_1$ for $L^1$, and the latter by $\epsilon$ and $\epsilon_1$. The precise definitions of $\delta, \delta_1, \epsilon, \epsilon_1$
are given in the following theorem.

\begin{thm}[Gr\"{o}nwall-type error bound for reduced system density operators]
For the Liouville dynamics \cref{eq:L_decomp}, where the environment and the system-environment coupling satisfies  Isserlis-Wick's condition \cref{defn:Wick_condition},  consider two open quantum systems: one characterized by the system Liouvillian $\boldsymbol{L}_{\text s}$ and  two-point BCFs $C_{\alpha_1,\alpha_2}(t_1,t_2)$ ($\alpha_1,\alpha_2=1,\cdots,N$) and the other by the same $\boldsymbol{L}_{\text s}$ and perturbed BCFs $C'_{\alpha_1,\alpha_2}(t_1,t_2)$. We assume that $(\mathrm e^{\boldsymbol{L}_{\text s}t})_{t\geq 0}$ forms a contractive semigroup.
Let $\hat \rho_S(t)$ and $\hat \rho_S'(t)$ be the corresponding reduced system density operators given by \cref{eq:rhoS_Dyson_short}.
Then, given maximal evolution time $T$, the difference  $\hat \rho_S(t)$ and $\hat \rho_S'(t)$  can be bounded by for all $t\in[0,T]$,
\begin{equation}
        \|\hat\rho_{\text S}(t) - \hat\rho_{\text S}'(t)\|_{\text{tr}} \leq \frac{\delta t^2}{2}\mathrm e^{M t^2/2},
    \end{equation}
where $M$ and $\delta$ are given by:
    \begin{equation}
        M =  \max\left\{\|C\|_{N, \infty },\|C'\|_{N, \infty }\right\},\quad \delta =  \left\| C - C'\right\|_{N, \infty }.
        \label{eq:def_eps}
    \end{equation}
    Here for a $N\times N$ matrix-valued function $C$,
    we define the norm $\|C\|_{N,\infty} = \sum_{\alpha,\alpha'=1}^N\|C_{\alpha,\alpha'}\|_{L^{\infty}([0,T]^2)}$.
    Similarly, we also have an error estimates in terms of the $L^1$ norm of the error in the BCF:
     \begin{equation}
        \|\hat\rho_{\text S}(t) - \hat\rho_{\text S}'(t)\|_{\text{tr}} \leq \frac{\delta_1 t}{2}\mathrm e^{M_1 t/2}.
    \end{equation}
    where $M_1$ and $\delta_1$ are defined as
      \begin{equation}
        M_1 =  \sup_{t_2\in [0,T]}\max\left\{\|C(\cdot,t_2)\|_{N, 1 },\|C'(\cdot,t_2)\|_{N, 1}\right\},\quad \delta =  \left\| C(\cdot,t_2) - C'(\cdot,t_2)\right\|_{N,1 },
         \label{eq:def_eps1}
    \end{equation}
    with $\|C(\cdot,t_2)\|_{N,1}$ defined as $\|C(\cdot,t_2)\|_{N,1} =\sum_{\alpha,\alpha'=1}^N\|C_{\alpha,\alpha'}(\cdot,t_2)\|_{L^{1}([0,T])} $.
   
    \label{thm:Gronwall_error_bound}
    \end{thm}
\begin{proof}
    Following from \cref{thm:rhoS_Dyson_short}, we have
    \begin{equation} 
        \begin{aligned}
           & \| \hat\rho_{\text S}(t) - \hat\rho_{\text S}'(t)\|_{\text{tr}} \leq 
        \sum_{n=0}^{\infty}\sum_{\alpha_1, \cdots, \alpha_{2 n}=1}^N \frac{1}{2^n n!}\int_{0}^{t}\int_0 ^{t}\cdots\int_0 ^{t}\mathrm d t_{2n}\cdots \mathrm d t_1  \\
          &   \left|\prod_{i=1}^n C_{\alpha_{2 i-1}, \alpha_{2 i}}\left(t_{2 i-1}, t_{2 i}\right) - \prod_{i=1}^nC'_{\alpha_{2 i-1}, \alpha_{2 i}}\left(t_{2 i-1}, t_{2 i}\right)\right| \|\mathrm e^{\boldsymbol{L}_{\text s}t}\mathbf{T}^{\pm}\left(\boldsymbol S_{\alpha_1}(t_1)\cdots \boldsymbol S_{\alpha_{2n}}(t_{2n})\right)\hat\rho_{\text S}(0)\|_{\text{tr}}.
        \end{aligned}
    \end{equation}
Let $\sigma$ be the permutation of $1,\cdots,2n$ such that $t_{\sigma(1)}\geq \cdots\geq t_{\sigma(2n)}$. Recall that $\mathrm e^{\boldsymbol{L}_{\text s}t}$ is a contraction map and $\|\boldsymbol{S}_\alpha\|_{\text{tr}}=1$, thus
$$
\begin{aligned}
    &\| \mathrm e^{\boldsymbol{L}_{\text s}t}\mathbf{T}^{\pm}\left(\boldsymbol S_{\alpha_1}(t_1)\cdots \boldsymbol S_{\alpha_{2n}}(t_{2n})\right)\|_{\text{tr}} = \|\mathrm e^{\boldsymbol{L}_{\text s}(t-t_{\sigma(1)})} \boldsymbol{S}_{\alpha_{\sigma(1)}}\mathrm e^{\boldsymbol{L}_{\text s}(t_{\sigma(1)}-t_{\sigma(2)})}\cdots\boldsymbol{S}_{\alpha_{\sigma(2n)}}\mathrm e^{\boldsymbol{L}_{\text s}t_{\sigma(2n)}}\|_{\text{tr}}\\
&\leq \|\mathrm e^{\boldsymbol{L}_{\text s}(t-t_{\sigma(1)})}\|_{\text{tr}}\|\boldsymbol{S}_{\alpha_{\sigma(1)}}\|_{\text{tr}}\|\mathrm e^{\boldsymbol{L}_{\text s}(t_{\sigma(1)}-t_{\sigma(2)})}\|_{\text{tr}}\cdots\|\boldsymbol{S}_{\alpha_{\sigma(2n)}}\|_{\text{tr}}\|\mathrm e^{\boldsymbol{L}_{\text s}t_{\sigma(2n)}}\|_{\text{tr}}\leq 1.
\end{aligned}
$$
Since $\|\hat\rho_{\text S}(0)\|_{\text{tr}}=1$, we have
$$
\begin{aligned}
\| \hat\rho_{\text S}(t) - \hat\rho_{\text S}'(t)\|_{\text{tr}} \leq
\sum_{n=0}^{\infty} \sum_{\alpha_1, \cdots, \alpha_{2 n}=1}^N \frac{1}{2^n n!}\int_{0}^{t}&\int_0 ^{t}\cdots\int_0 ^{t}\mathrm d t_{2n}\cdots \mathrm d t_1  \\&
\left|\prod_{i=1}^n C_{\alpha_{2 i-1}, \alpha_{2 i}}\left(t_{2 i-1}, t_{2 i}\right) - \prod_{i=1}^nC'_{\alpha_{2 i-1}, \alpha_{2 i}}\left(t_{2 i-1}, t_{2 i}\right)\right|.
\end{aligned}
$$
Let $\Delta C_{\alpha,\alpha'} = C_{\alpha,\alpha'}' - C_{\alpha,\alpha'}$.
Note that 
$$
\begin{aligned}
    &\prod_{i=1}^n C_{\alpha_{2 i-1}, \alpha_{2 i}}\left(t_{2 i-1}, t_{2 i}\right) - \prod_{i=1}^nC'_{\alpha_{2 i-1}, \alpha_{2 i}}\left(t_{2 i-1}, t_{2 i}\right)
    \\=&-\sum_{i=1}^n \left(\prod_{j=1}^{i-1}C_{\alpha_{2j-1},\alpha_{2j}}(t_{2j-1},t_{2j})\right)  \Delta C_{\alpha_{2i-1},\alpha_{2i}}(t_{2i-1},t_{2i})\left(\prod_{j=i+1}^{n}C'_{\alpha_{2j-1},\alpha_{2j}}(t_{2j-1},t_{2j})\right)
\end{aligned}
$$
Taking absolute value, and summing over $\alpha_1,\cdots,\alpha_{2N}$, with 
$\left|\sum_{\alpha_{2i-1},\alpha_{2i}=1}^N C_{\alpha_{2i-1},\alpha_{2i}}(t_{2i-1},t_{2i}) \right|  \leq M$  and  $\left|\sum_{\alpha_{2i-1},\alpha_{2i}=1}^N\Delta C_{\alpha_{2i-1},\alpha_{2i}} \right| \leq \delta$,
  we obtain $$\sum_{\alpha_1,\cdots,\alpha_{2n}=1}^N \left|\prod_{i=1}^n C_{\alpha_{2 i-1}, \alpha_{2 i}}\left(t_{2 i-1}, t_{2 i}\right) - \prod_{i=1}^nC'_{\alpha_{2 i-1}, \alpha_{2 i}}\left(t_{2 i-1}, t_{2 i}\right)\right|\leq n M^{n-1}\delta.$$
Then we have: 
$$\|
\hat\rho_{\text S}(t) - \hat\rho_{\text S}'(t)\|_{\text{tr}}\leq \sum_{n=1}^{\infty} \frac{1}{2^n n!}t^{2n}nM^{n-1}\delta =\frac{\delta t^2}{2} \sum_{n=1}^{\infty} \frac{t^{2n-2}M^{n-1}}{2^{n-1}(n-1)!}= 
\frac{\delta t^2}{2}\mathrm e^{M t^2/2}.
$$ 
If instead using the $L_1$ norm of the BCF, we have 
$$
\begin{aligned}
\int_0^t\cdots\int_0^t\sum_{\alpha_1,\cdots,\alpha_{2n}=1}^N\Bigg|\prod_{i=1}^n C_{\alpha_{2 i-1}, \alpha_{2 i}}(t_{2 i-1}, t_{2 i})& - \prod_{i=1}^nC'_{\alpha_{2 i-1}, \alpha_{2 i}}(t_{2 i-1}, t_{2 i})\Bigg|\mathrm dt_1\cdots\mathrm dt_{2n}\\
&\leq n M_1^{n-1}\delta_1\int_0^t\cdots\int_0^t\mathrm dt_2\cdots\mathrm dt_{2n}=t^{n}nM_1^{n-1}\delta_1,
\end{aligned}
$$ and therefore,
$$
\| \hat\rho_{\text S}(t) - \hat\rho_{\text S}'(t)\|_{\text{tr}}\leq \sum_{n=1}^{\infty} \frac{1}{2^n n!}t^{n}nM_1^{n-1}\delta_1 = \frac{\delta_1 t}{2}\sum_{n=1}^{\infty} \frac{t^{n-1}M_1^{n-1}}{2^{n-1}(n-1)!} = \frac{\delta_1 t}{2}\mathrm e^{M_1 t/2}.
$$
\end{proof}
\begin{rem}
    Recall that for $t_1\geq t_2$, $C_{\alpha_1,\alpha_2}(t_1,t_2)$ is defined as 
    $$
        C_{\alpha_1,\alpha_2}(t_1,t_2) = \operatorname{tr}\left({\mathrm e^{\boldsymbol L_{\text e}(t-t_1)}\boldsymbol E_{\alpha_1}\mathrm e^{\boldsymbol L_{\text e}(t_1-t_2)}\boldsymbol E_{\alpha_2}\mathrm e^{\boldsymbol L_{\text e}t_2}\hat\rho_{\text E}(0)}\right), \quad t_1\geq t_2.
   $$
    If $\hat\rho_{\text E}(0)$ is a stationary state of $\mathrm e^{\boldsymbol L_{\text e}t}$, (i.e., $\boldsymbol L_{\text e}\hat\rho_{\text E}(0) = 0$), then $\mathrm e^{\boldsymbol L_{\text e}t }\hat\rho_{\text E}(0)=\hat\rho_{\text E}(0)$ (for $t\geq 0$). With the fact that $\mathrm e^{\boldsymbol L_{\text e} t}$ is a trace-preserving map, we have 
     \begin{equation}
        C_{\alpha_1,\alpha_2}(t_1,t_2) = \operatorname{tr}\left({\boldsymbol E_{\alpha_1}\mathrm e^{\boldsymbol L_{\text e}(t_1-t_2)}\boldsymbol E_{\alpha_2}\hat\rho_{\text E}(0)}\right), \quad  t_1\geq t_2.
      \label{eq:two_point_BCF}
    \end{equation}
    In other words, $C_{\alpha_1,\alpha_2}(t_1,t_2)$ relies only on the time difference $(t_1-t_2)$. In such cases, let $\mathcal C_{\alpha_1,\alpha_2}(t_1-t_2) = C_{\alpha_1,\alpha_2}(t_1,t_2)$, and let us define $\epsilon$, $\mathcal M$, $\epsilon_1$, $\mathcal M_1$ as: 
    \begin{equation}
    \begin{aligned}
    \mathcal M = \max\left\{\sum_{\alpha,\alpha'=1}^N \|\mathcal C_{\alpha,\alpha'}\|_{L^{\infty}([0,T])}, \sum_{\alpha,\alpha'=1}^N \|\mathcal C'_{\alpha,\alpha'}\|_{L^{\infty}([0,T])}\right\},\quad \epsilon 
        =   \sum_{\alpha,\alpha'=1}^N \|\mathcal C_{\alpha,\alpha'}-\mathcal C'_{\alpha,\alpha'}\|_{L^{\infty}([0,T])} 
\\
       \mathcal M_1 = \max\left\{\sum_{\alpha,\alpha'=1}^N \|\mathcal C_{\alpha,\alpha'}\|_{L^{1}([0,T])}, \sum_{\alpha,\alpha'=1}^N \|\mathcal C'_{\alpha,\alpha'}\|_{L^{1}([0,T])}\right\},
    \quad
        \epsilon_1
        =   \sum_{\alpha,\alpha'=1}^N \|\mathcal C_{\alpha,\alpha'}-\mathcal C'_{\alpha,\alpha'}\|_{L^{1}([0,T])}    .
            \end{aligned}
            \label{eq:epsilon_new}
    \end{equation}
    $\mathcal C_{\alpha_1,\alpha_2}$ and $\mathcal C'_{\alpha_1,\alpha_2}$ are  single-variable functions, and with comparison to $\delta, M, \delta_1, M_1$ defined in \cref{thm:Gronwall_error_bound}, we have
    $$
     \delta = \epsilon,\quad 
  M =   \mathcal M, \quad
    \delta_1\leq 2\epsilon_1,
    \quad M_1\leq 2 \mathcal M.
    $$
    As a result, the following error bound holds:
    \begin{equation}
   \left\|\hat{\rho}_{\mathrm{S}}(t)-\hat{\rho}_{\mathrm{S}}^{\prime}(t)\right\|_{\mathrm{tr}} \leq \epsilon_1 {t}\mathrm{e}^{\mathcal M_1 t},\quad \left\|\hat{\rho}_{\mathrm{S}}(t)-\hat{\rho}_{\mathrm{S}}^{\prime}(t)\right\|_{\mathrm{tr}} \leq \epsilon \frac{t^2}{2} \mathrm{e}^{\mathcal M t^2 / 2}.
    \end{equation}
    \label{rmk:single_variable_Corr}
\end{rem}
\begin{rem}
    As an application of \cref{lem:rho_to_O},  an error bound similar to \cref{rmk:single_variable_Corr} holds for any bounded system observable $\hat O_{\text S}$. Therefore \cref{lem:Gronwall_informal} is the direct corollary of \cref{thm:Gronwall_error_bound} using \cref{lem:rho_to_O}.
    \label{rem:Gronwall_cor}
\end{rem}
\begin{rem}
    Note that it could be seen from the definition of $\epsilon$, $\delta$, $\mathcal M$, $ M$ that they may scale quadratically with the system size $N$. This dependence on $N$ is natural and expected. However, the primary focus of this work is on the exponential improvement in the $t$-dependence of the error bound, as discussed in \cref{sec:intro}, which is the more dominant factor.
\end{rem}
\subsection{Improved error bound}
\label{sec:main}
Now let us discuss the improved error bound. Under the same setup as above, we will assume $(\mathrm e^{\boldsymbol{L}_{\text s}t})_{t\geq 0}$ is a contractive semigroup. Furthermore, let $\boldsymbol{L}$, $\boldsymbol{L}'$ denote the total system-environment Liouvillian corresponds to $\hat\rho_{\text S}(t)$ and $\hat\rho_{\text S}'(t)$, respectively.
We  also assume that $(\mathrm e^{\boldsymbol{L}t})_{t\geq 0}$ forms a contractive semigroup. Note that we do not impose the same assumption on   the perturbed semigroup $(\mathrm e^{\boldsymbol{L}' t})_{t\geq 0}$, allowing our analysis applicable to non-contractive frameworks such as quasi-Lindblad formalism (see \cref{sec:quasi-Lind}).
\begin{thm}[Main theorem]
With the same setup as in \cref{thm:Gronwall_error_bound}, and $(\mathrm e^{\boldsymbol{L}t})_{t\geq 0}$ is a contractive semigroup, 
we have the following improved error bound
    \begin{equation}
        \|\hat\rho_{\text S}(t) - \hat\rho_{\text S}'(t)\|_{\text{tr}} \leq \mathrm e^{\delta t^2/2}-1,
    \end{equation}
    where $\delta$ is the difference between $C$ and $C'$ in the $L^{\infty}$ sense, as defined in \cref{eq:def_eps}, while in the $L^1$ sense, we have 
    \begin{equation}
        \|\hat\rho_{\text S}(t) - \hat\rho_{\text S}'(t)\|_{\text{tr}} \leq \mathrm e^{\delta_1 t/2}-1,
    \end{equation}
    where $\delta_1$ is also defined the same as in \cref{eq:def_eps1}.
    \label{thm:main_error_bound}
\end{thm}
\begin{rem}Similar to \cref{rmk:single_variable_Corr}, in the case that $C_{\alpha_1,\alpha_2}(t_1,t_2)$ depends only on the time difference $(t_1-t_2)$, we have the following simplified error bound:
\begin{equation}
    \left\|\hat{\rho}_{\mathrm{S}}(t)-\hat{\rho}_{\mathrm{S}}^{\prime}(t)\right\|_{\mathrm{tr}} \leq \mathrm{e}^{\epsilon t^2 / 2}-1,\quad 
    \left\|\hat{\rho}_{\mathrm{S}}(t)-\hat{\rho}_{\mathrm{S}}^{\prime}(t)\right\|_{\mathrm{tr}} \leq \mathrm{e}^{\epsilon_1 t }-1,
\end{equation}
where $\epsilon, \epsilon_1$ are defined as in \cref{eq:epsilon_new}.
\end{rem}
\begin{rem}
    Similar to \cref{rem:Gronwall_cor}, using \cref{lem:rho_to_O}, the error bound for system observables (see \cref{thm:main}) follow naturally as the corollary of \cref{thm:main_error_bound}.
\end{rem}

Prior to giving a rigorous proof for \cref{thm:main_error_bound}, let us first introduce a formal yet intuitive argument.
\begin{proof}[Formal Proof of \cref{thm:main_error_bound}]
Let us use the formal exponential formula in \cref{rmk:formal_exponential_formula}.
Formally, all operators (anti-)commute under the time-ordering operation $\mathbf T^{\pm}$, 
\begin{equation}
    \begin{aligned}
     &   \mathbf T^{\pm} \left(\mathrm{e}^{\frac 1 2\int_0^t\int_0^t \sum_{\alpha_1,\alpha_2=1}^N C'_{\alpha_1,\alpha_2}(t_1,t_2)\boldsymbol S_{\alpha_1}(t_1)\boldsymbol S_{\alpha_2}(t_2)\mathrm d t_1\mathrm d t_2}\right) 
     \\   =     
     &\mathbf T^{\pm} \left(\mathrm{e}^{\frac 1 2\int_0^t\int_0^t \sum_{\alpha_1,\alpha_2=1}^N C_{\alpha_1,\alpha_2}(t_1,t_2)\boldsymbol S_{\alpha_1}(t_1)\boldsymbol S_{\alpha_2}(t_2)\mathrm d t_1\mathrm d t_2}\mathrm{e}^{\frac 1 2\int_0^t\int_0^t \sum_{\alpha_1,\alpha_2=1}^N \Delta C_{\alpha_1,\alpha_2}(t_1,t_2)\boldsymbol S_{\alpha_1}(t_1)\boldsymbol S_{\alpha_2}(t_2)\mathrm d t_1\mathrm d t_2}\right). 
    \end{aligned}
    \label{eq:formal_exp_decomp}
\end{equation}
Note that \cref{eq:formal_exp_decomp} is a formal expression for now.
Then formally, we have
\begin{equation}
    \begin{aligned}
        \hat\rho_{\text S}'(t) - \hat\rho_{\text S}(t)
        = &\mathrm e^{\boldsymbol{L}_{\text s}t}\left(\mathbf T^{\pm} 
        \left(\mathrm e^{ \frac 1 2\int_0^t\int_0^t \sum_{\alpha_1,\alpha_2=1}^N C_{\alpha_1,\alpha_2}(t_1,t_2) \boldsymbol{S}_{\alpha_1}(t_1)\boldsymbol{S}_{\alpha_2}(t_2)\mathrm d t_1\mathrm d t_2}\right) \right.\\
       & \cdot \left.
        \left(\mathrm e^{ \frac 1 2\int_0^t\int_0^t \sum_{\alpha_1,\alpha_2=1}^N \Delta C_{\alpha_1,\alpha_2}(t_1,t_2) \boldsymbol{S}_{\alpha_1}(t_1)\boldsymbol{S}_{\alpha_2}(t_2)\mathrm d t_1\mathrm d t_2}-1\right)\right)\hat\rho_{\text S}(0).    
    \end{aligned}
    \label{eq:exp_diff_formal}
   \end{equation}
 A formal expansion on the second exponential in \cref{eq:exp_diff_formal} gives 
   \begin{equation}
    \begin{aligned}
      & \hat\rho_{\text S}'(t) - \hat\rho_{\text S}(t)
    = 
    \sum_{n=1}^{\infty}\frac{1}{2^nn!}\sum_{\alpha_1,\cdots,\alpha_{2n}=1}^N\int_0^t\cdots\int_0^t\mathrm dt_1\cdots \mathrm dt_{2n}
        \prod_{i=1}^n\Delta C_{\alpha_{2i-1},\alpha_{2i}}(t_{2i-1},t_{2i})
       \\&\cdot \mathrm e^{\boldsymbol{L}_{\text s}t}\mathbf T^{\pm} \left( \left(\prod_{j=1}^{2n}\boldsymbol{S}_{\alpha_j}(t_j)\right)\mathrm e^{\frac 1 2\int_0^t \int_0^t \sum_{\beta_1,\beta_2=1}^N C_{\beta_1,\beta_2}(s_1,s_2)\boldsymbol{S}_{\beta_1}(t_1)\boldsymbol{S}_{\beta_2}(t_2)\mathrm d s_1\mathrm d s_2}
    \right)   \hat\rho_{\text S}(0).
    \end{aligned}
    \label{eq:diff_formal}
   \end{equation} 
   The derivation of \cref{eq:diff_formal} here is formal and will be justified rigorously in \cref{lem:correlator_formal}.
  We will also prove the following equality in \cref{lem:correlator_return_to_extended_system}:
  \begin{equation}
    \begin{aligned}
        &\mathrm e^{\boldsymbol{L}_{\text s}t}\mathbf T^{\pm} \left( \left(\prod_{j=1}^{2n}\boldsymbol{S}_{\alpha_j}(t_j)\right)\mathrm e^{\frac 1 2\int_0^t \int_0^t \sum_{\beta_1,\beta_2=1}^N C_{\beta_1,\beta_2}(s_1,s_2)\boldsymbol{S}_{\beta_1}(s_1)\boldsymbol{S}_{\beta_2}(s_2)\mathrm d s_1\mathrm d s_2}
    \right)   \hat\rho_{\text S}(0)\\
    = & (\pm)^{\sigma}\operatorname{tr}_{\text E}\left(\mathrm e^{\boldsymbol L(t-t_{\sigma(1)})}\widetilde{\boldsymbol{S}}_{\alpha_{\sigma(1)}}
    \prod_{i=2}^{2n}\left(\mathrm e^{ \boldsymbol L(t_{\sigma(i-1)}-t_{\sigma(i)})}\widetilde{\boldsymbol{S}}_{\alpha_{\sigma(i)}}
    \right)
    \mathrm e^{\boldsymbol{L}t_{\sigma(2n)}}
    \hat\rho(0) \right),
    \end{aligned}
\label{eq:correlator_return_to_extended_system}
   \end{equation}
where $\sigma\in S_{2n}$ is the permutation of $1,\cdots,2n$ such that $t_{\sigma(1)}\geq \cdots\geq t_{\sigma(2n)}$, 
and $\widetilde{\boldsymbol S}_{\alpha}$ is the extension of $\boldsymbol{S}_{\alpha}$ to the entire extended system. For bosonic environments, we define 
$\widetilde{\boldsymbol S}_{\alpha} =\boldsymbol 1_{\text e}\otimes \boldsymbol S_{\alpha}$.
For fermionic systems, we refer to \cref{sec:proof_correlator_return_to_extended_system_fermionic} for the definition of $\widetilde{\boldsymbol S}_{\alpha}$.
Since the partial trace does not increase the trace norm, the trace norm of LHS of \cref{eq:correlator_return_to_extended_system} is bounded by the  trace norm  for the following:
 $$
 \left(\mathrm e^{\boldsymbol L(t-t_{\sigma(1)})}\widetilde{\boldsymbol{S}}_{\alpha_{\sigma(1)}}
    \prod_{i=2}^{2n}\left(\mathrm e^{ \boldsymbol L(t_{\sigma(i-1)}-t_{\sigma(i)})}\widetilde{\boldsymbol{S}}_{\alpha_{\sigma(i)}}
    \right)
    \mathrm e^{\boldsymbol{L}t_{\sigma(2n)}}
 \right)\hat\rho(0).
 $$
 Furthermore, since  $\mathrm e^{\boldsymbol L t}$ are contraction mappings for $t\geq 0$, i.e., $\|\mathrm e^{\boldsymbol Lt}\|_{\text{tr}}\leq 1$, the trace norm of the above is further bounded by 
 \begin{equation}
 \prod_{i=1}^{2n}\| \widetilde{\boldsymbol{S}}_{\alpha_{\sigma(i)}}\|_{\text{tr}} \|\hat\rho(0)\|_{\text{tr}}.
 \label{eq:finalbound}
 \end{equation} In the bosonic case, $\|\widetilde{\boldsymbol{S}}_{\alpha_i}\|_{\text tr} = \|\boldsymbol{S}_{\alpha_i}\|_{\text tr} = 1$, and then
 \cref{eq:finalbound} is bounded by 1. In the fermionic case, the same error bound holds (see \cref{sec:proof_correlator_return_to_extended_system_fermionic}).
Now, we have 
\begin{equation}
    \|\hat\rho_{\text S}'(t) - \hat\rho_{\text S}(t)\|_{\text{tr}} \leq\int_0^t\cdots\int_0^t\mathrm dt_1\cdots \mathrm dt_{2n} \sum_{n=1}^{\infty}\frac{1}{2^nn!}\prod_{i=1}^n\sum_{\alpha_{2i-1},\alpha_{2i}=1}^N|\Delta C_{\alpha_{2i-1},\alpha_{2i}}(t_{2i-1},t_{2i})|.
\end{equation}
Hence, similar to the proof of \cref{thm:Gronwall_error_bound}, in the $L^\infty$ sense, we have $\|\hat\rho_{\text S}'(t) - \hat\rho_{\text S}(t)\|_{\text{tr}} \leq \sum_{n=1}^\infty \frac{\delta^nt^{2n}}{2^nn!}=\mathrm e^{\delta t^2/2} -1 $, and in the $L^1$ sense, we have 
    $\|\hat\rho_{\text S}'(t) - \hat\rho_{\text S}(t)\|_{\text{tr}} \leq \sum_{n=1}^\infty \frac{\delta_1^n t^n}{2^nn!}=\mathrm e^{\delta_1 t/2} -1 $.
\end{proof}
\begin{rem}
    We remark that the equality \cref{eq:correlator_return_to_extended_system} (proven in \cref{lem:correlator_return_to_extended_system}) has been formally derived previously in physics literature in coherent state path integral formalism (for example, see \cite[Eq. 5]{MascherpaSmirneHuelgaetal2017}). However, this type of coherent state path integral formalism, aside from the challenges of achieving rigorous justification, has been shown to sometimes yield incorrect results \cite{Kochetov2019}. In contrast, the derivations presented in this paper, based on the superoperator formalism, are fully valid.
\end{rem}
\subsection{Proof of Theorem \ref{thm:main_error_bound}}
\label{sec:proof_main}
To rigorously prove \cref{thm:main_error_bound}, we need to mathematically justify \cref{eq:diff_formal} and the equality \cref{eq:correlator_return_to_extended_system}.
Let us first prove the former:

\begin{lem}
    \cref{eq:diff_formal} holds, where the expression
    \begin{equation}
        \mathbf T^{\pm}\left(\prod_{i=1}^{2n}\boldsymbol S_{\alpha_i}(t_i) \mathrm e^{\frac 1 2\int_0^t\int_0^t \sum_{\beta_1,\beta_2=1}^N C_{\beta_1,\beta_2}(s_1,s_2)\boldsymbol S_{\beta_1}(s_1)\boldsymbol S_{\beta_2}(s_2)\mathrm d s_1\mathrm d s_2}\right)
    \label{eq:correlator_formal}
    \end{equation}
    is understood as shorthand for the following series
    \begin{equation}
        \sum_{k=0}^\infty\frac{1}{2^{k}k!} \sum_{\beta_1,\cdots,\beta_{2k}=1}^N \prod_{j=1}^k C_{\beta_{2 j-1}, \beta_{2 j}}\left(s_{2 j-1}, s_{2 j}\right)
     \mathbf T^{\pm}\left(\prod_{i=1}^{2 n}\boldsymbol{S}_{\alpha_i}\left(t_i\right)\prod_{j=1}^{2 k} \boldsymbol{S}_{\beta_j}\left(s_j\right)\right) ,
     \label{eq:correlator_rigorous}
    \end{equation}
    where the time-ordering acts jointly on the full list of times $\{t_1,\ldots,t_{2n},s_1,\ldots,s_{2k}\}$ as in \cref{rem:time_ordering}.
\label{lem:correlator_formal}
\end{lem}
\begin{proof}
    Since $C_{\alpha_1,\alpha_2}'(t_1,t_2) = C_{\alpha_1,\alpha_2}(t_1,t_2) + \Delta C_{\alpha_1,\alpha_2}(t_1,t_2)$, according to \cref{eq:rhoS_Dyson_short}, we have
 
    \begin{align*}
        \hat\rho_{\text S}'(t) &= \sum_{n=0}^{\infty} \frac{1}{2^n n!} \int_0^t \cdots \int_0^t \sum_{\alpha_1, \cdots, \alpha_{2 n}=1}^N\left(\prod_{i=1}^n (C_{\alpha_{2 i-1}, \alpha_{2 i}}+\Delta C_{\alpha_{2 i-1}, \alpha_{2 i}})\left(t_{2 i-1}, t_{2 i}\right)\right) 
        \\&\cdot\mathrm e^{\boldsymbol L_{\text s}t}\mathbf T^{\pm}\left(\prod_{j=1}^{2 n} \boldsymbol{S}_{\alpha_j}\left(t_j\right)\right) \mathrm{d} t_1 \cdots \mathrm{d} t_{2 n} \\=& \sum_{n=0}^{\infty} \frac{1}{2^n n!} \int_0^t \cdots \int_0^t \sum_{\alpha_1, \cdots, \alpha_{2 n}=1}^N\sum_{I \subset \{1,\cdots,n\}}\left(\prod_{i\in I}  C_{\alpha_{2 i-1}, \alpha_{2 i}} \left(t_{2 i-1}, t_{2 i}\right)\prod_{i'\in  I'}  \Delta C_{\alpha_{2 i'-1}, \alpha_{2 i'}} \left(t_{2 i'-1}, t_{2 i'}\right)\right) 
        \\&\cdot\mathrm e^{\boldsymbol L_{\text s}t}\mathbf T^{\pm}\left(\prod_{j=1}^{2 n} \boldsymbol{S}_{\alpha_j}\left(t_j\right)\right) \mathrm{d} t_1 \cdots \mathrm{d} t_{2 n} \quad (I' =\{1,\cdots,n\}\backslash I) 
    \\ = & \sum_{n=0}^{\infty} \frac{1}{2^n n!} \int_0^t \cdots \int_0^t \sum_{\alpha_1, \cdots, \alpha_{2 n}=1}^N\sum_{k=0}^n\left(\begin{array}{c} n \\ k \end{array}\right)
    \left(\prod_{i=1}^k C_{\alpha_{2 i-1}, \alpha_{2 i}}\left(t_{2 i-1}, t_{2 i}\right)\right) \\
    &\cdot\left(\prod_{i=k+1}^n \Delta C_{\alpha_{2 i-1}, \alpha_{2 i}}\left(t_{2 i-1}, t_{2 i}\right)\right)\mathrm e^{\boldsymbol L_{\text s}t} \mathbf T^{\pm}\left(\prod_{j=1}^{2 n} \boldsymbol{S}_{\alpha_j}\left(t_j\right)\right) \mathrm{d} t_1 \cdots \mathrm{d} t_{2 n}\\
    =& 
    \sum_{m=0}^\infty \frac{1}{2^m m!} \int_0^t\mathrm dt_1\cdots\int_0^t\mathrm dt_{2m}\sum_{\alpha_1,\cdots,\alpha_{2m}=1}^N\left(\prod_{i=1}^m \Delta C_{\alpha_{2 i-1}, \alpha_{2 i}}\left(t_{2 i-1}, t_{2 i}\right)\right) 
    \sum_{k=0}^\infty\frac{1}{2^{k}k!} 
    \\ &\cdot \int_0^t \mathrm ds_1\cdots \int_0^t \mathrm ds_{2j}\sum_{\beta_1,\cdots,\beta_{2k}=1}^N \left(\prod_{j=1}^k C_{\beta_{2 j-1}, \beta_{2 j}}\left(s_{2 j-1}, s_{2 j}\right)\right)
    \mathrm e^{\boldsymbol L_{\text s}t}\mathbf T^{\pm}\left(\prod_{i=1}^{2 m}\boldsymbol{S}_{\alpha_j}\left(t_j\right)\prod_{j=1}^{2 k} \boldsymbol{S}_{\beta_j}\left(s_j\right)\right) .
\end{align*}

Here, in the third equality, we have used the following fact: for $t_1,\cdots, t_{2n}$ that are mutually distinct, $ \mathrm e^{\boldsymbol L_{\text s}t}\mathbf T^{\pm}\left(\prod_{i=1}^{2 n}\boldsymbol{S}_{\alpha_j}(t_j)\right)$ is invariant after exchanging $ \boldsymbol{S}_{\alpha_{2i-1}}(t_{2_i-1})\boldsymbol{S}_{\alpha_{2i }}(t_{2_i })$ with $ \boldsymbol{S}_{\alpha_{2i'-1}}(t_{2_i'-1})\boldsymbol{S}_{\alpha_{2i' }}(t_{2_i' })$, for any $i,i'\in \{1,\cdots,2n\}$. In the last equality, we exchange the order of summation, which is justified by the dominated convergence theorem since all series summations are absolutely convergent.
Thus, $\hat\rho_{\text{S}}'(t) - \hat\rho_{\text S}(t)$ could be obtained by replacing $\sum_{m=0}^\infty$ with $\sum_{m=1}^\infty$ in the above equation,
i.e., \cref{eq:diff_formal} holds.
\end{proof}

Now let us prove the equality \cref{eq:correlator_return_to_extended_system}, which would enable us to take advantage of the contractive properties of $\mathrm e^{\boldsymbol{L}t}$.

\begin{lem} Under the bounds in \cref{eq:estimate}, the identity in \cref{eq:correlator_return_to_extended_system} holds, with \cref{eq:correlator_formal} interpreted as the shorthand of \cref{eq:correlator_rigorous} (see also \cref{lem:correlator_formal}).
\label{lem:correlator_return_to_extended_system}
\end{lem}
\begin{proof}[Proof of \cref{lem:correlator_return_to_extended_system}]
    Let us  prove both sides of \cref{eq:correlator_return_to_extended_system} are equal to the following:
    \begin{equation}
   \begin{aligned}
    &(\pm)^{\sigma_t}\operatorname{tr}_{\text E}\left(\mathrm e^{\boldsymbol{L}_0t}\sum_{m=0}^\infty\frac 1 {m!}\int_0^t\cdots\int_0^t\mathrm ds_1\cdots \mathrm ds_m\mathbf T^{\pm} 
    \left(
       \prod_{i=1}^{2n} \widetilde{\boldsymbol{S}}_{\alpha_i}(t_i)
       \prod_{j=1}^m \boldsymbol{L}_{\text{SE}}(s_j)
    \right)\hat\rho(0)\right).
   \end{aligned}
\label{eq:return_to_extended_system_intermediate}
    \end{equation}
    where $\sigma_t$ is the permutation of $1,\cdots,2n$ such that $t_{\sigma_t(1)}\geq \cdots\geq t_{\sigma_t(2n)}$.
   {Note that similar to the proof of \cref{lem:dyson_conv}, the infinite series in \cref{eq:return_to_extended_system_intermediate} is absolutely convergent. }
    For the left-hand side, we have
        \begin{align}
            &LHS\notag\\ =& \sum_{k=0}^{\infty} \frac{\mathrm e^{\boldsymbol{L}_{\text s}t}}{2^k k!}\int_0^t\mathrm ds_1\cdots \int_0^t\mathrm d s_{2k} \sum_{\beta_1, \cdots, \beta_{2 k}=1}^N \prod_{j=1}^k C_{\beta_{2 j-1}, \beta_{2 j}}\left(s_{2 j-1}, s_{2 j}\right) \mathbf{T}^{\pm}\left(\prod_{i=1}^{2 n} \boldsymbol{S}_{\alpha_i}\left(t_i\right) \prod_{j=1}^{2 k} \boldsymbol{S}_{\beta_j}\left(s_j\right)\right)\hat\rho_{\text S}(0)\notag\\
             = & \sum_{k=0}^{\infty}\sum_{\sigma\in\Pi_{2k}}\frac{\mathrm e^{\boldsymbol{L}_{\text s}t}}{(2k)!}
            \int_0^t\mathrm ds_1\cdots \int_0^t\mathrm d s_{2k}\sum_{\beta_1, \cdots, \beta_{2 k}=1}^N\notag\\
            & \qquad \prod_{j=1}^k C_{\beta_{\sigma(2 j-1)}, \beta_{\sigma(2 j)}}\left(s_{\sigma(2 j-1)}, s_{\sigma(2 j)}\right) \mathrm e^{\boldsymbol{L}_{\text s}t}\mathbf{T}^{\pm}\left(\prod_{i=1}^{2 n} \boldsymbol{S}_{\alpha_i}\left(t_i\right) \prod_{j=1}^{2 k} \boldsymbol{S}_{\beta_{\sigma(j)}}\left(s_{\sigma({j})}\right)\right)\hat\rho_{\text S}(0) \notag\\
            = & \sum_{k=0}^{\infty}\sum_{\sigma\in\Pi_{2k}}\frac{\mathrm e^{\boldsymbol{L}_{\text s}t}}{(2k)!}
            \int_0^t\mathrm ds_1\cdots \int_0^t\mathrm d s_{2k}\sum_{\beta_1, \cdots, \beta_{2 k}=1}^N (\pm)^\sigma\notag\\
            & \qquad \prod_{j=1}^k C_{\beta_{\sigma(2 j-1)}, \beta_{\sigma(2 j)}}\left(s_{\sigma(2 j-1)}, s_{\sigma(2 j)}\right) \mathrm e^{\boldsymbol{L}_{\text s}t}\mathbf{T}^{\pm}\left(\prod_{i=1}^{2 n} \boldsymbol{S}_{\alpha_i}\left(t_i\right) \prod_{j=1}^{2 k} \boldsymbol{S}_{\beta_{j}}\left(s_{j}\right)\right)\hat\rho_{\text S}(0)  
            \notag\\
            =& \sum_{k=0}^{\infty}\frac{\mathrm e^{\boldsymbol{L}_{\text s}t}}{(2k)!}
            \int_0^t\mathrm ds_1\cdots \int_0^t\mathrm d s_{2k}\sum_{\beta_1, \cdots, \beta_{2 k}=1}^N  C^{(2k)}_{\beta_1,\cdots,\beta_{2k}}(s_1,\cdots,s_{2k})
            \mathbf{T}^{\pm}\left(\prod_{i=1}^{2 n} \boldsymbol{S}_{\alpha_i}\left(t_i\right) \prod_{j=1}^{2 k} \boldsymbol{S}_{\beta_{j}}\left(s_{j}\right)\right)\hat\rho_{\text S}(0)
           \notag \\
            =& \sum_{m=0}^\infty\frac{\mathrm e^{\boldsymbol{L}_{\text s}t}}{m!}\int_0^t \mathrm ds_1\cdots\int_0^t\mathrm ds_m
            \sum_{\beta_1, \cdots, \beta_{m}=1}^N  C^{(m)}_{\beta_1,\cdots,\beta_{m}}(s_1,\cdots,s_{m})
            \mathbf T^{\pm}\left(\prod_{i=1}^{2n} {\boldsymbol{S}}_{\alpha_i}(t_i)
            \prod_{j=1}^{m} \boldsymbol{S}_{\beta_j}(s_j)\right)\hat\rho_{\text S}(0).\notag
    \end{align}
    Here, we have used the Isserlis-Wick's conditions (see \cref{defn:Wick_condition}) for $m$-point BCFs $C^{(m)}(s_1,\cdots, s_m)
    = \operatorname{tr}\left(
    \mathbf{T}^{\pm}\left(\prod_{j=1}^m \boldsymbol{E}_{\beta_j}(s_j)\right)\hat\rho_{\text E}(0)    
    \right)$. Then, we only need to prove 
    \begin{equation}
        \begin{aligned}
           & \sum_{m=0}^\infty\frac{\mathrm e^{\boldsymbol{L}_{\text s}t}}{m!}\int_0^t \mathrm ds_1\cdots\int_0^t\mathrm ds_m
            \sum_{\beta_1, \cdots, \beta_{m}=1}^N  C^{(m)}_{\beta_1,\cdots,\beta_{m}}(s_1,\cdots,s_{m})
            \mathbf T^{\pm}\left(\prod_{i=1}^{2n} {\boldsymbol{S}}_{\alpha_i}(t_i)
            \prod_{j=1}^{m} \boldsymbol{S}_{\beta_j}(s_j)\right)\hat\rho_{\text S}(0).\\
         = &(\pm)^{\sigma_t}\operatorname{tr}_{\mathrm{E}}\left(\mathrm{e}^{\boldsymbol{L}_0 t} \sum_{m=0}^{\infty} \frac{1}{m!} \int_0^t \cdots \int_0^t \mathrm{~d} s_1 \cdots \mathrm{d} s_m \boldsymbol{\mathcal T}\left(\prod_{i=1}^{2 n} \widetilde{\boldsymbol{S}}_{\alpha_i}(t_i) \prod_{j=1}^m \boldsymbol{L}_{\mathrm{SE}}(s_j)\right) \hat{\rho}(0)\right).
        \end{aligned}
        \label{eq:the_thing_need_to_be_proved}
    \end{equation}

    For bosonic cases, we further have
    $$
    \mathbf{T}^{+}\left(\prod_{i=1}^{2 n} \mathbf{1}_{\mathrm{e}}\left(t_i\right) \prod_{j=1}^m \boldsymbol{E}_{\beta_j}\left(s_j\right)\right)\hat\rho_{\text E}(0) = \mathbf{T}^{+}\left(\prod_{j=1}^m \boldsymbol{E}_{\beta_j}\left(s_j\right)\right)\hat\rho_{\text E}(0).
    $$
Here,   $\boldsymbol 1_{\text e}(t) = \boldsymbol 1_{\text e}$ and the time $t$ serves as the time argument while performing time-ordering.
    The above equality holds since ${\mathbf T}^+$  is the time-ordering operation without adding any additional sign. We have
            \begin{equation}
                \begin{aligned}     
                    &\text{LHS of \cref{eq:the_thing_need_to_be_proved}}\\
             =& \sum_{m=0}^{\infty}\frac{\mathrm e^{\boldsymbol{L}_{\text s}t}}{m!}
            \int_0^t \mathrm ds_1\cdots \int_0^t \mathrm d s_{m}\\&\qquad\qquad            \operatorname{tr}\left(\mathbf T^+ \left(\prod_{i=1}^{2n}\boldsymbol 1_{\text e}(t_i) \prod_{j=1}^{m}\boldsymbol E_{\beta_j}(s_j)\right)\hat\rho_{\text E}(0)\right)
           \mathbf T^+\left( \prod_{i=1}^{2n}\boldsymbol S_{\alpha_i}(t_i)\prod_{j=1}^{m}\boldsymbol S_{\beta_j}(s_j)\right)\hat\rho_{\text S}(0)
            \\ = & \text{RHS of \cref{eq:the_thing_need_to_be_proved}}.
        \end{aligned}
    \end{equation}
    The above derivation is similar to reversing the proof of \cref{thm:rhoS_Dyson_short}, and we also have used that 
    $|\Pi_{2k}| = (2k-1)!!$, $\mathrm e^{L_{\text e}t}$ is a trace-preserving map, and the Isserlis-Wick's condition (see \cref{defn:Wick_condition}). During this process, we have also exchanged the partial trace and integral, which is justified by the dominated convergence theorem, similar as in the proof of \cref{lem:reduced_dyson}.
    In the fermionic case, one needs to pay extra attention to the additional sign arising from the $\mathbb Z_2$-graded tensor product structure (see \cref{sec:Z2}).
    We leave the proof of \cref{eq:the_thing_need_to_be_proved} in the fermionic case to \cref{sec:proof_correlator_return_to_extended_system_fermionic}.

    For RHS of \cref{eq:correlator_return_to_extended_system}, without loss of generality, let us assume that $t_{2n}\geq t_{2n-1}\geq \cdots\geq t_{1}$.
    Comparing \cref{eq:return_to_extended_system_intermediate} and the RHS of \cref{eq:correlator_return_to_extended_system}, we only need to prove that
    \begin{equation}
        \begin{aligned}
            &\mathrm e^{\boldsymbol{L}_{0}t}\sum_{m=0}^\infty\frac 1 {m!}\int_0^t\cdots\int_0^t\mathrm ds_1\cdots \mathrm ds_m\boldsymbol{\mathcal T}
            \left(
               \prod_{i=1}^{2n} \widetilde{\boldsymbol{S}}_{\alpha_i}(t_i)
               \prod_{j=1}^m \boldsymbol{L}_{\text{SE}}(s_j)
            \right)\hat\rho(0)
            \\=& \left(\mathrm e^{\boldsymbol L(t-t_{2n})}\widetilde{\boldsymbol{S}}_{\alpha_{2n}}
    \prod_{i=2}^{2n}\left(\mathrm e^{ \boldsymbol L(t_{i}-t_{i-1})}\widetilde{\boldsymbol{S}}_{\alpha_{i-1}}
    \right)
    \mathrm e^{\boldsymbol{L}t_{1}}\right) \hat\rho(0).        \end{aligned}
\label{eq:correlator_return_to_extended_RHS}
    \end{equation}
Since LHS of \cref{eq:correlator_return_to_extended_RHS} is absolutely convergent, we have
\allowdisplaybreaks
\begin{align}
&\text{LHS of \cref{eq:correlator_return_to_extended_RHS}} \nonumber
\\
= & \mathrm e^{\boldsymbol L_0t}
\left(\sum_{m_{2n}=0}^\infty \frac 1 {m_{2n}!}\int_{t_{2n}}^t \mathrm{~d} s_1^{(2n)} \cdots \int_{t_{2n}}^t \mathrm{~d} s_{m_{2n}}^{(2n)} \boldsymbol{\mathcal T}\left(\prod_{i=1}^{m_{2n}} \boldsymbol{L}_{\mathrm{SE}}\left(s_i^{(2n)}\right)\right)\right)
\widetilde{\boldsymbol{S}}_{\alpha_{2n}}(t_{2n}) \nonumber\\
&\cdot \left(\sum_{m_{2n-1}=0}^\infty \frac 1 {m_{2n-1}!}\int_{t_{2n-1}}^{t_{2n}}  \mathrm d s_1^{(2n-1)} \cdots \int_{t_{2n-1}}^{t_{2n}}  \mathrm d s_{m_{2n-1}}^{({2n-1})} \boldsymbol{\mathcal T}\left(\prod_{i=1}^{m_{2n-1}} \boldsymbol{L}_{\mathrm{SE}}\left(s_i^{({2n-1})}\right)\right)\right) \widetilde{\boldsymbol{S}}_{\alpha_{2n-1}}(t_{2n-1}) \cdots \nonumber \\ 
&\cdot \left(\sum_{m_{1}=0}^\infty \frac 1 {m_{1}!}\int_{t_{1}}^{t_{2}}  \mathrm d s_1^{(1)} \cdots \int_{t_{1}}^{t_{2}}\mathrm d s_{m_{1}}^{(1)} \boldsymbol{\mathcal T}\left(\prod_{i=1}^{m_{1}} \boldsymbol{L}_{\mathrm{SE}}\left(s_i^{(1)}\right)\right)\right)
\widetilde{\boldsymbol{S}}_{\alpha_{1}}(t_{1}) \nonumber\\
&\cdot \left(\sum_{m_{0}=0}^\infty \frac 1 {m_{0}!}\int_{0}^{t_{2n}}  \mathrm d s_1^{(0)} \cdots \int_{0}^{t_{1}} \mathrm d s_{m_{0}}^{(0)} \boldsymbol{\mathcal T}\left(\prod_{i=1}^{m_{0}} \boldsymbol{L}_{\mathrm{SE}}\left(s_i^{(0)}\right)\right)\right)\hat\rho(0) .\label{eq:intermediate_long}
\end{align}

 Here we have used the following combinatorial equality: for $m = m_1+\cdots+m_{2n+1}$, 
$$
\frac{1}{m!}\left(
    \begin{array}{c}
        m \\ m_1
    \end{array}
\right)\left(
    \begin{array}{c}
        m-m_1 \\ m_2
    \end{array}\right)\cdots \left(
        \begin{array}{c}
            m-m_1-\cdots-m_{2n} \\ m_{2n+1}
        \end{array}
    \right) = \frac{1}{m_1!}\cdots \frac{1}{m_{2n+1}!}.
$$
Proceeding further, using the definition of the time-ordering operator, we see that \cref{eq:intermediate_long} is equivalent to the following equality, for all $k\geq 0$, and all $\alpha_i = 1,\cdots,N$:
\begin{align}
    &    \mathrm e^{\boldsymbol{L}(t-t_k)}\widetilde{\boldsymbol{S}}_{\alpha_k} \mathrm e^{\boldsymbol{L}(t_k-t_{k-1})}\widetilde{\boldsymbol{S}}_{\alpha_{k-1}} \cdots \mathrm e^{\boldsymbol{L}(t_2-t_{1})}\widetilde{\boldsymbol{S}}_{\alpha_{1}} \mathrm e^{\boldsymbol{L}t_1}\hat\rho(0)\nonumber \\
    =&\sum_{n_0,n_1,\cdots, n_k=0}^\infty
    \int_{t\geq s_{n_k}^{(k)}\geq \cdots\geq s_1^{(k)}\geq t_k}\mathrm ds_1^{(k)}\cdots\mathrm ds_{n_k}^{(k)}  \int_{t_k\geq s_{n_{k-1}}^{(k-1)}\geq \cdots\geq s_1^{(k-1)}\geq t_{k-1}}\mathrm ds_1^{(k-1)}\cdots\mathrm ds_{n_{k-1}}^{(k-1)} \nonumber\\&\cdots  \int_{t_1\geq s_{n_0}^{(0)}\geq \cdots\geq s_1^{(0)}\geq t_0}\mathrm ds_1^{(0)}\cdots\mathrm ds_{n_0}^{(0)}   \left(\mathrm e^{\boldsymbol L_0 (t-s_{n_k}^{(k)})}\boldsymbol L_{\text{SE}}\mathrm e^{\boldsymbol L_0 (s_{n_k}^{(k)}-s_{n_k-1}^{(k)})}\cdots \boldsymbol L_{\text{SE}}\mathrm e^{\boldsymbol L_0 (s_1^{(k)}-t_k)}\right) \widetilde{\boldsymbol{S}}_{\alpha_k}\nonumber \\
    &\quad \quad \cdot\left(\mathrm e^{\boldsymbol L_0 (t_k-s_{n_{k-1}}^{(k-1)})}\boldsymbol L_{\text{SE}}\mathrm e^{\boldsymbol L_0 (s_{n_{k-1}}^{(k-1)}-s_{n_{k-1}-1}^{(k-1)})}\cdots \boldsymbol L_{\text{SE}}\mathrm e^{\boldsymbol L_0 (s_1^{(k-1)}-t_{k-1})}\right) \widetilde{\boldsymbol{S}}_{\alpha_{k-1}}\cdots\nonumber \\
     &\quad \quad \cdot\left(\mathrm e^{\boldsymbol L_0 (t_2-s_{n_{1}}^{(1)})}\boldsymbol L_{\text{SE}}\mathrm e^{\boldsymbol L_0 (s_{n_{1}}^{(1)}-s_{n_{1}-1}^{(1)})}\cdots \boldsymbol L_{\text{SE}}\mathrm e^{\boldsymbol L_0 (s_1^{(1)}-t_{1})}\right) \widetilde{\boldsymbol{S}}_{\alpha_{1}}\nonumber \\
     &\quad \quad \cdot\left(\mathrm e^{\boldsymbol L_0 (t_1-s_{n_{0}}^{(0)})}\boldsymbol L_{\text{SE}}\mathrm e^{\boldsymbol L_0 (s_{n_{0}}^{(0)}-s_{n_{0}-0}^{(0)})}\cdots \boldsymbol L_{\text{SE}}\mathrm e^{\boldsymbol L_0 (s_0^{(0)})}\right)\hat\rho(0).
     \label{eq:expansion_longest}
\end{align}
Formally, this is to justify the replacement of $\mathrm e^{\boldsymbol{L}(t_i-t_{i-1})}$ with the following for all $i$:
$$\sum_{n_i=0}^{\infty}\int_{t_{i+1}\geq s_{n_{i}}^{(i)}\geq \cdots\geq s_1^{(i)}\geq t_{i}}\mathrm ds_1^{(i)}\cdots\mathrm ds_{n_{i}}^{(i)}\left(\mathrm e^{\boldsymbol L_0 (t_i-s_{n_{i-1}}^{(i-1)})}\boldsymbol L_{\text{SE}}\mathrm e^{\boldsymbol L_0 (s_{n_{i-1}}^{(i-1)}-s_{n_{i-1}-1}^{(i-1)})}\cdots \boldsymbol L_{\text{SE}}\mathrm e^{\boldsymbol L_0 (s_1^{(i-1)}-t_{i-1})}\right) .
$$
Recall that given $\hat\rho\in B_1(\mathcal H_{\text E}\otimes \mathcal H_{\text S})$, the following equality holds:
$$        \mathrm e^{\boldsymbol L(t-s)}  \hat \rho = \sum_{n=0}^\infty\int_{t\geq s_n\geq \cdots\geq s_1\geq s}\mathrm ds_1\cdots\mathrm ds_n
\mathrm e^{\boldsymbol L_0 (t-s_n)}\boldsymbol L_{\text{SE}}\mathrm e^{\boldsymbol L_0 (s_n-s_{n-1})}\cdots \boldsymbol L_{\text{SE}}\mathrm e^{\boldsymbol L_0 (s_1-s)}\hat \rho,$$
under the condition that the RHS is absolutely convergent in trace norm. Therefore the proof of \cref{eq:expansion_longest} rests on verifying the RHS is absolutely convergent under trace norm. This property is guaranteed by the estimate \cref{eq:estimate}. We refer to the following \cref{lem:bound_longest} for the details of this calculation.

\end{proof}
\begin{lem}
    Under the bounds in \cref{eq:estimate},   \cref{eq:expansion_longest}  holds for all $k\geq 0$, and all $\alpha_i = 1,\cdots,N$; moreover, the infinite series on the RHS of \cref{eq:expansion_longest} is absolutely convergent under trace norm.
\label{lem:bound_longest}    
\end{lem}
\begin{proof}[Proof of \cref{lem:bound_longest}]
We will prove this by induction. For $k=0$, \cref{eq:expansion_longest} reduces to \cref{eq:rho_perturbation_series}, for which we have established the absolute convergence using the estimate \cref{eq:estimate}. Assume that we have already shown that the lemma holds for $k = k_0-1$. In this case, it suffices to show that the RHS of \cref{eq:expansion_longest} is absolutely convergent when $k=k_0$. Let us derive the upper bound of the trace norm for the following: 
$$\begin{aligned}
   \hat\rho_{\alpha_1,\cdots,\alpha_k}&(t,t_{1:k},s_{1:n_0}^{(0)},\cdots,s_{1:n_k}^{(k)})=  \left(\mathrm e^{\boldsymbol L_0 (t-s_{n_k}^{(k)})}\boldsymbol L_{\text{SE}}\mathrm e^{\boldsymbol L_0 (s_{n_k}^{(k)}-s_{n_k-1}^{(k)})}\cdots \boldsymbol L_{\text{SE}}\mathrm e^{\boldsymbol L_0 (s_1^{(k)}-t_k)}\right)\\& \cdot \prod_{i=1}^k \left(\widetilde{\bm S}_{\alpha_i}  \mathrm e^{\boldsymbol L_0 (t_i-s_{n_{i-1}}^{(i-1)})}\boldsymbol L_{\text{SE}}\mathrm e^{\boldsymbol L_0 (s_{n_{i-1}}^{(i-1)}-s_{n_{i-1}-1}^{(i-1)})}\cdots \boldsymbol L_{\text{SE}}\mathrm e^{\boldsymbol L_0 (s_1^{(i-1)}-t_{i-1})}\right) \hat\rho(0).
\end{aligned}$$
In the fermionic case, since both $\mathcal H_{\text S}$ and $\mathcal H_{\text E}$ is finite-dimensional, $\boldsymbol{L}_{\text{SE}}$ are bounded (say by a constant $\mathfrak c$), and recall that $\|\widetilde{\boldsymbol{S}}_\alpha\|=1$, thus $$\|\hat\rho_{\alpha_1,\cdots,\alpha_k}(t,t_{1:k},s_{1:n_0}^{(0)},\cdots,s_{1:n_k}^{(k)})\|_{\text{tr}}\leq \mathfrak c^{n_0+\cdots+n_k}. $$
Then we have
$$
\|\text{ RHS of \cref{eq:expansion_longest}} \|_{\text{tr}}\leq \sum_{n_0,\cdots,n_k=0}^\infty \mathfrak c^{n_0+\cdots+n_k}\frac{(t-t_k)^{n_k}}{n_k!}\cdots \frac{(t_2-t_1)^{n_1}}{n_1!}\frac{t_1^{n_0}}{n_0!}=\prod_{i=0}^k\mathrm e^{\mathfrak c(t_{i+1}-t_{i})}<+\infty.
$$
Now let us focus on the bosonic case, in which $\widetilde{\boldsymbol{S}}_{\alpha} =\boldsymbol{1}_{\text{E}} \otimes\boldsymbol{S}_{\alpha}$. Using $\boldsymbol{L}_{\text{SE}} = \sum_{\beta=1}^N E_{\beta}\otimes S_{\beta}$,  we have:
$$
\hat\rho_{\alpha_1,\cdots,\alpha_k}(t,t_{1:k},s_{1:n_0}^{(0)},\cdots,s_{1:n_k}^{(k)})=\sum_{\beta_{1}^{(0)},\cdots \beta_{n_0}^{(0)}=1}^N\cdots \sum_{\beta_{1}^{(k)},\cdots \beta_{n_k}^{(k)}=1}^N \hat {\mathcal E}_{\beta_{1}^{(0)},\cdots \beta_{n_0}^{(0)},\cdots,\beta_{1}^{(k)},\cdots \beta_{n_k}^{(k)}}\otimes \hat {\mathcal S}_{\beta_{1}^{(0)},\cdots \beta_{n_0}^{(0)},\cdots,\beta_{1}^{(k)},\cdots \beta_{n_k}^{(k)}},
$$
$$
\begin{aligned}
\hat{\mathcal E}_{\beta_{1}^{(0)},\cdots \beta_{n_0}^{(0)},\cdots,\beta_{1}^{(k)},\cdots \beta_{n_k}^{(k)}} =  &\left(\mathrm e^{\boldsymbol L_{\text{e}} (t-s_{n_k}^{(k)})} \boldsymbol{E}_{\beta_{n_k}^{(k)}}\mathrm e^{\boldsymbol L_{\text{e}} (s_{n_k}^{(k)}-s_{n_k-1}^{(k)})}\cdots \boldsymbol{E}_{\beta_{1}^{(k)}}\mathrm e^{\boldsymbol L_{\text e} (s_1^{(k)}-t_k)}\right)\\& \cdot \prod_{i=1}^k \left(   \mathrm e^{\boldsymbol L_{\text{e}} (t_i-s_{n_{i-1}}^{(i-1)})}\boldsymbol{E}_{\beta_{n_{i-1}}^{(i-1)}}\mathrm e^{\boldsymbol L_{\text{e}} (s_{n_{i-1}}^{(i-1)}-s_{n_{i-1}-1}^{(i-1)})}\cdots \boldsymbol{E}_{\beta_{1}^{(i-1)}}\mathrm e^{\boldsymbol L_{\text e} (s_1^{(i-1)}-t_{i-1})}\right) \hat\rho(0).
\end{aligned}
$$
$$
\begin{aligned}
    \mathcal{\hat S}_{\beta_{1}^{(0)},\cdots \beta_{n_0}^{(0)},\cdots,\beta_{1}^{(k)},\cdots \beta_{n_k}^{(k)}} =  &\left(\mathrm e^{\boldsymbol L_{\text{s}} (t-s_{n_k}^{(k)})} \boldsymbol{S}_{\beta_{n_k}^{(k)}}\mathrm e^{\boldsymbol L_{\text{s}} (s_{n_k}^{(k)}-s_{n_k-1}^{(k)})}\cdots \boldsymbol{S}_{\beta_{1}^{(k)}}\mathrm e^{\boldsymbol L_{\text s} (s_1^{(k)}-t_k)}\right)\\& \cdot \prod_{i=1}^k \left(   \mathrm e^{\boldsymbol L_{\text{s}} (t_i-s_{n_{i-1}}^{(i-1)})}\boldsymbol{S}_{\beta_{n_{i-1}}^{(i-1)}}\mathrm e^{\boldsymbol L_{\text{s}} (s_{n_{i-1}}^{(i-1)}-s_{n_{i-1}-1}^{(i-1)})}\cdots \boldsymbol{S}_{\beta_{1}^{(i-1)}}\mathrm e^{\boldsymbol L_{\text s} (s_1^{(i-1)}-t_{i-1})}\right) \hat\rho(0).
\end{aligned}
$$
Since $\|\boldsymbol{S}_{\beta}\|=1$, we have $\|\mathcal{\hat S}_{\beta_{1}^{(0)},\cdots \beta_{n_0}^{(0)},\cdots,\beta_{1}^{(k)},\cdots \beta_{n_k}^{(k)}}\|_{\text{tr}}\leq 1$. Using the estimate \cref{eq:estimate}, we have the following bound:
$$
\|\hat{\mathcal E}_{\beta_{1}^{(0)},\cdots \beta_{n_0}^{(0)},\cdots,\beta_{1}^{(k)},\cdots \beta_{n_k}^{(k)}}\|_{\text{tr}}\leq \sqrt{(n_0+\cdots+n_k)!}(\mathfrak c(t))^{n_0+\cdots+n_k}.
$$
Therefore we could bound the total infinite series as follows:
$$
\begin{aligned}
    \|\text{ RHS of \cref{eq:expansion_longest}} \|_{\text{tr}}&\leq \sum_{n_0,\cdots,n_k=0}^\infty \sqrt{(n_0+\cdots+ n_k)!}\mathfrak c(t)^{n_0+\cdots+n_k}\frac{(t-t_k)^{n_k}}{n_k!}\cdots \frac{(t_2-t_1)^{n_1}}{n_1!}\frac{t_1^{n_0}}{n_0!} \\
    & = \sum_{m=0}^\infty \sqrt{m!} (\mathfrak c(t))^m \frac{((t-t_k)+\cdots +(t_2-t_1)+t_1)^m}{m!}=\sum_{m=0^\infty} \frac{(t\mathfrak c(t))^m}{\sqrt{m!}}<+\infty.
\end{aligned}
$$ 
\end{proof}

\section{Applications to spin-boson and fermionic models described by unitary and Lindblad dynamics}
\label{sec:applications}

In this section, we develop the error estimates for spin-boson and fermionic impurity models, under unitary dynamics and Lindblad dynamics, respectively. We will also discuss quasi-Lindblad pseuodomode theory in the form of \cite{ParkHuangZhuetal2024} in  \cref{sec:quasi-Lind}. The differences in terms of the structure of the Liouvillians are shown in \cref{table:comparison}. 
\begin{table}[ht]
\centering
\begin{tabular}{|c|c|c|}
\hline
               & Environment Liouvillian $\boldsymbol L_{\text e}$ & System-environment Liouvillian $\boldsymbol L_{\text{SE}}$ \\ \hline
Unitary        & $\boldsymbol L_{\text e} = \boldsymbol L_{\text e}^{\text U} = -\mathrm i[\hat H_{\text e},\cdot] $                               & $\boldsymbol L_{\text{SE}} = \boldsymbol L_{\text{SE}}^{\text U} = -\mathrm i[\hat H_{\text{SE}},\cdot] $                                                 \\ \hline
Lindblad       &      $\boldsymbol L_{\text e} = \boldsymbol L_{\text e}^{\text U} +\boldsymbol D_{\text e} $                                             &             $\boldsymbol L_{\text{SE}} = \boldsymbol L_{\text{SE}}^{\text U} = -\mathrm i[\hat H_{\text{SE}},\cdot] $                                                \\ \hline
Quasi-Lindblad &        $\boldsymbol L_{\text e} = \boldsymbol L_{\text e}^{\text U} +\boldsymbol D_{\text e} $                                                      &          $\boldsymbol L_{\text{SE}} = \boldsymbol L_{\text{SE}}^{\text U} +\boldsymbol{D}_{\text{SE}} $                                                   \\ \hline
\end{tabular}
\caption{Comparison of Liouvillian forms in unitary, Lindblad, and quasi-Lindblad dynamics. $\boldsymbol{L}_{\text e}$ is the environment Liouvillian where $\boldsymbol{L}_{\text e}^{\text U}$ and $\boldsymbol D_{\text e}$ are the environment Liouvillians from the environment Hamiltonian and environment dissipator, respectively, as defined in \cref{eq:spin_boson_env} for spin-boson models and in \cref{eq:Le_U_fermionic} and \cref{eq:D_e_fermion} for fermionic impurity models. $\boldsymbol{L}_{\text{SE}}$ is the system-environment coupling Liouvillian where $\boldsymbol L_{\text{SE}}^{\text U}$ and $\boldsymbol{D}_{\text{SE}}$ are the system-environment coupling Liouvillian from the system-environment Hamiltonian and system-environment dissipator, respectively, defined in \cref{lem:Liouvillian_unitary_spin_boson}, \cref{eq:system_env_dissipation_spin_boson} for spin-boson models and \cref{eq:L_SE_U_fermion} and \cref{eq:D_SE_fermion} for fermionic impurity models.}
\label{table:comparison}
\end{table}

This section is organized as follows. The spin-boson model and the fermionic impurity model are discussed in \cref{sec:spin_boson} and \cref{sec:fermion}, respectively. In each case, we first discuss the bosonic/fermionic environment in \cref{sec:spin_boson_env} and \cref{sec:fermion_env}, in which we give superoperator formalisms for the unitary and non-unitary environment.
Then, we describe the system-environment interaction in \cref{sec:spin_boson_env_sys} and \cref{sec:fermion_env_sys}.
Finally, we discuss the reduced system dynamics $\hat\rho_{\text S}(t)$ and the error bound in \cref{sec:spin_boson_error} and \cref{sec:fermion_error},
as a direct application of \cref{thm:rhoS_Dyson_short} and \cref{thm:main_error_bound}.

\subsection{Spin-boson model}
\label{sec:spin_boson}

\subsubsection{Bosonic environment}
\label{sec:spin_boson_env}

Let $N_{\text e}$ be the number of environment modes.
For bosonic environments, let $\hat b_k,\hat b_k^\dagger$ be the bosonic annihilation and creation operators for the $k$-th environment modes ($k =1,\cdots, N_{\text e}$), which satisfy the canonical commutation relation (CCR) $[\hat b_k,\hat b_{k'}^\dagger] = \delta_{kk'}\hat 1_{\text e}$. $\hat 1_{\text e} \in B(\mathcal H_{\text E})$ means the identity operator.
Let us introduce the following environment superoperators, which are defined by multiplication of $\hat b_k,\hat b_k^\dagger $ to an operator in $B_1(\mathcal H_{\text E})$ from the left and from the right:
\begin{defn}[Bosonic environment superoperators]
    For bosonic environment operators $\hat b_k,\hat b_k^\dagger$ ($k=1,\cdots, {N_{\text e}}$), 
    let us define the following environment superoperators, denoted as 
    $\boldsymbol b_k$, $\boldsymbol b_k^\dagger$, $\widetilde{\boldsymbol{b}}_k$, $\widetilde{\boldsymbol{b}}_k^\dagger$:
    \begin{equation}
    \boldsymbol b_k: (\cdot)\rightarrow \hat b_k(\cdot),\quad \boldsymbol b_k^\dagger: (\cdot)\rightarrow \hat b_k^\dagger(\cdot),\quad \widetilde{\boldsymbol{b}}_k: (\cdot)\rightarrow (\cdot)\hat b_k^{\dagger},\quad \widetilde{\boldsymbol{b}}_k^\dagger: (\cdot)\rightarrow (\cdot)\hat b_k.
        \label{eq:bosonic_env_superoperators}
    \end{equation}
    $\widetilde{\boldsymbol O}$ indicates the operation that multiplies operator $\hat O^\dag$ from the right. 
\label{defn:bosonic_superoperators}
\end{defn}
\begin{rem}
    Note that in  \cref{defn:bosonic_superoperators}, $\boldsymbol b_k^\dagger$ is not the conjugate of $\boldsymbol b_k$ as an operator on $B_1(\mathcal H_{\text E})$. Moreover, we define $\widetilde{\boldsymbol b}_k^{\dagger}$ to be multiplying $\hat b_k$ (not $\hat b_k^\dagger$) from the right so that the usual commutation relation holds:
$$
\left(\widetilde{\boldsymbol b}_k \widetilde{\boldsymbol{b}}_l^{\dagger} - \widetilde{\boldsymbol{b}}_l^{\dagger}\widetilde{\boldsymbol b}_k\right)\hat\rho =\hat\rho\hat b_l\hat b_k^\dagger  - \hat\rho \hat b_k^\dagger\hat b_l = \hat\rho\delta_{kl}.
$$
\end{rem}
Let us summarize the CCR in the following lemma:
\begin{lem}[Canonical commutation relations (CCR) for bosonic environment superoperators]
    The superoperators $\boldsymbol b_k,\boldsymbol b_k^\dagger,\widetilde{\boldsymbol{b}}_k,\widetilde{\boldsymbol{b}}_k^\dagger$ ($k,k' = 1,\cdots, {N_{\text e}}$) satisfy the following canonical commutation relations:
\begin{equation}
    \begin{aligned}
    \left[\boldsymbol{b}_k, \boldsymbol{b}_{k'}^{\dagger}\right] &=  \delta_{kk'}\boldsymbol 1_{\text e},
    \quad \left[\widetilde{\boldsymbol{b}}_k, \widetilde{\boldsymbol{b}}_{k'}^{\dagger}\right] =  \delta_{kk'}\boldsymbol 1_{\text e},\quad 
    \\
    \left[\boldsymbol{b}_k, {\boldsymbol{b}}_{k'}\right]  = 
    \left[\boldsymbol{b}_k^{\dagger},{\boldsymbol{b}}_{k'}^{\dagger}\right]  = 
    \left[\widetilde{\boldsymbol{b}}_k^{\dagger},\widetilde{\boldsymbol{b}}_{k'}^{\dagger}\right] & = 
    \left[\widetilde{\boldsymbol{b}}_k,\widetilde{\boldsymbol{b}}_{k'}\right]  = 
    \left[\boldsymbol{b}_k, \widetilde{\boldsymbol{b}}_{k'}\right]  = 
    \left[\boldsymbol{b}_k, \widetilde{\boldsymbol{b}}_{k'}^{\dagger}\right]  = 
    \left[\boldsymbol{b}_k^{\dagger}, \widetilde{\boldsymbol{b}}_{k'}^{\dagger}\right]  = 0.
    \end{aligned}
    \label{eq:canonical}
\end{equation}
\label{lem:CCR}
\end{lem}
In this paper, we will use the following two bosonic Liouvillians: a Liouvillian $\boldsymbol L_{\text e}^{\text U}$ from environment Hamiltonian and a Liouvillian $\boldsymbol{D}_{\text e}$ from the environment Lindblad dissipator:
$$
\boldsymbol L_{\text e}^{\text U}\hat\rho = -\mathrm i [\hat H_{\text e}, \hat\rho],\quad 
\boldsymbol{D}_{\text e} \hat\rho = \sum_{k, l=1}^{N_e} \Gamma_{k l}\left(2 \hat{b}_l \hat{\rho}\hat b_k^{\dagger}-\hat{b}_k^{\dagger} \hat{b}_l \hat{\rho}-\hat{\rho} \hat{b}_k^{\dagger} \hat{b}_l\right),
$$
where $\hat H_{\text e}=\sum_{k, l=1}^{N_{\text e}} H_{k l}\hat b_k^\dagger\hat b_l$ is the environment Hamiltonian. $H$ is a $N_{\mathrm{e}}  \times N_{\mathrm{e}}$ Hermitian matrix and $\Gamma$ is a $N_{\mathrm{e}}  \times N_{\mathrm{e}}$ positive semi-definite matrix. In the general case, $H$
and $\Gamma$ are not diagonal matrices \cite{lednev2024lindblad}.
In the superoperator representation, $\boldsymbol L_{\text e}^{\text U}$ and $\boldsymbol{D}_{\text e}$ are expressed as
\begin{equation}
    \boldsymbol{L}_{\text e}^{\text U} = \sum_{k, l=1}^{N_{\text e}}-\mathrm{i} H_{k l} \boldsymbol{b}_k^{\dagger} \boldsymbol{b}_l+\mathrm{i} H_{k l} \widetilde{\boldsymbol{b}}_l^{\dagger} \widetilde{\boldsymbol{b}}_k
,\quad
\boldsymbol{D}_{\mathrm{e}}=\sum_{k,l=1}^{N_{\mathrm{e}}} \Gamma_{k l}\left(2 \widetilde{\boldsymbol{b}}_k \boldsymbol{b}_l-\boldsymbol{b}_k^{\dagger} \boldsymbol{b}_l-\widetilde{\boldsymbol{b}}_l^{\dagger} \widetilde{\boldsymbol{b}}_k\right).
\label{eq:spin_boson_env}
\end{equation}

The goal of the pseudomode theory is to simulate the reduced density  dynamics
by replacing the unitary environment with the non-unitary pseudomode environment. The Liouvillian form for the unitary environment is given by $\boldsymbol{L}_{\text e} = \boldsymbol{L}_{\text e}^{\text U}$. For the unitary dynamics, we consider the initial environment density operator $\hat\rho_{\text{E}}(0)$ given by a Gibbs state:
\begin{equation}
    \hat\rho_{\text{E}}(0) =\frac{1}{Z_{\text e}} \mathrm e^{-\beta\hat H_{\text e}},\quad Z_{\text e} = \operatorname{tr}(\mathrm e^{-\beta\hat H_{\text e}}), 
    \label{eq:rhoE_init}
\end{equation}
for some inverse temperature $\beta>0$. For the pseudomode environment, the Liouvillian form is given by $\boldsymbol{L}_{\text e} = \boldsymbol{L}_{\text e}^{\text U} + \boldsymbol{D}_{\text e}$ and we follow the initial density operator setting in \cite{TamascelliSmirneHuelga2018, Mascherpa2020, ParkHuangZhuetal2024} to consider also the initial vacuum environment:
\begin{equation}
    \hat\rho_{\text{E}}(0) = |0\rangle_{\text E}\langle 0|_{\text{E}}, 
    \label{eq:rhoE_init_empty}
\end{equation}
where $|0\rangle_{\text E}$ is the environment vacuum state. 
Adapting the following analysis to different initial pseudomode density operators (e.g., in \cite{Lotem2020, Brenes2020, Zwolak2020}) is straightforward. Note that both \cref{eq:rhoE_init} and \cref{eq:rhoE_init_empty} satisfy $\boldsymbol{L}_{\text e}^{\text U}\hat\rho_{\text{E}}(0)=0$, and \cref{eq:rhoE_init_empty} satisfies that $\boldsymbol{D}_{\text e} \hat\rho_{\text{E}}(0)=0$. As a result,  as discussed in \cref{rmk:single_variable_Corr}, the two-point BCFs $C_{\alpha\alpha'}(t,t')$ rely only on $(t-t')$ for both unitary and pseudomode cases.

\subsubsection{System-environment coupling}
\label{sec:spin_boson_env_sys}
In the spin-boson model, the system is made up of $n$ spins. The system Hilbert space is $\mathcal H_{\text S} = \otimes_{i=1}^n \mathbb C^2$.
Let $\hat\sigma_x^{(j)},\hat\sigma_y^{(j)},\hat\sigma_z^{(j)}\in B(\mathcal H_{\text S})$  be Pauli matrices acting on the $j$-th spin, i.e.,
    $$
    \begin{aligned}
        \hat \sigma_{a}^{(j)} &= (\otimes_{i=1}^{j-1}\hat \sigma_0) \otimes \hat \sigma_a \otimes (\otimes_{i=j+1}^n\hat \sigma_0),\quad a\in\{x,y,z\}. \\
        \hat\sigma_0 = \begin{pmatrix}
            1 & 0\\
            0 & 1
        \end{pmatrix},&\quad \hat\sigma_x = \begin{pmatrix}
            0 & 1\\
            1 & 0
        \end{pmatrix},\quad \hat\sigma_y = \begin{pmatrix}
            0 & -\mathrm i\\
            \mathrm i & 0
        \end{pmatrix},\quad \hat\sigma_z = \begin{pmatrix}
            1 & 0\\
            0 & -1
        \end{pmatrix}.
    \end{aligned}
    $$
Similar to \cref{defn:bosonic_superoperators},  for $\hat\sigma_0^{(j)}$ and $\hat\sigma_{x,y,z}^{(j)}$, let us define system superoperators $\boldsymbol\sigma_{a}^{(j)}$, $\widetilde{\boldsymbol\sigma}_{a}^{(j)} $ for $a = 0,x,y,z$: 
\begin{equation}
\boldsymbol\sigma_{a}^{(j)}:(\cdot)\rightarrow \hat\sigma_{a}^{(j)}(\cdot),\quad \widetilde{\boldsymbol\sigma}_{a}^{(j)}:(\cdot)\rightarrow (\cdot)\hat\sigma_{a}^{(j)}.
\end{equation}
For $j=1,\cdots,n$, let us define $ \boldsymbol S_{2j-1} = \boldsymbol\sigma_z^{(j)}$, $\boldsymbol S_{2j} = \widetilde{\boldsymbol\sigma}_z^{(j)}$.
We consider a system-environment coupling Liouvillian $\boldsymbol L_{\text{SE}}^{\text U} $ given by $\boldsymbol L_{\text{SE}}^{\text U} = -\mathrm i [\hat H_{\text{SE}}, \cdot]$, where $\hat H_{\text{SE}}$ is the system-environment Hamiltonian $\hat H_{\text{SE}} = \sum_{j=1}^n \sum_{k=1}^{N_{\mathrm{e}}} \left(g_{jk}\hat b_k+g_{jk}^*\hat b_k^{\dagger}\right)\otimes \hat\sigma_z^{(j)}$. $g$ is a $n\times N_{\mathrm{e}}$ complex-valued matrix and is the coupling coefficient matrix for the system-environment coupling. 
We decompose the coupling Liouvillian $\boldsymbol L_{\text{SE}}^{\text{U}}$ in terms of system and environment superoperator:
 \begin{equation}
        \begin{aligned}
        \boldsymbol L_{\text{SE}}^{\text{U}} &= \sum_{\alpha=1}^{2n}\boldsymbol E_{\alpha}^{\text U}\otimes \boldsymbol S_{\alpha},
    \quad   \boldsymbol E^{\text U}_{2j-1} = -\mathrm i\sum_{k=1}^{N_{\text e}} (g_{jk}\boldsymbol b_k+g_{jk}^*\boldsymbol b_k^\dagger),\quad \boldsymbol E^{\text U}_{2j} = \mathrm i\sum_{k=1}^{N_{\text e}} (g_{jk}^*\widetilde{\boldsymbol b}_k+g_{jk}\widetilde{\boldsymbol b}_k^\dagger).
    \end{aligned}     
    \label{lem:Liouvillian_unitary_spin_boson}
\end{equation}
To establish the growth estimate described in \cref{eq:estimate}, we require the following condition on matrices $H$ and $g$, which is a natural physical condition \cite{Davies1981symmetry} known as {\emph{infrared regularity}} :
\begin{assumption}[Discrete infrared regularity condition]
    Assume that all rows of matrix $g$ are non-zero. 
    For the  matrix $H$ in environment Hamiltonian $\hat H_{\text e} = \sum_{k,l=1}^{N_{\text e}} H_{kl}\hat b_k^\dagger \hat b_l$ and system-environment coupling $g$ defined in \cref{sec:spin_boson_env_sys}, there exists a constant $c_{\text J}>0$ such that
    \begin{equation}
        c_{\text J} H\succeq g^\dagger g,
        \label{eq:infrared_regular}
    \end{equation}
    Here we use the notation $A\succeq B$ to mean that $A-B$ is a positive semi-definite matrix. 
    \label{assump:infrared_regular}
\end{assumption}

\begin{example}
   Let us use the example of a discretized Ohmic bath coupled to a single spin to illustrate the infrared regularity condition. For simplicity, let us consider a uniform bath discretization, though similar results hold for discretization based on other quadrature rules. Let $\omega_k = \frac{k}{N_{\text{e}}}\omega_{\text M}$ for $k = 1,\cdots, N_{\text{e}}$. Here $\omega_{\text M}>0$ is the maximum cutoff frequency. Then $g_k$ is given by $g_k = \sqrt{\frac{\omega_{\text M}}{N_{\text e}}}\cdot\sqrt{\omega_k \mathrm e^{-\omega_k/\omega_{\text c}}}$, where $\omega_{\text c}>0$ is the characteristic frequency. This is because $|g_k|^2$ comes from discretizing the Ohmic spectral density $J(\omega) = \omega \mathrm e^{-\omega/\omega_{\text c}}$ on the interval $[0,\omega_{\text M}]$, and the factor $\frac{\omega_{\text M}}{N_{\text e}}$ is the corresponding quadrature weight.
   
   In this case, the environment Hamiltonian matrix $H$ is a diagonal matrix with $H_{kk} = \omega_k$, and $g^\dagger g$ is also a diagonal matrix with $(g^\dagger g)_{kk} = |g_k|^2 = \frac{\omega_{\max}}{N_{\text e}}\cdot \omega_k \mathrm e^{-\omega_k/\omega_{\text c}}$. Therefore, we have
   \begin{equation}
    c_{\text J} H \succeq g^\dagger g,\quad c_{\text J} = \frac{\omega_{\max}}{N_{\text e}}.
   \end{equation}
   Note that the $\frac{1}{N_{\text e}}$ scaling of $c_{\text J}$ is a direct consequence of the quadrature weight from discretization, and does not rely on the Ohmic nature of the bath.
   \label{example:infrared}
\end{example}

\subsubsection{Growth estimate of environment correlators}
\label{sec:estimate}
Let us recall the environment superoperator correlators defined in \cref{eq:env_correlator}:
$$
\hat{\mathcal E}_{\alpha_1,\cdots,\alpha_n}(t,t_1,\cdots,t_n)= \mathrm e^{\boldsymbol{L}_{\text e}(t-t_1)}\boldsymbol{E}_{\alpha_1}\mathrm e^{\boldsymbol{L}_{\text e}(t_1-t_2)}\cdots \boldsymbol{E}_{\alpha_n}\mathrm e^{\boldsymbol{L}_{\text e}t_n}\hat\rho_{\text E}(0) .
$$
We now establish the growth estimate described in \cref{eq:estimate}.  Let us start with the following useful inequality for a single bosonic mode:
\begin{proposition}
    Let $\hat E$ be a monomial in $\hat b$, $\hat b^\dagger$ of total degree $n\in \mathbb N$. Let $|m\rangle$ be the $m$-particle state, i.e., $\hat b^\dagger \hat b |m\rangle = m|m\rangle$. Then for all $m\geq 0$, we have
    \begin{equation}
        \|\hat E |m\rangle \|^2 \leq \frac{(m+n)!}{m!}.
    \end{equation}
    \label{prop:single}
\end{proposition}
\begin{proof}
We will use induction on $n$. For $n=1$, $\|\hat b|m\rangle\rangle\|^2 = m$ and $\|\hat b^\dagger |m\rangle\|^2 = m+1$, so the proposition holds. Now suppose the proposition holds for degree $n-1$. Let $\hat E = \hat B\hat E'$, where $\hat B$ is either $\hat b$ or $\hat b^\dagger$ and $\hat E'$ is of degree $(n-1)$.  Let $|\psi\rangle = \hat E'|m\rangle$. By the induction hypothesis, we have $\|\psi\|^2 \leq (m+n-1)!/m!$. 

Let $m'$ be the particle number of $|\psi\rangle$. Since $\hat E'$ is of degree $(n-1)$, thus $m'\leq m+n-1$. If $\hat B = \hat b$, then $\|\hat E|\psi\rangle\|^2 = m' \|\psi\|^2$. If $\hat B = \hat b^\dagger$, then $\|\hat E|\psi\rangle\|^2 = (m'+1)\|\psi\|^2$. In both cases, we have $\|\hat E|\psi\rangle\|^2 \leq (m'+1)\|\psi\|^2 \leq (m+n) \|\psi\|^2 \leq (m+n)(m+n-1)!/m! = (m+n)!/m!$. This completes the induction.
\end{proof}
For $\omega>0$,   the thermal state is given by $\hat\rho^{\text{th} }_{\beta,\omega} = (1-\mathrm e^{-\beta \omega})\sum_{m=0}^\infty \mathrm e^{-m\beta \omega} |m\rangle\langle m|$. Here $\beta \in (0,+\infty]$. For $\beta=+\infty$, $\hat\rho^{\text{th} }(\omega) = |0\rangle\langle 0|$ is the vacuum state. Let us also define the thermal factor $c_{\beta}(\omega) = (1-\mathrm e^{-\beta \omega})^{-1/2}$. For $\beta = +\infty$, we have $c_{\beta}(\omega) = 1$. Using \cref{prop:single}, we could establish the following  estimate:
\begin{corollary}
    Let $\hat E$ be a monomial in $\hat b$, $\hat b^\dagger$ of total degree $n\in \mathbb N$.  Then we have
    \begin{equation}
        \operatorname{tr}(\hat E^\dagger \hat E \hat\rho^{\text{th}}_{\beta,\omega}) \leq n! c_{\beta}(\omega)^{2n}.
    \end{equation}
    \label{cor:thermal}
\end{corollary}
\begin{proof}
    Using \cref{prop:single}, we have
    $$
    \begin{aligned}
        \operatorname{tr}(\hat E^\dagger \hat E \hat\rho^{\text{th}}_{\beta,\omega}) & = (1-\mathrm e^{-\beta \omega})\sum_{m=0}^\infty \mathrm e^{-m\beta \omega} \| \hat E |m\rangle\|^2 \leq (1-\mathrm e^{-\beta \omega})\sum_{m=0}^\infty \mathrm e^{-m\beta \omega} \frac{(m+n)!}{m!} 
    \end{aligned}
    $$
    Using the binomial identity $\sum_{m=0} ^\infty \binom{m+n}{n} x^m = (1-x)^{-(n+1)}$ for $|x|<1$, the above sum is equal to $n!  (1-\mathrm e^{-\beta \omega})^{-n} = n! c_{\beta}(\omega)^{2n}$.
\end{proof}
Now let us establish the following inequality for superoperators acting on the thermal state:
\begin{corollary}
    Let $\boldsymbol B_{1},\cdots, \boldsymbol B_{n} \in \{\boldsymbol{b} , \boldsymbol{b}^\dagger, \widetilde{\boldsymbol{b}}, \widetilde{\boldsymbol{b}}^\dagger\}$ for some $n\in\mathbb N$, where $\boldsymbol b$, $\boldsymbol b^\dagger$, $\widetilde{\boldsymbol{b}}$, $\widetilde{\boldsymbol{b}}^\dagger$ are defined in \cref{eq:bosonic_env_superoperators} meaning multiplications of $\hat b$, $\hat b^\dagger$ from the left and from the right. Then we have
    \begin{equation}
        \|\boldsymbol B_{1}\cdots \boldsymbol B_{n} (\hat\rho^{\text{th}}(\omega))\|_{\text{tr}} \leq \sqrt{n!} c_{\beta}(\omega)^{n}.
    \end{equation}
    \label{cor:thermal_growth}
\end{corollary}
\begin{proof}
    The resulting operator $\boldsymbol B_{1}\cdots \boldsymbol B_{n} (\hat\rho^{\text{th}}(\omega))$ could be written as $  (\hat B_{\mathrm L} (\hat\rho^{\text{th}}(\omega)) \hat B_{\mathrm R})$, where $\hat B_{\mathrm L}$ and $\hat B_{\mathrm R}$ are monomials in $\hat b$ and $\hat b^\dagger$ corresponding to the left and right multiplications. Let $n_{\mathrm L}$ and $n_{\mathrm R}$ be their degrees, respectively. Then $ n_{\mathrm L} + n_{\mathrm R} = n$.  

    Let $Y = \hat B_{\mathrm L} (\hat\rho^{\text{th}}(\omega))^{1/2}$ and $Z = (\hat\rho^{\text{th}}(\omega))^{1/2} \hat B_{\mathrm R}$. We calculate the Hilbert-Schmidt norms of $Y$ and $Z$. We have $\|Y\|_{\text{HS}}^2 = \operatorname{tr}((\hat\rho^{\text{th}}(\omega))^{1/2} \hat B_{\mathrm L}^{\dagger} \hat B_{\mathrm L} (\hat \rho^{\text{th}}(\omega))^{1/2})  = \operatorname{tr}(\hat B_{\mathrm L}^{\dagger} \hat B_{\mathrm L} \hat\rho^{\text{th}}(\omega))$. By \cref{cor:thermal}, this is bounded by $n_{\mathrm L}! c_{\beta}(\omega)^{2 n_{\mathrm L}}$. Similarly, we have $\|Z\|_{\text{HS}}^2 \leq n_{\mathrm R}! c_{\beta}(\omega)^{2 n_{\mathrm R}}$. 

    By Cauchy-Schwarz inequality $\|YZ\|_{\text{tr}} \leq \|Y\|_{\text{HS}} \|Z\|_{\text{HS}}$, we have
    $$
    \begin{aligned}
        \|\boldsymbol B_{1}\cdots \boldsymbol B_{n} (\hat\rho^{\text{th}}(\omega))\|_{\text{tr}} & = \|YZ\|_{\text{tr}} \leq \|Y\|_{\text{HS}} \|Z\|_{\text{HS}} \leq \sqrt{n_{\mathrm L}! n_{\mathrm R}!} c_{\beta}(\omega)^{n_{\mathrm L}+ n_{\mathrm R}} \leq \sqrt{n!} c_{\beta}(\omega)^{n}.
    \end{aligned}
    $$
\end{proof}

Without loss of generality, let us assume that the environment Hamiltonian matrix $H$ is diagonal, i.e., $H_{kl} = \omega_k \delta_{kl}$. Otherwise, one can always diagonalize $H$ with $\tilde H = U^\dagger H U$ and changing the coupling matrix accordingly as $\tilde g = g U$. 
We also note that for a dissipative bosonic environment, we can assume that $H$ is diagonal; however, we cannot assume that both $H$ and $\Gamma$ are diagonal simultaneously. With the diagonal assumption of $H$, the initial environment density operator \cref{eq:rhoE_init} could be written as
\begin{equation}
    \hat\rho_{\text E}(0) = \bigotimes_{k=1}^{N_{\text e}} \left(1-\mathrm e^{-\beta \omega_k}\right) \sum_{m=0}^\infty \mathrm e^{-m\beta \omega_k} |m\rangle_k \langle m|_k,
    \label{eq:rhoE_init_diagonal}
\end{equation}
 Note that from the infrared regularity condition \cref{eq:infrared_regular}, we have $c_{\text J}\omega_k \geq (g^\dagger g)_{kk}  = (\sum_{j=1}^{N_{\text s} } |g_{jk}|^2)$. Thus all $\omega_k$'s are strictly positive and $c_{\beta}(\omega_k)$ is well-defined. 
 Then we have the following multi-mode growth estimate:
 \begin{corollary}
    Let $\boldsymbol B_{1},\cdots, \boldsymbol B_{n} \in \{\boldsymbol{b}_k, \boldsymbol{b}_k^\dagger, \widetilde{\boldsymbol{b}}_k, \widetilde{\boldsymbol{b}}_k^\dagger\}_{k=1}^{N_{\text e}}$ for some $n\in\mathbb N$. With the initial environment density operator given by \cref{eq:rhoE_init_diagonal}, we have
    \begin{equation}
        \|\boldsymbol B_{1}\cdots \boldsymbol B_{n}  \hat\rho_{\text E}(0) \|_{\text{tr}} \leq \prod_{k=1}^n \sqrt{n_k!} c_{\beta}(\omega_k)^{n_k},
    \end{equation}
    where $n_k$ is the total number of superoperators in $\{\boldsymbol B_{1},\cdots, \boldsymbol B_{n}\}$ that belong to $\{\boldsymbol{b}_k, \boldsymbol{b}_k^\dagger, \widetilde{\boldsymbol{b}}_k, \widetilde{\boldsymbol{b}}_k^\dagger\}$.
    \label{cor:multi_mode_growth}
\end{corollary}
\begin{proof}
For $\boldsymbol B_{1},\ldots,\boldsymbol B_{n}\in \{\boldsymbol b_k,\boldsymbol b_k^\dagger,\widetilde{\boldsymbol b}_k,\widetilde{\boldsymbol b}_k^\dagger\}_{k=1}^{N_{\text e}}$, we can rewrite $\boldsymbol B_{ 1}\cdots \boldsymbol B_{ n} (\hat\rho_{\text E}(0))$ as $\bigotimes_{k=1}^{N_{\text e}} (\boldsymbol{B}^{(k)} (\hat\rho^{k,\text{th}}_{\beta,\omega}))$, where $\boldsymbol B^{(k)}$ is a monomial in $\boldsymbol b_k$, $\boldsymbol b_k^\dagger$, $\widetilde{\boldsymbol b}_k$, $\widetilde{\boldsymbol b}_k^\dagger$ of degree $n_k$ with $\sum_{k=1}^{N_{\text e}} n_k = n$, and $\hat\rho^{k,\text{th}}_{\beta,\omega}$ is the thermal state for the $k$-th mode. Then using the single-mode estimate in \cref{cor:thermal_growth}, we have 
$\|\boldsymbol B^{(k)} (\hat\rho^{k,\text{th}}_{\beta,\omega})\|_{\text{tr}} \leq \sqrt{n_k!} c_{\beta}(\omega_k)^{n_k}$. Multiplying these estimates for all $k=1,\cdots, N_{\text e}$ gives the desired result.
\end{proof}

We will use that $\boldsymbol E_{\alpha}$'s are linear combinations of the bosonic environment superoperators with bounded coefficients. Let 
$\boldsymbol b_{k,1} = \boldsymbol b_k$, $\boldsymbol b_{k,2} = \boldsymbol b_k^\dagger$, $\boldsymbol b_{k,3} = \widetilde{\boldsymbol b}_k$, $\boldsymbol b_{k,4} = \widetilde{\boldsymbol b}_k^\dagger$, then $\boldsymbol E_{\alpha} $ is of the following form
\begin{equation}
    \boldsymbol E_{\alpha} = \sum_{k=1}^{N_{\text e}} \sum_{s=1}^4 p_{\alpha, (k,s)} \boldsymbol b_{k,s} , \quad \alpha = 1,\cdots, N,\quad  |p_{\alpha, (k,s)}| < +\infty.
    \label{eq:E_alpha_form}
\end{equation}
To capture the system-environment coupling strength for bath mode $k$, let us define $C_{k}^{\mathrm{int}}$ ($k=1,\cdots, N_{\text e}$)   as
\begin{equation}
    C_{k}^{\mathrm{int}} = \left(\sum_{\alpha=1}^N\sum_{s=1}^4 |p_{\alpha, (k,s)}|^2\right)^{1/2}.
    \label{eq:coupling_strength}
\end{equation}
For the spin-boson model described in \cref{sec:spin_boson_env_sys}, we have $(C^{\text{int}}_k)^2 = 2\sum_{j=1}^n |g_{jk}|^2 = 2 (g^\dagger g)_{kk}$. Using the infrared regularity condition $c_{\text J} H \succeq g^\dagger g$ (see \cref{eq:infrared_regular}) and the diagonal assumption $H_{kk'} = \omega_k \delta_{kk'}$, we have 
\begin{equation}
    C^{\text{int}}_k \leq \sqrt{2 c_{\text J} \omega_k}.\quad \text{for } k=1,\cdots, N_{\text e}. \quad \text{(from infrared regularity condition)}
    \label{eq:infrared_here}
\end{equation}
 
\begin{lem}[Growth estimate for $\hat{\mathcal E}_{\alpha_1,\cdots,\alpha_n}(t,t_1,\cdots,t_n)$]
Consider the  unitary or pseudomode spin-boson dynamics in \cref{sec:spin_boson} with initial state \cref{eq:rhoE_init_diagonal} and the infrared regularity condition in \cref{assump:infrared_regular}.
Then there exists a constant $\mathfrak c$ (for unitary environment, depending on $N_{\text e}$, $c_{\text J}$, $\beta$, $\omega_{\text M}$, and for dissipative environment, further depending on $t$) such that for any $n\in\mathbb N$ and any choice of $\alpha_1, \cdots, \alpha_n\in\{1,\cdots, N\}$, we have the following estimate for $\hat{\mathcal E}_{\alpha_1,\cdots,\alpha_n}(t,t_1,\cdots,t_n) = \mathrm e^{\boldsymbol{L}_{\text e}(t-t_1)}\boldsymbol E_{\alpha_1}\mathrm e^{\boldsymbol{L}_{\text e}(t_1-t_2)}\cdots \boldsymbol E_{\alpha_n} \mathrm e^{\boldsymbol{L}_{\text e}t_n}\hat\rho_{\text E}(0)$: 
\begin{equation}
    \|\hat{\mathcal E}_{\alpha_1,\cdots,\alpha_n}(t,t_1,\cdots,t_n)\|_{\text{tr}} \leq \sqrt{n!}\mathfrak c^n .
    \label{eq:estimate_vacuum}
\end{equation}
    \label{lem:boson_growth_estimate}
\end{lem}
We will show that for unitary environment, the constant $\mathfrak c$ could be chosen as follows for all $t\geq 0$:
\begin{equation}
    \mathfrak c  = 2\sqrt 2 N_{\text e} c_{\text J} \cdot\max\{\sqrt{\omega_{\text M}}, \beta^{-1/2}\}, 
    \label{eq:boson_growth_constant}
\end{equation}
 where $N_{\text e}$ is the number of environment modes, $c_{\text J}$ is the constant in the infrared regularity condition \cref{eq:infrared_regular}, $\omega_{\text M} = \max_{k=1}^{N_{\text e}} \omega_k$ is the maximum environment frequency. For dissipative environment, there is a further exponential time dependence in $\mathfrak c$. 

We also remark that in \cref{eq:boson_growth_constant}, there is a  linear dependence on $N_{\text e}$. However, as shown in \cref{example:infrared}, $c_{\text J}$ scales as $1/N_{\text e}$ when the environment is discretized from a continuous spectral density. Thus the overall constant $\mathfrak c$ is independent of $N_{\text e}$ in such cases.

Let us first prove the growth estimate for $t=t_1=\cdots=t_n=0$. We will address the time dependence later. 
\begin{proof}[Proof of \cref{lem:boson_growth_estimate} for $t=t_1=\cdots=t_n=0$.]Our goal is to show that $\|\boldsymbol E_{\alpha_1}\cdots \boldsymbol E_{\alpha_n} \hat\rho_{\text E}(0)\|_{\text{tr}} \leq \sqrt{n!}\mathfrak c^n$, for all $n\in\mathbb N$ and any choice of $\alpha_1,\cdots,\alpha_n$. Since $\boldsymbol E_{\alpha}$'s are linear combinations of $\boldsymbol b_k$, $\boldsymbol b_k^\dagger$, $\widetilde{\boldsymbol b}_k$, $\widetilde{\boldsymbol b}_k^\dagger$, namely $\boldsymbol{E}_{\alpha} = \sum_{k=1}^{N_{\text e}} \sum_{s=1}^4 c^s_{\alpha,k} \boldsymbol b_{k,s} $, we have
$$
\begin{aligned}
    \|
\boldsymbol E_{\alpha_1}\cdots \boldsymbol E_{\alpha_n}( \hat\rho_{\text E}(0))
\|_{\text{tr}}&\leq
  \sum_{k_1,\cdots,k_n=1}^{N_{\text e}} \sum_{s_1,\cdots,s_n=1}^4 |p_{\alpha_1,(k_1,s_1)}| \cdots |p_{\alpha_n,(k_n,s_n)}| \|\boldsymbol b_{k_1,s_1} \cdots \boldsymbol b_{k_n,s_n} (\hat\rho_{\text E}(0))\|_{\text{tr}}\\
  &\leq 2^n \sum_{k_1,\cdots,k_n=1}^{N_{\text e}} C_{k_1}^{\mathrm{int}}\cdots C_{k_n}^{\mathrm{int}} \max_{  s_1,\cdots,s_n} \|\boldsymbol b_{k_1,s_1} \cdots \boldsymbol b_{k_n,s_n} (\hat\rho_{\text E}(0))\|_{\text{tr}}
  \end{aligned}
$$
Here we have used that $\sum_{s=1}^4|p_{\alpha,(k,s)}| \leq 4^{\frac 1 2}(\sum_{s'=1}^4  |p_{\alpha,(k,s')}|^2)^{1/2} \leq 2(\sum_{s'=1}^4\sum_{\alpha'=1}^N |p_{\alpha',(k,s')}|^2)^{1/2} = 2 C_k^{\mathrm{int}}$; where the first inequality holds by Cauchy-Schwarz. Now let us bound each term in the above sum. For any fixed $k_1,\cdots,k_n$, let $n_k$ be the number of times that index $k$ appears in the list $k_1,\cdots,k_n$. Then $\sum_{k=1}^{N_{\text e}} n_k = n$. Using the multi-mode estimate in \cref{cor:multi_mode_growth}, we have
$$
\begin{aligned}
C_{k_1}^{\mathrm{int}}\cdots C_{k_n}^{\mathrm{int}} \max_{  s_1,\cdots,s_n} \|\boldsymbol b_{k_1,s_1} \cdots \boldsymbol b_{k_n,s_n} (\hat\rho_{\text E}(0))\|_{\text{tr}} & \leq   \prod_{k=1}^{N_{\text e}} (C_k^{\mathrm{int}})^{n_k} \sqrt{n_k!} c_{\beta}(\omega_k)^{n_k}
\end{aligned}
$$
Using that $C_k^{\mathrm{int}} \leq \sqrt{2 c_{\text J} \omega_k}$ from the infrared regularity condition,  we have
$$
\begin{aligned}
 \prod_{k=1}^{N_{\text e}} (C_k)^{n_k} \sqrt{n_k!} c_{\beta}(\omega_k)^{n_k} \leq \prod_{k=1}^{N_{\text e}} \sqrt{n_k!} (\sqrt{2 c_{\text J}\omega_k} c_{\beta}(\omega_k))^{n_k} .
\end{aligned}
$$
Using $\sqrt{\omega_k} c_{\beta}(\omega_k) \leq \max\{\sqrt{\omega_{\text M}}, \beta^{-1/2}\} =:C_{\beta, \omega_{\text M}}$, where $\omega_{\text M} = \max_k \omega_k$, we have 
$$\prod_{k=1}^{N_{\text e}} \sqrt{n_k!} (\sqrt{2 c_{\text J} \omega_k}c_{\beta}(\omega_k))^{n_k} \leq \prod_{k=1}^{N_{\text e}} \sqrt{n_k!} (\sqrt{2c_{\text J}} C_{\beta, \omega_M})^{n_k}\leq \sqrt{n!} (\sqrt{2c_{\text J}} C_{\beta, \omega_M})^{n}.$$
Combining all the above estimates, we have 
$$
\begin{aligned}
    \|
\boldsymbol E_{\alpha_1}\cdots \boldsymbol E_{\alpha_n}( \hat\rho_{\text E}(0))
\|_{\text{tr}}&\leq \sqrt{n!}\left( 2\sqrt 2N_{\text e} c_{\text J} C_{\beta, \omega_M} \right)^n.
\end{aligned}
$$
\end{proof}
Next we prove the estimate for general $t\geq t_1\geq t_2\cdots\geq t_n\geq 0$. To address the time dependence in \cref{lem:boson_growth_estimate}, it suffices to calculate the partial derivatives with respect to  $t_1$, $\cdots$, $t_n$ respectively. This involves the commutators between $\boldsymbol{E}_{\alpha_i}$ and $\boldsymbol{L}_{\text e}$, which we derived in Lemma \ref{lem:commutator_L} and Lemma \ref{lem:commutator_L_D} in \cref{sec:env_correlation}.

\begin{proof}[Proof of time dependence in \cref{lem:boson_growth_estimate}]
Note that the partial derivative of $\mathcal E_{\alpha_1,\cdots,\alpha_n}(t,t_1,\cdots, t_n) $ with respect to $t_i$ (for $i=1,\cdots,n$), is a linear combination in the following form:
$$
\sum_{\beta_i} c_{\beta_i}\mathcal E_{\alpha_1,\cdots,\beta_i,\cdots, \alpha_n}(t,t_1,\cdots, t_n)
$$
This is because if taking derivative with respect to $t_i$, we have  $\mathrm e^{\boldsymbol{L}_{\text e}(t-t_1)} \boldsymbol{E}_{\alpha_1} e^{\boldsymbol{L}_{\text e}(t_1-t_2)} \cdots \boldsymbol{E}_{\alpha_n} e^{\boldsymbol{L}_{\text e} t_n }\hat\rho_{\text E}(0)$ with $\boldsymbol{E}_{\alpha_i}$ replaced by $[\boldsymbol{E}_{\alpha_i}, \boldsymbol{L}_{\text e}]$. This is because $\boldsymbol{E}_{\alpha}$ is linear in  $\{\boldsymbol{b}_k, \boldsymbol{b}_k^\dagger, \widetilde{\boldsymbol{b}}_k,\widetilde{\boldsymbol{b}}_k^\dagger\}$,  $\boldsymbol{L}_{\text e}$ is quadratic, and $\{\boldsymbol{b}_k, \boldsymbol{b}_k^\dagger, \widetilde{\boldsymbol{b}}_k,\widetilde{\boldsymbol{b}}_k^\dagger\}$ satisfies bosonic commutation relations, the commutator $[\boldsymbol{E}_{\alpha_i}, \boldsymbol{L}_{\text e}]$ is also the linear combination of $\{\boldsymbol{b}_k, \boldsymbol{b}_k^\dagger, \widetilde{\boldsymbol{b}}_k,\widetilde{\boldsymbol{b}}_k^\dagger\}$ (see \cref{eq:commutator_Le_H}, \cref{eq:commutator_Le_D}).
Then using \cref{eq:commutator_Le_H}, \cref{eq:commutator_Le_D}, the dynamics in $t_i$ is controlled by a system of ODEs with the following matrix: 
$$
K = p K_0 p^{+},\quad K_0=\begin{pmatrix}
    -\mathrm i H - \Gamma & 0 & 0 & 0\\
    0 & \mathrm i H^{\text T} - \Gamma^{\text T} & 0 & 0\\
    0 & -2\Gamma^{\text T} & \mathrm i H^{\text T}  + \Gamma^{\text T} & 0\\
    -2\Gamma & 0 & 0 & -\mathrm i H + \Gamma
\end{pmatrix}
$$
Here $K_0$ is of size $4N_{\text e}\times 4N_{\text e}$, $K = p K_0 p^{+}$ is of size $N\times N$, and $p$ is the $N\times 4 N_{\text e}$ coefficient matrix satisfying $\boldsymbol E_{\alpha} = \sum_{\alpha} p_{\alpha, (k,s)} \boldsymbol b_{k,s}$, and $p^{+}$ is its Moore-Penrose pseudo-inverse. Note that nonzero eigenvalues of $K = p K_0 p^{+}$ coincides with those of $K_0$. In the unitary case, $\Gamma = 0$, thus all eigenvalues of $K$ are purely imaginary. Then $\|\mathcal E_{\alpha_1,\cdots,\alpha_n}(t,t_1,\cdots,t_n)\|_{\text{tr}}$ is constant in all $t_i$'s, and the estimate in \cref{eq:estimate_vacuum} holds for all $t\geq t_1\geq \cdots \geq t_n\geq 0$ with the same constant $\mathfrak c$. 

For the dissipative case,  we have 
$$
\begin{aligned}
    \|\mathcal E_{\alpha_1,\cdots,\alpha_n}(t,t_1,\cdots,t_n)\|_{\text{tr}}\leq & \|\mathrm e^{K t_n}\| \cdot\|\mathcal E_{\alpha_1,\cdots,\alpha_n}(t,t_1,\cdots,t_{n-1}, 0)\|_{\text{tr}} \\ \leq &\|\mathrm e^{K t_{n-1}}\| \|\mathrm e^{K t_n}\|  \|\mathcal E_{\alpha_1,\cdots,\alpha_n}(t,t_1,\cdots,t_{n-2}, 0,0)\|_{\text{tr}} \\ \leq & \cdots \leq \|\mathcal E_{\alpha_1,\cdots,\alpha_n}(t,0,\cdots,0)\|_{\text{tr}}\mathrm e^{\gamma_{\max} (t_1+\cdots+t_n)} \\=& \left\|\mathcal E_{\alpha_1,\cdots,\alpha_n}\right\|_{\text{tr}}\prod_{i=1}^n \|\mathrm e^{K t_i}\|
    \leq \left\|\mathcal E_{\alpha_1,\cdots,\alpha_n}\right\|_{\text{tr}}\left(\max_{s\in[0,t]}\| \mathrm e^{K s}\|\right)^n . 
\end{aligned}
$$ 
Here $\|\cdot \|$ denotes the matrix operator norm. In the last line we have used that $\mathrm e^{\boldsymbol{L}_{\text e} t} $ is trace-preserving. Thus it suffices to multiply $\mathfrak c$ by $\max_{s\in[0,t]}\| \mathrm e^{K s}\|$. 

We remark that in the case that $K$ is diagonalizable,   $\| \mathrm e^{K t}\|$ is bounded by $\mathrm e^{\gamma_{\max} t}$, where $\gamma_{\max}$ is the maximum eigenvalue of $\Gamma$. 
\end{proof}

\subsubsection{Dynamics and error bound}
\label{sec:spin_boson_error}
We have introduced the environment Liouvillian $\boldsymbol L_{\text e}$ and system-environment coupling $\boldsymbol L_{\text{SE}}$.
Now, let us introduce the system Liouvillian.
The system Liouvillian $\boldsymbol L_{\text s}^{\text U}$ from the system Hamiltonian $\hat H_{\text s}$ is given by 
$
\boldsymbol L_{\text s}^{\text U} = -\mathrm i[\hat H_{\text s},\hat\rho].
$
The system Hamiltonian $\hat H_{\text s}$ could be taken as any Hermitian operator acting on $ \mathcal H_{\text s} $. 
One typical physical example would be  $$\hat H_{\text s} = \sum_{j=1}^n
 -\frac{\Delta_j}{2} \hat\sigma_x^{(j)}+\frac{h_j}{2} \hat\sigma_z^{(j)},$$ where
$\Delta_j, h_j \in \mathbb{R}$ are a tunneling constant and an energy gap between the two states of the $j$-th spin. One could also consider non-unitary system dynamics with the system Liouvillian $
\boldsymbol L_{\text s} = \boldsymbol L_{\text s}^{\text U} + \boldsymbol D_{\text s}
$, where the system dissipator $\boldsymbol D_{\text s}$ is defined as $\boldsymbol D_{\text s}\hat\rho = \sum_{q=1}^{n_q}\hat l_q\hat\rho \hat l_q^\dagger-\frac{1}{2} \{\hat l_q^\dagger\hat l_q,\hat\rho \}$ for some $n_q$ and $\hat l_q\in B(\mathcal H_{\text s})$.

As a corollary of \cref{thm:rhoS_Dyson_short} (see also \cref{eq:rhoS_Dyson_short_formal}), the dynamics of unitary and Lindblad spin-boson model are given by:

\begin{cor}[$\hat\rho_{\text{S}}(t)$ for spin-boson model]
    We have the following results for $\hat\rho_{\text{S}}(t)$:
     \begin{equation}
        \begin{aligned}
            \hat\rho_{\text S}(t) &=\mathrm e^{\boldsymbol L_{\text s}t} \mathbf T^{+}\left( 
            \mathrm{exp} \left( \int_0^t \mathrm dt_1 \int_0^{t_1 } \mathrm dt_2   \boldsymbol{F} (t_1,t_2)  \right)
            \right) \hat\rho_{\text S}(0), \\
            \boldsymbol{F} (t_1,t_2) &=\sum_{i,j=1}^{n} \left(-c_{ij}(t_1-t_2)\boldsymbol \sigma_z^{(i)}(t_1)\boldsymbol \sigma_z^{(j)}(t_2)+c_{ij}(t_1-t_2)\widetilde{\boldsymbol \sigma}_z^{(i)}(t_1)\boldsymbol \sigma_z^{(j)}(t_2)\right.\\
            &\qquad  \qquad
            \left.  -(c_{ij}(t_1-t_2))^*  {\boldsymbol \sigma}_z^{(i)}(t_1)\widetilde{\boldsymbol \sigma}_z^{(j)}(t_2)+(c_{ij}(t_1-t_2))^* \widetilde{\boldsymbol \sigma}_z^{(i)}(t_1)\widetilde{\boldsymbol \sigma}_z^{(j)}(t_2)\right) .
        \end{aligned}
        \label{eq:rhoS_spin_boson}
    \end{equation}
Here, \cref{eq:rhoS_spin_boson} is understood as a shorthand notation for an infinite series summation, as defined in \cref{eq:T_exp_defn}.
$c_{ij}(t-t' )$($i,j = 1,\cdots,n$) refers to the conventional bath correlation function (BCF) of the spin-boson model~\cite[Sec. 3.2.4]{breuer2007open}. For $t\geq t'$, the coefficients $c_{ij}(t-t' )$ ($i,j = 1,\cdots,n$) are given as follows.
\begin{enumerate}
    \item For the unitary dynamics ($\boldsymbol L_{\text e} = \boldsymbol L_{\text e}^{\text U}$ (see \cref{eq:spin_boson_env}), $\boldsymbol L_{\text{SE}} = \boldsymbol L_{\text{SE}}^{\text U}$ (see \cref{lem:Liouvillian_unitary_spin_boson}), with the initial environment density operator as the Gibbs state \cref{eq:rhoE_init}), we have
    \begin{equation}
        c(t-t') = gG^>(t-t')g^\dagger + (g G^<(t-t')g^\dagger)^{*}, \quad t\geq t'.
        \label{eq:c_t_unitary_spin_boson}
    \end{equation}
    Note that here, $*$ means complex conjugate (not Hermitian conjugate), and we introduce the  greater and lesser Green's functions $G^>(t-t')$ and $G^<(t-t')$:
\begin{equation}
    G^>(t-t') = \mathrm e^{-\mathrm i H(t-t')}\left(I - \mathrm e^{-\beta H}\right)^{-1},\quad  G^<(t-t') = \mathrm e^{-\mathrm i H(t-t')}\left(\mathrm e^{\beta H} - I \right)^{-1}.\label{eq:Greens_unitary_spin_boson}
\end{equation}
\item For the Lindblad dynamics ($\boldsymbol L_{\text e} = \boldsymbol L_{\text e}^{\text U}+\boldsymbol{D}_{\text e}$ (see \cref{eq:spin_boson_env}), $\boldsymbol L_{\text{SE}} = \boldsymbol L_{\text{SE}}^{\text U}$ (see \cref{lem:Liouvillian_unitary_spin_boson}), with the initial vacuum environment \cref{eq:rhoE_init_empty}), we have 
 \begin{equation}
        c(t-t') = g \mathrm{e}^{(-\mathrm{i} H-\Gamma) (t-t')} g^{\dagger}, \quad t\geq t'.
        \label{eq:c_t_Lindblad_spin_boson}
    \end{equation}
\end{enumerate}
Here, $g$ is an $n\times N_{\text e}$ coefficient matrix in $\boldsymbol L_{\text{SE}}^{\text U}$,  $\boldsymbol D_{\text{SE}}$, and $H,\Gamma$ are $N_{\text e}\times N_{\text e}$  Hermitian matrices in \cref{eq:spin_boson_env}.  
\label{cor:spin_boson_dynamics}
\end{cor}
Since $\boldsymbol{b}_k, \boldsymbol{b}_k^\dagger, \widetilde{\boldsymbol{b}}_k, \widetilde{\boldsymbol{b}}_k^\dagger$ satisfy the CCR, and both 
$\boldsymbol L_{\text e}^{\text{U}}$ and $\boldsymbol{D}_{\text e}$ are quadratic in $\boldsymbol{b}_k, \boldsymbol{b}_k^\dagger, \widetilde{\boldsymbol{b}}_k, \widetilde{\boldsymbol{b}}_k^\dagger$, the Isserlis-Wick's condition holds.
To apply the result of \cref{thm:rhoS_Dyson_short} to get  \cref{cor:spin_boson_dynamics}, we only need to verify that the superoperator BCF $C_{\alpha\alpha'}(t-t')$ ($\alpha,\alpha'=1,\cdots,2n$) takes the following form:

\begin{lem}[BCFs for spin-boson model]
For $t\geq t'$, the two-point superoperator correlation function $\mathcal C_{\alpha\alpha'}(t-t')$ ($\alpha,\alpha'=1,\cdots,2n$)  is given by 
    \begin{equation}
        \begin{aligned}
            \mathcal C_{2i-1,2j-1}(t-t') &= -c_{ij}(t-t'),
        \quad  && \mathcal C_{2i,2j}(t-t')=-(c_{ij}(t-t'))^*, \\
        \mathcal C_{2i-1,2j}(t-t') &=c_{ij}(t-t')^*, \quad &&\mathcal C_{2i,2j-1}(t-t') =c_{ij}(t-t') ,
        \end{aligned}
        \label{eq:two_point_corr_unitary_spin_boson}
    \end{equation}
    where $c(t-t')$ is a $N_{\text e}\times N_{\text e}$ matrix-valued function, given by \cref{eq:c_t_unitary_spin_boson}, \cref{eq:c_t_Lindblad_spin_boson} 
    in the unitary and Lindblad dynamics, respectively.
    \label{lem:corr_spin_boson}
\end{lem}
\begin{rem}
\label{rem:ct_op_spin_boson}
The definitions of $G^>(t)$ and $G^<(t)$  in \cref{eq:Greens_unitary_spin_boson} align with the standard convention for greater and lesser Green's functions in the physics literature, differing only by a constant factor of $-\mathrm i$ (for example, see \cite[Section 6.1.3]{StefanucciVanLeeuwen2013}), namely,
\begin{equation}
    G_{kl}^>(t-t') = \operatorname{tr}\left(
   \hat b_k(t)\hat b_l^\dagger(t')\hat\rho_{\text E}(0)
    \right),\quad  G_{kl}^<(t-t') = \operatorname{tr}\left(
  \hat b_l^\dagger(t') \hat b_k(t)\hat\rho_{\text E}(0)
    \right).
    \label{eq:G<G>_op}
\end{equation}
Here we have used the notation that for operator $\hat O$ acting on $\mathcal H_{\text e}$, $\hat O(t)= \mathrm e^{-\mathrm i \hat H_{\text e}t}\hat O\mathrm e^{\mathrm i \hat H_{\text e}t}$.
As a result, $c_{ij}(t,t')$ coincides with the bath correlation functions of the environment operators, namely,
\begin{equation}
    c_{ij}(t-t') = \text{tr}(\hat E_i(t)\hat E_j(t')\hat\rho_{\text E}(0)),\quad \hat E_i = \sum_{k=1}^{N_{\text e}}g_{ik}\hat b_k + g_{ik}^*\hat b_k^\dagger.
    \label{eq:ct_op}
\end{equation}
Though in \cref{cor:spin_boson_dynamics} and \cref{lem:corr_spin_boson} we have only used $c(t-t')$ for $t\geq t'$,  using \cref{eq:ct_op,eq:G<G>_op}, one can see that in unitary systems \cref{eq:c_t_unitary_spin_boson,eq:Greens_unitary_spin_boson} still holds for $t<t'$. In particular, $c(t-t')$ satisfies the following Hermitian property:
$$
c(t-t') = (c(t'-t))^\dagger.
$$
\end{rem}
\begin{rem}
    In the case of diagonal $\hat H_{\text e}$, i.e., $H_{kl} = \omega_k\delta_{kl}$ ($k,l=1,\cdots N_{\text e}$), and real coupling coefficient $g$, the BCF $c_{ij}(t-t')$ reduces to the widely-used form in literature (for example, see \cite[Eqs. (3.388) and (3.389)]{breuer2007open}):
    \begin{equation}
    \begin{aligned}
          c_{ij}(t-t') &= 
       \sum_{k=1}^{N_{\text e}}   g_{ik}g_{jk}  \left(\coth\left(\frac{\beta\omega_k}{2}\right)\cos(\omega_k (t-t')) - \mathrm i\sin(\omega_k (t-t'))\right).
    \end{aligned}
    \end{equation}
    \label{rem:diagonal_H_e}
    Our results are applicable to the more general case, known as \emph{coupled modes} \cite{Mascherpa2020, lednev2024lindblad, HuangParkChanetal2025}, where $H$ is a non-diagonal matrix.
\end{rem}
We will prove \cref{lem:corr_spin_boson} in \cref{sec:env_correlation}. 
Furthermore, we have already shown that the estimate \cref{eq:estimate} holds for this example in \cref{sec:estimate}.
Then, as an application of \cref{thm:main_error_bound} and \cref{thm:main}, 
we have the following error bound  for spin-boson models:
\begin{cor}[Error bound for spin-boson  models]
Let $\hat\rho_{\text S}(t)$, $\hat\rho_{\text S}'(t)$ be the reduced system density operators of a unitary or Lindblad spin-boson model. Let $c_{ij}(t)$ and $c_{ij}'(t)$ be the 
correlation function corresponding to $\hat\rho_{\text S}(t)$, $\hat\rho_{\text S}'(t)$, as defined by \cref{eq:c_t_unitary_spin_boson} or \cref{eq:c_t_Lindblad_spin_boson}. Then 
 the same error bounds as in \cref{thm:main_error_bound} for reduced system densities and \cref{thm:main} for any bounded system observables $\hat O_{\text S}$ hold, i.e.,
\begin{equation}
\begin{aligned}
    \|\hat\rho_{\text S}'(t) - \hat\rho_{\text S}(t)\|_{\text{tr}}\leq \mathrm e^{\epsilon t^2/2}-1,\quad 
    |O_{\text S}(t) - O_{\text S}'(t)|\leq (\mathrm e^{\epsilon t^2/2}-1)\|\hat O_{\text S}\|, \quad \forall t\in [0,T],  \\
    \|\hat\rho_{\text S}'(t) - \hat\rho_{\text S}(t)\|_{\text{tr}}\leq \mathrm e^{\epsilon_1 t}-1,\quad 
    |O_{\text S}(t) - O_{\text S}'(t)|\leq (\mathrm e^{\epsilon_1 t}-1)\|\hat O_{\text S}\|, \quad \forall t\in [0,T], 
\end{aligned}
\label{eq:error_bound_application}
\end{equation}
where $\epsilon$ and $\epsilon_1$ are given by:
$$\epsilon = 4n  \|c  - c '\|_{L^\infty([0,T])},\quad \epsilon_1 = 4n   \|c  - c '\|_{L^1([0,T])}.$$
\label{cor:error_spin_boson}
\end{cor}
We can directly apply \cref{cor:error_spin_boson} to quantify the error bound of approximations in the pseudomode theory for the spin-boson model by considering the dynamics of $\hat{\rho}_S(t)$ as the original unitary dynamics and the dynamics of $\hat{\rho}'_S(t)$ as the pseudomode dynamics from the Lindblad or quasi-Lindblad dynamics.

\subsection{Fermionic impurity model}
\label{sec:fermion}

\subsubsection{Fermionic environment}
\label{sec:fermion_env}
The fermionic environment Hilbert space is given by
\begin{equation}
    \mathcal H_{\text E} = \operatorname{span}\left(\left\{(\hat c_1^\dagger)^{j_1}\cdots(\hat c_n^\dagger)^{j_{N_{\text e}}}|0\rangle_{\text E}, j_1,\cdots,j_{N_{\text e}}=0,1\right\}\right).
\label{eq:Hilbertspace_env_fermion}
\end{equation}
$N_{\text e}$ is the number of fermionic modes, $|0\rangle_{\text E}$ is the environment vacuum state, and
$\hat c_k,\hat c_k^\dagger$ are the fermionic annihilation and creation operators for the $k$-th environment modes ($k =1,\cdots, N_{\text e}$), which satisfy the canonical anti-commutation relations (CAR) $\{\hat c_k,\hat c_{k'}^\dagger\} = \delta_{kk'}\hat 1_{\text e}$. We might want to define the following superoperators, just as in the bosonic case 
\cref{eq:bosonic_env_superoperators}:
$$
\boldsymbol c_k: (\cdot)\rightarrow \hat c_k(\cdot),\quad \boldsymbol c_k^\dagger: (\cdot)\rightarrow \hat c_k^\dagger(\cdot),\quad \widetilde{\boldsymbol c}_k: (\cdot)\rightarrow (\cdot)\hat c_k^\dagger,\quad \widetilde{\boldsymbol c}_k^\dagger: (\cdot)\rightarrow (\cdot)\hat c_k.
$$
However, in this definition, though $\boldsymbol c_k$ and $\boldsymbol c_k^\dagger$  satisfies that $\{\boldsymbol c_k,\boldsymbol c_{k'}^\dagger\} = \delta_{kk'}\boldsymbol 1_{\text e}$,
$\boldsymbol c_k$ and $\widetilde{\boldsymbol c}_{k'}$ actually commute rather than anticommute.
Thus, this definition does not satisfy the CAR, which is necessary for  the Isserlis-Wick's conditions (see \cref{defn:Wick_condition}) to hold. 
To overcome this, we introduce the parity operator $\hat P_{\text e}\in B(\mathcal H_{\text E})$ defined as follows:
\begin{equation}
    \hat P_{\text e} = (-1)^{\hat N_{\text e}},\quad \hat N_{\text e} = \sum_{k=1}^N \hat c_k^\dagger\hat c_k.
    \label{eq:parity_operator}
\end{equation}
A useful property is that $\hat P_{\text e}$ and $\hat c_k,\hat c_k^{\dagger}$ anticommute, i.e., $\hat P_{\text e}\hat c_k = -\hat c_k \hat P_{\text e}$, $\hat P_{\text e}\hat c_k^{\dagger} = -\hat c_k^{\dagger}\hat P_{\text e}$.
With the parity operator, the fermionic environment superoperators are defined as follows: 
\begin{defn}[Fermionic environment superoperators]
    For environment fermionic creation and annihilation operators $\hat c_k$, $\hat c_k^\dagger$, let us define the following superoperators, denoted as 
    $\boldsymbol c_k$, $\boldsymbol c_k^\dagger$, $\widetilde{\boldsymbol{c}}_k$, $\widetilde{\boldsymbol{c}}_k^\dagger$ ($k=1,\cdots,{N_{\text e}}$):
\begin{equation}
    \begin{aligned}
    \boldsymbol{c}_k : (\cdot)\rightarrow \hat P_{\text{e}}\hat c_k(\cdot) ,\quad 
    \boldsymbol{c}_k^{\dagger}: (\cdot)\rightarrow - \hat P_{\text{e}} \hat c_k^\dagger(\cdot), \quad
    \widetilde{\boldsymbol{c}}_k: (\cdot)\rightarrow \hat P_{\text{e}}(\cdot) \hat c_k^{\dagger},\quad
    \widetilde{\boldsymbol{c}}_k^{\dagger}: (\cdot)\rightarrow \hat P_{\text{e}}(\cdot) \hat c_k.
    \end{aligned}
    \label{eq:fermionic_env_superoperators}
\end{equation}
\label{defn:fermionic_superoperators}
\end{defn}
Such definitions satisfy the CAR, summarized as follows,
\begin{lem}[Canonical anti-commutation relations for fermionic superoperators]
    The bath superoperators $\boldsymbol c_k,\boldsymbol c_k^\dagger,\widetilde{\boldsymbol{c}}_k,\widetilde{\boldsymbol{c}}_k^\dagger$ ($k=1,\cdots,{N_{\text e}}$) satisfy the following canonical anti-commutation relations:
    \begin{equation}
        \begin{aligned}
        \left\{\boldsymbol{c}_k, \boldsymbol{c}_{k'}^{\dagger}\right\} &=  \delta_{kk'}\boldsymbol 1_{\text e},
        \quad \left\{\widetilde{\boldsymbol{c}}_k, \widetilde{\boldsymbol{c}}_{k'}^{\dagger}\right\} = \delta_{kk'}\boldsymbol 1_{\text e},\quad
        \\
        \left\{\boldsymbol{c}_k, {\boldsymbol{c}}_{k'}\right\}  = 
        \left\{\boldsymbol{c}_k^{\dagger},{\boldsymbol{c}}_{k'}^{\dagger}\right\}  =
        \left\{\widetilde{\boldsymbol{c}}_k^{\dagger},\widetilde{\boldsymbol{c}}_{k'}^{\dagger}\right\}  &=
        \left\{\widetilde{\boldsymbol{c}}_k,\widetilde{\boldsymbol{c}}_{k'}\right\}  =
        \left\{\boldsymbol{c}_k, \widetilde{\boldsymbol{c}}_{k'}\right\}  =
        \left\{\boldsymbol{c}_k, \widetilde{\boldsymbol{c}}_{k'}^{\dagger}\right\}  =
        \left\{\boldsymbol{c}_k^{\dagger}, \widetilde{\boldsymbol{c}}_{k'}^{\dagger}\right\}  = 0.
        \end{aligned}
        \label{eq:canonical_fermionic}
    \end{equation}
    \label{lem:CAR}
\end{lem}
The environment Liouvillian $\boldsymbol L_{\text e}^{\text U}$ is given by
$\boldsymbol{L}_{\text e}^{\text U} = -\mathrm i [\hat H_{\text e}, \cdot]$, $\hat H_{\text e} = \sum_{k,l=1}^{N_{\text{e}}}H_{kl}\hat c_k^\dagger\hat c_l$,
 where $H$ is a $N_{\text e}\times N_{\text e}$ Hermitian matrix. 
 Since $\boldsymbol c_k^\dagger\boldsymbol c_l\hat\rho = -\hat P_{\text e}\hat c_k^\dagger\hat P_{\text e}\hat c_l\hat\rho = \hat c_k^\dagger\hat P_{\text e}\hat P_{\text e}\hat c_l\hat \rho = \hat c_k^\dagger\hat c_l\hat\rho$ and 
$\widetilde{\boldsymbol{c}}_l^\dagger\widetilde{\boldsymbol{c}}_k\hat\rho =  \widetilde{\boldsymbol{c}}_l^\dagger(\hat P_{\text e}\hat\rho\hat c_k^\dagger) = \hat P_{\text e}\hat P_{\text e}\hat\rho\hat c_k^\dagger\hat c_l = \hat\rho\hat c_k^\dagger\hat c_l$, $\boldsymbol{L}_{\text e}^{\text U}$ could also be written in the following superoperator form:
\begin{equation}
     \boldsymbol L_{\text e}^{\text U} = \sum_{k,l=1}^{N_{\text e}} -\mathrm iH_{kl}\boldsymbol c_k^\dagger\boldsymbol c_l + \mathrm iH_{kl}\widetilde{\boldsymbol c}_l^\dagger\widetilde{\boldsymbol c}_k,
     \label{eq:Le_U_fermionic}
\end{equation}

For the non-unitary dynamics, we focus on the specific setup used in \cite{Chen_2019, ParkHuangZhuetal2024}. Therein, the environment is divided into two separate parts: one with an initial vacuum environment and the other with an initial fully occupied environment, i.e.,
\begin{equation}
    \hat\rho_{\text E}(0) = (|\underbrace{0\cdots 0}_{N_{\text 1}}\rangle\langle \underbrace{0\cdots 0}_{N_{\text 1}}|)\otimes(|\underbrace{1\cdots 1}_{ N_{\text e}-N_1}\rangle\langle \underbrace{1\cdots 1}_{  N_{\text e}-N_1}|),
    \label{eq:rhoE_init_emp_fil}
\end{equation}
and there are no Hamiltonian and dissipative terms between the two different parts.
As a result, the Hamiltonian $\hat H$ is block-diagonal, i.e., $H = \operatorname{diag}\{H_1,H_2\}$, where $H_1,H_2$ are Hermitian matrices of size $N_1\times N_1$ and $(N_{\text e}-N_1)\times (N_{\text e}-N_1)$.
The environment dissipator $\boldsymbol{D}_{\text e}$ is set such that the initial density operator becomes stationary, $\boldsymbol{D}_{\mathrm{e}} \hat\rho_{\text E}(0)=0$, thus according to \cref{rmk:single_variable_Corr}, the two-point correlation functions $C_{\alpha,\alpha'}(t,t')$ depend only on $(t-t')$. In particular, the dissipator has the following operator form, $\boldsymbol{D}_{\text e} = \boldsymbol{D}_{\text e}^-+\boldsymbol{D}_{\text e}^+$,
\begin{equation}
        \boldsymbol D_{\text e}^-\hat\rho = \sum_{k,k'=1}^{N_1} 2\Gamma^-_{kk'}
        \left(  \hat c_{k'}\hat\rho \hat c_k^\dagger - \frac{1}{2} 
        \{ \hat c_k^\dagger \hat c_{k'},\hat\rho\}\right),\quad
        \boldsymbol D_{\text e}^+\hat\rho = \sum_{l,l'=N_1+1}^{N_{\text e}} 2\Gamma^+_{ll'}
        \left(  \hat c_l^\dagger\hat\rho \hat c_{l'} - \frac{1}{2}
        \{ \hat c_{l'} \hat c_l^\dagger,\hat\rho\}\right),
    \end{equation}
where $\Gamma^-$ and $\Gamma^+$ are positive semi-definite matrices with sizes $N_1\times N_1$ and $(N_{\text e}-N_1)\times (N_{\text e}-N_1)$, respectively. In the superoperator form, using that $
\widetilde{\boldsymbol{c}}_k \boldsymbol{c}_{k'}\hat\rho = 
\hat P_{\text e}(\hat P_{\text e}\hat c_{k'}\hat\rho) \hat c_k^\dagger = \hat c_{k'} \hat\rho \hat c_k^\dagger$ and $ \boldsymbol{c}_l^{\dagger}
\widetilde{\boldsymbol{c}}_{l'}^{\dagger}\hat\rho = -\hat P_{\text e}\hat c_l^\dagger(\hat P_{\text e}\hat\rho \hat c_{l'}) = \hat c_l^\dagger \hat\rho \hat c_{l'}
$, we can see that $\boldsymbol{D}_{\text e}^-$ and $\boldsymbol{D}_{\text e}^+$ take the following superoperator form:
\begin{equation}
    \boldsymbol D_{\text e}^- = \sum_{k,k'=1}^{N_1}2 \Gamma_{kk'}^- \widetilde{\boldsymbol{c}}_k \boldsymbol{c}_{k'} -\Gamma_{kk'}^- \boldsymbol{c}_k^{\dagger} \boldsymbol{c}_{k'}-\Gamma_{kk'}^- \widetilde{\boldsymbol{c}}_{k'}^{\dagger} \widetilde{\boldsymbol{c}}_k,\text{ }  \boldsymbol{D}_{\text e}^+ =
    \sum_{l,l' =N_1+1}^{N_{\text e}}  2 \Gamma_{ll'}^+ \boldsymbol{c}_l^{\dagger}\widetilde{\boldsymbol{c}}_{l'}^{\dagger} -\Gamma_{ll'}^+  \boldsymbol{c}_{l'}\boldsymbol{c}_l^{\dagger}  -\Gamma_{ll'}^+  \widetilde{\boldsymbol{c}}_l 
    \widetilde{\boldsymbol{c}}_{l'}^{\dagger} .
    \label{eq:D_e_fermion} 
\end{equation}
\subsubsection{System-environment coupling and the fermionic superselection rule}
\label{sec:fermion_env_sys}
The system Hilbert space is given by
\begin{equation}
    \mathcal H_{\text S} = \operatorname{span}\left(\left\{(\hat a_1^\dagger)^{j_1}\cdots(\hat a_n^\dagger)^{j_{n}}|0\rangle_{\text S}, j_1,\cdots,j_n=0,1\right\}\right),
    \label{eq:Hilbertspace_sys_fermion}
\end{equation}
where $|0\rangle_{\text S}$ is a system vacuum state. $\hat a_i$ and $\hat a_i^\dagger$ ($i=1,\cdots,n$) are system fermionic annihilation and creation operators of $n$ fermionic modes that also satisfy the CAR: $\{\hat a_i,\hat a_{i'}^\dagger\} = \delta_{ii'}\hat{1}_{\text s}$, where $\hat 1_{\text s}\in B(\mathcal H_{\text s})$ is the system identity operator. 
For $\hat a_j$, $\hat a_j^\dagger\in B(\mathcal H_{\text S})$ ($j=1,\cdots,n$), let us define the following superoperators, denoted as 
$\boldsymbol a_j$, $\boldsymbol a_j^\dagger$, $\widetilde{\boldsymbol{a}}_j$, $\widetilde{\boldsymbol{a}}_j^\dagger$: 
\begin{equation}
    \boldsymbol{a}_j: (\cdot)\rightarrow\hat a_j(\cdot),\quad
    \boldsymbol{a}_j^{\dagger}: (\cdot)\rightarrow\hat a_j^{\dagger}(\cdot),\quad
    \widetilde{\boldsymbol{a}}_j: (\cdot)\rightarrow -\hat P_{\text s}(\cdot)\hat a_j^{\dagger}\hat P_{\text s},\quad
    \widetilde{\boldsymbol{a}}_j^{\dagger}: (\cdot) \rightarrow \hat P_{\text s}(\cdot)\hat a_j\hat P_{\text s}.
    \label{eq:env_system_supero}
\end{equation}
We introduce the system parity operator $\hat P_{\text s} = (-1)^{\hat N_{\text s}}$, where $\hat N_{\text s} = \sum_{i=1}^n \hat a_i^\dagger\hat a_i$. Similar to $\hat P_{\text e}$, which anticommutes with $\hat c_k$ and $\hat c_k^\dagger$, $\hat P_{\text s}$ also anticommutes with $\hat a_i$ and $\hat a_i^\dagger$. Later, we will also use the parity operator of the entire extended system $\hat P = (-1)^{\hat N} = (-1)^{\hat I_{\text e}\otimes\hat N_{\text s}+\hat N_{\text e}\otimes \hat I_{\text s}} = \hat P_{\text e}\otimes\hat P_{\text s}$.
The system superoperators in \cref{eq:env_system_supero} satisfy the similar anti-commutation rules as in \cref{lem:CAR}.

For $i=1,\cdots,n$ 
and $k=1,\cdots,N_{\text e}$, let us define $\hat a_i|_{\mathcal H},\hat a_i^{\dagger}|_{\mathcal H}$ and $\hat c_k|_{\mathcal H},\hat c_k^{\dagger}|_{\mathcal H}$, which are extensions of system operators $\hat a_i,\hat a_i^{\dagger}\in B(\mathcal H_{\text S})$ and environment operators $\hat c_k,\hat c_k^{\dagger}\in\mathcal H_{\text E}$ onto the entire Hilbert space $\mathcal H$. 
In fermionic systems, such an extension could not be trivially defined as $\hat a_i^{\dagger}|_{\mathcal H} = \hat I_{\text e}\otimes \hat a_i^{\dagger}$ and $\hat c_k^{\dagger}|_{\mathcal H} = \hat c_k^{\dagger}\otimes \hat I_{\text s}$, since in this case $\hat a_i|_{\mathcal H}$ and $\hat c_k|_{\mathcal H}$ commute while they should anticommute.
As shown in the \cref{sec:Z2}, the definition of
 $\hat a_i|_{\mathcal H},\hat a_i^\dagger|_{\mathcal H}$ and $\hat c_j|_{\mathcal H},\hat c_j^\dagger|_{\mathcal H}$ should follow from the
$\mathbb Z_2$-graded  structure of the Fock space (for example, see \cite[Chap 3]{Berezin2013}):
\begin{equation}
    \begin{aligned}
        \hat a_i^{\dagger}|_{\mathcal H} = \hat P_{\text e}\otimes\hat a_i^{\dagger},\quad 
        \hat a_i|_{\mathcal H} = \hat P_{\text e}\otimes\hat a_i,\quad \hat c_k^{\dagger}|_{\mathcal H} = \hat c_k^{\dagger}\otimes \hat I_{\text s},
        \quad \hat c_k|_{\mathcal H} = \hat c_k\otimes \hat I_{\text s}\label{eq:fermionic_extension}.
    \end{aligned}
\end{equation}
The anti-commutation relationship can directly be shown from this definition:
\begin{equation}\begin{aligned}\hat{a}^\dag_i |_{\mathcal{H}} \hat{c}_k |_{\mathcal{H}} &= (\hat P_{\text e}\otimes \hat a_i^{\dagger})(\hat c_k \otimes \hat I_{\text s}) = (\hat P_{\text e}\hat c_k) \otimes (\hat a_i^{\dagger}\hat I_{\text s})  \\ &= - (\hat c_k\hat P_{\text e}) \otimes (\hat I_{\text s} \hat a_i^{\dagger}) = - (\hat c_k \otimes \hat I_{\text s}) (\hat P_{\text e}\otimes \hat a_i^{\dagger}) = - \hat c_k|_{\mathcal H}\hat a_i^{\dagger}|_{\mathcal H}.\end{aligned}\end{equation}
Based on the extended operators, the system-environment Liouvillian $\boldsymbol L_{\text{SE}}^{\text U}$ is given by
\begin{equation}
         \boldsymbol L_{\text{SE}}^{\text U}\hat\rho =  -\mathrm i[\hat H_{\text{SE}},\hat\rho],\quad   \hat H_{\text{SE}} = 
            \sum_{i=1}^n\sum_{k=1}^{N_{\text e}} 
           ( \nu_{ik}\hat a_i^\dagger|_{\mathcal H} \hat c_k|_{\mathcal H} + \nu_{ik}^*\hat c_k^\dagger|_{\mathcal H}\hat a_i|_{\mathcal H}).
           \label{eq:L_SE_U_fermion_O}
\end{equation}
$\hat H_{\text{SE}}$ is a Hermitian operator known as the system-environment Hamiltonian, $\nu$ is a $n\times N_{\text e}$ coefficient matrix that describes the system-environment coupling. Since this is a fermionic system, due to the  fermionic superselection rule \cite{wick1997intrinsic} \footnote{Though according to the superselection rule we only need to consider $\hat\rho\in  \mathcal P(\mathcal H_{\text E}\otimes \mathcal H_{\text S})$. However, mathematically, one could treat $\rho^{(0)}\notin  \mathcal P(\mathcal H_{\text E}\otimes \mathcal H_{\text S})$ by keeping track of the action of the superoperators in both even and odd parity sectors. We refer to \cite{cirio2022canonical} for details.}, we have the following additional constraint for the intial density operator $\hat\rho(0)$:
\begin{equation}
\begin{aligned}
      \text{Fermionic superselection rule: }& \hat\rho(0)\in \mathcal P(\mathcal H_{\text E}\otimes \mathcal H_{\text S}),\\& \qquad 
\mathcal P(\mathcal H_{\text E}\otimes \mathcal H_{\text S}) = \{\hat \rho\in B_1(\mathcal H_{\text E}\otimes \mathcal H_{\text S}) | [\hat \rho,\hat P]=0\}.
\end{aligned}
\label{rem:SE_decomposition}
\end{equation}
In other words, $\hat\rho(0)$ commutes with the parity operator $\hat P$. We refer to  $\mathcal P(\mathcal H_{\text E}\otimes \mathcal H_{\text S})$ as the physical subspace of $B_1(\mathcal H_{\text E}\otimes \mathcal H_{\text S})$. 
Since $\mathcal P(\mathcal H_{\text E}\otimes \mathcal H_{\text S})$ is an invariant subspace for every term in $\boldsymbol{L}_{\text E}$, $\boldsymbol{L}_{\text{SE}}^{\text U}$, and later in $\boldsymbol{ L}_{\text S}$, it suffices to consider their restrictions to the subspace $\mathcal P(\mathcal H_{\text E}\otimes \mathcal H_{\text S})\subset B_1(\mathcal H_{\text E}\otimes \mathcal H_{\text S})$. This reduction significantly simplifies the subsequent calculations.
We will show in \cref{lem:L_SE_D_SE_Expression} that the Liouvillian $\boldsymbol{L}_{\text{SE}}^{\text U}$ in \cref{eq:L_SE_U_fermion_O} could  be written in the following superoperator form:
    \begin{align}        \boldsymbol L_{\text{SE}}^{\text U} &= 
        \sum_{\alpha=1}^{4n}\boldsymbol E_{\alpha}^{\text U}\otimes \boldsymbol S_{\alpha},\label{eq:L_SE_U_fermion}\\
        \boldsymbol S_{4j-3} &= \boldsymbol a_j^\dagger, \quad \boldsymbol S_{4j-2} = \boldsymbol a_j, \quad \boldsymbol S_{4j-1} = \widetilde{\boldsymbol a}_j^\dagger, \quad \boldsymbol S_{4j} = \widetilde{\boldsymbol a}_j\\
        \boldsymbol E_{4j-3}^{\text U} &= -\mathrm i\sum_{k=1}^{N_{\text e}} \nu_{jk} \boldsymbol c_k,\quad 
        \boldsymbol E_{4j-2}^{\text U} =-\mathrm i \sum_{k=1}^{N_{\text e}} \nu_{jk}^* \boldsymbol c_k^\dagger, \quad
        \boldsymbol E_{4j-1}^{\text U} =\mathrm i \sum_{k=1}^{N_{\text e}} \nu_{jk}^* \widetilde{\boldsymbol c}_k,
        \quad
        \boldsymbol E_{4j}^{\text U} = \mathrm i\sum_{k=1}^{N_{\text e}} \nu_{jk} \widetilde{\boldsymbol c}_k^\dagger, \label{eq:E_alpha_fermionic}    \end{align}
for $j=1,\cdots,n$. 
Note that both sides of the equality $\boldsymbol L_{\text{SE}}^{\text U} =     \sum_{i=1}^{4n}\boldsymbol E^{\text{U}}_{i}\otimes \boldsymbol S_{i}$ is to be understood in the sense that both sides act on the subspace $\mathcal P(\mathcal H_{\text E}\otimes \mathcal H_{\text S})$.
For future use, let us define $\nu^-$, $\nu^+$  as sublocks of $\nu$: $\nu = [\nu^-,\nu^+]$, where $\nu^-$ are of size $n\times N_1$, and $\nu^+$ are of size $n\times (N_{\text e}-N_1)$.

\subsubsection{Dynamics and error bound}
\label{sec:fermion_error}
For fermionic impurity models, let us consider the following system Hamiltonian:
\begin{equation}
        \begin{aligned}
            \hat H_{\text s} &=  \sum_{ii'=1}^n h_{ii'}\hat a_i^\dagger\hat a_{i'} + \sum_{ii'j'j=1}^n V_{ii'j'j}  \hat a_i^\dagger\hat a_{i'}^\dagger\hat a_j\hat a_{j'}.
        \end{aligned}
        \label{eq:fermionic_Hamiltonian}
    \end{equation}
Here
$h$ is a $n\times n$ Hermitian matrix, $V$ is a $n\times n\times n\times n$ tensor for electron-electron interaction. The system Liouvillian $\boldsymbol{L}_{\text s} $ could be taken as the Liouvillian $\boldsymbol{L}_{\text s}^{\text U}= -\mathrm i[\hat H_{\text s},\cdot]$. One could also take non-unitary system dynamics, i.e., $\boldsymbol{L}_{\text s} =\boldsymbol{L}_{\text s}^{\text U}+\boldsymbol{D}_{\text s}$, where $\boldsymbol D_{\text s}\hat\rho = \sum_{q=1}^{n_q}\hat l_q\hat\rho \hat l_q^\dagger-\frac{1}{2} \{\hat l_q^\dagger\hat l_q,\hat\rho \}$ for some  $n_q$ and $\hat l_q\in B(\mathcal H_{\text s})$.

As a corollary of \cref{thm:rhoS_Dyson_short} (see also \cref{eq:rhoS_Dyson_short_formal}), the reduced system dynamics of the fermionic impurity model (for unitary, Lindblad and quasi-Lindblad dynamics) are given by:
\begin{cor}[$\hat\rho_{\text{S}}(t)$ for fermionic impurity model]
    We have the following results for $\hat\rho_{\text{S}}(t)$:
     \begin{equation}
        \begin{aligned}
            \hat\rho_{\text S}(t) &=\mathrm e^{\boldsymbol L_{\text s}t} \mathbf T^{-}\left( 
            \mathrm{exp} \left( \int_0^t \mathrm dt_1 \int_0^{t_1 } \mathrm dt_2   \boldsymbol{F} (t_1,t_2)  \right)
            \right) \hat\rho_{\text S}(0), \\
            \boldsymbol{F} (t_1,t_2) &=\sum_{i,j=1}^{n}-\Delta^<_{ij}(t_1-t_2) \left(\boldsymbol a_i^\dagger(t_1) +\widetilde{\boldsymbol a}_i (t_1)\right)\boldsymbol a_j(t_2) - \Delta^>_{ij}(t_1-t_2) \left(\boldsymbol a_i^\dagger(t_1) +\widetilde{\boldsymbol a}_i (t_1)\right)
    \widetilde{\boldsymbol a}_j^\dagger(t_2) \\
    &\qquad   +(\Delta_{ij}^<(t_1-t_2))^*\left(\boldsymbol a_i(t_1) - \widetilde{\boldsymbol a}_i^\dagger(t_1)\right)\widetilde{\boldsymbol a}_j(t_2)+(\Delta_{ij}^>(t_1-t_2))^*\left(  \widetilde{\boldsymbol a}_i^\dagger(t_1)-\boldsymbol a_i(t_1)\right)\boldsymbol a_j^\dagger(t_2).
        \end{aligned}
        \label{eq:rhoS_fermion}
    \end{equation}
\cref{eq:rhoS_fermion} is understood as a shorthand notation for an infinite series summation, as defined in \cref{eq:T_exp_defn}.
We refer to $\Delta^>_{ij}(t-t' )$, $\Delta^<_{ij}(t-t' )$ ($i,j = 1,\cdots,n$) as the greater and lesser hybridization function. For $t\geq t'$, $\Delta^>_{ij}(t-t' )$ and $\Delta^<_{ij}(t-t')$ ($i,j = 1,\cdots,n$) is defined as:
\begin{enumerate}
    \item For the unitary dynamics ($\boldsymbol L_{\text e} = \boldsymbol L_{\text e}^{\text U}$ (see \cref{{eq:Le_U_fermionic}}), $\boldsymbol L_{\text{SE}} = \boldsymbol L_{\text{SE}}^{\text U}$ (see  \cref{eq:L_SE_U_fermion}),  with the initial environment density operator as the Gibbs state \cref{eq:rhoE_init}), we have
     \begin{equation}
        \Delta^>(t-t') = \nu G^>(t-t') \nu^\dagger,\quad \Delta^<(t-t') = \nu G^<(t-t') \nu^\dagger.\label{eq:c_t_fermion_unitary}
    \end{equation}
    The lesser and greater Green's functions $G^<(t)$ and $G^>(t)$ are defined as:
\begin{equation}
    G^<(t-t') = \left(\mathrm{e}^{\beta H}+I\right)^{-1} \mathrm{e}^{-\mathrm{i} H (t-t')},\quad G^>(t-t') = \left(I+\mathrm{e}^{-\beta H}\right)^{-1} \mathrm{e}^{-\mathrm{i} H (t-t')}.
    \label{eq:Greens_unitary_fermionic}
\end{equation}
\item For the Lindblad dynamics ($\boldsymbol L_{\text e} = \boldsymbol L_{\text e}^{\text U}+\boldsymbol{D}_{\text e}$ (see \cref{eq:Le_U_fermionic,eq:D_e_fermion}), $\boldsymbol L_{\text{SE}} = \boldsymbol L_{\text{SE}}^{\text U}$ (see \cref{eq:L_SE_U_fermion}),  with the initial environment density operator in \cref{eq:rhoE_init_emp_fil}), we have 

    \begin{equation}
        \Delta^>(t-t') = \nu^- \mathrm e^{(-\mathrm iH^--\Gamma^-)(t-t')}(\nu^-)^\dagger,\quad\Delta^<(t-t') = \nu^+\mathrm e^{(-\mathrm iH^+-\Gamma^+)(t-t')}(\nu^+)^\dagger.
        \label{eq:c_t_fermion_Lindblad}
    \end{equation}
\end{enumerate}
\label{cor:dynamics_fermion}
\end{cor}
Since $\boldsymbol{c}_k, \boldsymbol{c}_k^\dagger, \widetilde{\boldsymbol{c}}_k, \widetilde{\boldsymbol{c}}_k^\dagger$ satisfy the CAR (see \cref{lem:CAR}), and both 
$\boldsymbol L_{\text e}^{\text{U}}$ and $\boldsymbol{D}_{\text e}$ are quadratic in $\boldsymbol{c}_k, \boldsymbol{c}_k^\dagger, \widetilde{\boldsymbol{c }}_k, \widetilde{\boldsymbol{c}}_k^\dagger$, the Isserlis-Wick's condition holds.
Thus, according to \cref{thm:rhoS_Dyson_short}, \cref{cor:spin_boson_dynamics} holds as long as the following lemma for the two-point BCFs holds:
\begin{lem}[BCFs for fermionic impurity model]
For $t \geq t^{\prime}$, the two-point correlation function $\mathcal C_{\alpha \alpha^{\prime}}\left(t-t^{\prime}\right)\left(\alpha, \alpha^{\prime}=1, \cdots, 4 n\right)$ is given by the following:
\begin{equation}
    \begin{aligned}
        & \mathcal C_{4i-3,4j-2}(t-t') = \mathcal C_{4i,4j-2}(t-t') = -\Delta_{ij}^>(t-t'),\\ &
         \mathcal C_{4i-3,4j-1}(t-t') = \mathcal C_{4i,4j-1}(t-t') = -\Delta_{ij}^<(t-t'),\\
       &  \mathcal C_{4i-2,4j}(t-t') = -\mathcal C_{4i-1,4j}(t-t') = (\Delta_{ij}^>(t-t'))^*,\\&
         \mathcal C_{4i-1,4j-3}(t-t') = -\mathcal C_{4i-2,4j-3}(t-t') = (\Delta_{ij}^<(t-t'))^*,
       \end{aligned}
       \label{eq:two_point_corr_unitary_fermionic}
\end{equation}
    where $\Delta^<(t)$ and $\Delta^>(t)$ are the lesser and greater hybridization functions, given by \cref{eq:c_t_fermion_unitary}, \cref{eq:c_t_fermion_Lindblad}
    in \cref{cor:dynamics_fermion}, for the unitary and Lindblad dynamics, respectively.
    \label{lem:corr_fermionic}
\end{lem}
\begin{rem}
    Similar to \cref{rem:ct_op_spin_boson} and \cref{rem:diagonal_H_e},
    the greater and lesser Green's functions in \cref{eq:Greens_unitary_fermionic} follow the convention in condensed matter physics, differing only by a constant factor of $\mp\mathrm i$, where $G^>(t)$ and $G^<(t)$ are defined as in \cref{eq:G<G>_op}, with $\hat b_k,\hat b_k^\dagger$ replaced by $\hat c_k,\hat c_k^\dagger$. As a result, for unitary dynamics, $\Delta^>$ and $\Delta^<$ in \cref{eq:c_t_fermion_unitary} coincides with the following operator BCFs:
    $$
    \Delta_{ij}^>(t-t') = \text{tr}(\hat E_i(t)\hat E_j^\dagger(t')\hat\rho_{\text E}(0)),\quad \Delta_{ij}^<(t-t') = \text{tr}(\hat E_j^\dagger(t')\hat E_i(t)\hat\rho_{\text E}(0)),\quad \hat E_i = \sum_{k=1}^{N_{\text e}}\nu_{ik}\hat c_k.
    $$
    Thus, also similar to the spin-boson case, though in \cref{cor:dynamics_fermion}  and \cref{lem:corr_fermionic} we have only used $\Delta^{>,<}(t-t')$ for $t>t'$, using the above definition, $\Delta^{>,<}(t-t')$ could be extended to $t<t'$. In other words, for unitary dynamics, \cref{lem:corr_fermionic} still holds for $t<t'$, and the Hermiticity property also holds, i.e., 
    $$
    \Delta^>(t-t') = (\Delta^>(t'-t))^\dagger, \quad \Delta^<(t-t') = (\Delta^<(t'-t))^\dagger.
    $$
    If $\hat H_{\text e}$ is diagonal, i,e, $H_{kl} = \omega_k\delta_{kl}$, then $\Delta^{\lessgtr}(t)$ reduces to the following form, which is commonly used in non-equilibrium dynamical mean-field theory \cite{AokiTsujiEcksteinetal2014}:
    \begin{equation}
        \Delta_{ij}^{>}(t) = \sum_{k=1}^{N_{\text e}} \nu_{ik}\nu_{jk}^*\frac{\mathrm e^{-\mathrm i\omega_k t}}{1+\mathrm e^{-\beta\omega_k}},\quad
        \Delta_{ij}^{<}(t) = \sum_{k=1}^{N_{\text e}} \nu_{ik}\nu_{jk}^*\frac{\mathrm e^{-\mathrm i\omega_k t}}{1+\mathrm e^{\beta\omega_k}}.
    \end{equation}
\end{rem}
We will prove \cref{lem:corr_fermionic} in \cref{sec:env_correlation}. Since here the fermionic $\mathcal H_{\text E}$ is finite-dimensional, \cref{eq:estimate} naturally holds, then as an application of \cref{thm:main_error_bound} and \cref{thm:main}, we have the following error bound  for fermionic impurity models:
\begin{cor}[Error bounds for  fermionic impurity models]
Let $\hat\rho_{\text S}(t)$, $\hat\rho_{\text S}'(t)$ be the reduced system density operators of a unitary or Lindblad fermionic impurity model. Let $\Delta^{\lessgtr}_{ij}(t)$ and ${\Delta^{\lessgtr}_{ij}}'(t)$ be the 
lesser and greater hybridization functions corresponding to $\hat\rho_{\text S}(t)$, $\hat\rho_{\text S}'(t)$, as defined by \cref{eq:c_t_fermion_unitary} or \cref{eq:c_t_fermion_Lindblad}. Then, 
 the same error bounds as in \cref{thm:main_error_bound} for reduced system densities and \cref{thm:main} for any bounded system observables $\hat O_{\text S}$ hold, i.e., \cref{eq:error_bound_application} holds, in which
 \begin{equation}
     \begin{aligned}
        \epsilon& = 4 n \left( \|\Delta^> - {\Delta^{>}}'\|_{L^\infty([0,T])} + \|\Delta^< - {\Delta^{<}}'\|_{L^\infty([0,T])}\right), \\ 
        \epsilon_1 & =
        4 n\left( \|\Delta^> - {\Delta^{>}}'\|_{L^1([0,T])} + \|\Delta^< - {\Delta^{<}}'\|_{L^1([0,T])}\right).
     \end{aligned}
    \end{equation}
    \label{cor:error_fermion}
\end{cor}

\section{Applications to quasi-Lindblad dynamics}
\label{sec:quasi-Lind}

In this section, we discuss the applications of \cref{thm:rhoS_Dyson_short} and \cref{thm:main_error_bound} to quasi-Lindblad dynamics.
For concreteness, we will focus on the quasi-Lindblad pseudomode introduced in \cite{ParkHuangZhuetal2024}, which preserves the Hermiticity of density operators without the CP condition.  Our analysis is also relevant in other pseudomode formulations involving non-Hermitian dynamics \cite{Lambert2019, Pleasance2020, Cirio2023, menczel2024nonhermitian}, and also for the hierarchical equations of motion (HEOM) method interpreted as quasi-Lindblad dynamics \cite{xu2023universal, ivander2024unifiedframeworkopenquantum}. 

All quasi-Lindblad dynamics to date may suffer from subtle instability issues \cite{Witt_2017,Dunn2019instabilities, Yan2020heomstability, LiYanShi2022, Krug2023,ParkHuangZhuetal2024}. In particular, the dynamics can be asymptotically stable but non-contractive~\cite{ParkHuangZhuetal2024}. Therefore, the differences between the two quasi-Lindblad dynamics could not be compared directly.
Nonetheless, by examining the proof of \cref{thm:main_error_bound} carefully, we find that it is sufficient for one of the dynamics to be contractive. In other words, the error analysis is still applicable when comparing a unitary or Lindblad dynamics, governed by $\boldsymbol{L}$, and a quasi-Lindblad dynamics, governed by $\boldsymbol{L}'$. On the other hand, when comparing two quasi-Lindblad dynamics, we need to introduce an intermediary contractive dynamics denoted by $\tilde{\boldsymbol{L}}$, and then compare the distance between the dynamics generated by  $\boldsymbol{L},\tilde{\boldsymbol{L}}$, and that by $\boldsymbol{L}',\tilde{\boldsymbol{L}}$, respectively.

\subsection{System-environment coupling}
Apart from the system-environment coupling Liouvillian $\boldsymbol{L}_{\text{SE}}^{\text U}$ (\cref{lem:Liouvillian_unitary_spin_boson} and \cref{eq:L_SE_U_fermion}), the quasi-Lindblad pseudomode theory \cite{ParkHuangZhuetal2024} introduces an additional system-environment coupling Liouvillian term, denoted as $\boldsymbol{D}_{\text{SE}}$, 
which takes a Lindblad-like dissipator form. In the spin-boson model, $\boldsymbol{D}_{\text{SE}}$ is given by
\begin{equation}
    \boldsymbol D_{\text{SE}}\hat\rho = \sum_{j=1}^n \hat{M}_j \hat{\rho} \hat \sigma_z^{(j)}|_{\mathcal H} + \hat \sigma_z^{(j)}|_{\mathcal H} \hat{\rho} \hat{M}^\dag_j - \frac{1}{2} \{\hat{\sigma_z}^{(j)}|_{\mathcal H} \hat{M}_j + \hat{M}^\dag_j \hat{\sigma_z}^{(j)}|_{\mathcal H} , \hat{\rho} \}, \quad \hat{M}_j = \sum_k 2 M_{jk} \hat b_k|_{\mathcal H},
    \label{eq:system_env_dissipation_spin_boson0}
\end{equation}
where $M$ is a $n\times N_{\mathrm{e}}$ complex coefficient matrix,  $\hat\sigma_z^{(j)}|_{\mathcal H}$ and $\hat b_k^{\dagger}|_{\mathcal H}$ are the extensions of $\hat \sigma_z^{(j)}$ and $\hat b_k^{\dagger}$ to $\mathcal H$, defined as
$$
    \hat\sigma_z^{(j)}|_{\mathcal H} =\hat 1_{\text e} \otimes \hat\sigma_z^{(j)},\quad \hat b_k^{\dagger}|_{\mathcal H} = \hat b_k^{\dagger}\otimes \hat 1_{\text s}.
$$
 This new term $\boldsymbol{D}_{\text{SE}}$  breaks the completely positive (CP) property. The superoperator form of $\boldsymbol{D}_{\text{SE}}$ is given by
\begin{equation}
    \begin{aligned}
\boldsymbol{D}_{\text{SE}} &= \sum_{\alpha=1}^{2n}
\boldsymbol E_{\alpha}^{\text D}\otimes \boldsymbol S_{\alpha},
\\
    \boldsymbol{E}_{2 j-1}^{\text{D}}&= \sum_{k=1}^{N_{\mathrm{e}}} \left(2 M_{jk}^*\widetilde{\boldsymbol{b}}_k-M_{jk}\boldsymbol{b}_k-M_{jk}^*\boldsymbol{b}_k^{\dagger}\right), \text{ } \boldsymbol{E}_{2j}^{\text{D}}=\sum_{k=1}^{N_{\mathrm{e}}} \left(2 M_{jk}\boldsymbol{b}_k-M_{jk}^*\widetilde{\boldsymbol{b}}_k-M_{jk}\widetilde{\boldsymbol{b}}_k^{\dagger}\right).\text{ } j=1,\cdots,n.
\end{aligned}
\label{eq:system_env_dissipation_spin_boson}
\end{equation}

In the fermionic case, the system-environment Liouvillian $\boldsymbol{D}_{\text{SE}}$ includes the coupling of the system with both parts of the environment, i.e., $\boldsymbol{D}_{\text{SE}}=\boldsymbol{D}_{\text{SE}}^-+\boldsymbol{D}_{\text{SE}}^+$, where $\boldsymbol{D}_{\text{SE}}^{\mp}$ has the following Lindblad-like structure:
   \begin{equation}
        \begin{aligned}
            \boldsymbol D_{\text{SE}}^-\hat\rho &= \sum_{i=1}^n\sum_{k=1}^{N_1}   M_{ik}^*  (2\hat a_i|_{\mathcal H}\hat\rho\hat c_k^\dagger|_{\mathcal H} - \{ \hat c_k^\dagger|_{\mathcal H}\hat a_i|_{\mathcal H},\hat\rho\})
            +M_{ik}  (2 \hat c_k |_{\mathcal H}\hat\rho \hat a_i^\dagger|_{\mathcal H} - \{ \hat a_i^\dagger|_{\mathcal H}\hat c_k|_{\mathcal H},\hat\rho\}),\\
            \boldsymbol D_{\text{SE}}^+\hat\rho &= \sum_{i=1}^n\sum_{l=N_1+1}^{N_e}  M_{il}  (2\hat a_i^\dagger|_{\mathcal H}\hat\rho\hat c_l|_{\mathcal H} - \{ \hat c_l|_{\mathcal H}\hat a_i^\dagger|_{\mathcal H},\hat\rho\}) + 
            M_{il}^*(2 \hat c_l^\dagger |_{\mathcal H}\hat\rho \hat a_i|_{\mathcal H} - \{ \hat a_i|_{\mathcal H}\hat c_l^\dagger|_{\mathcal H},\hat\rho\}),
        \end{aligned}
        \label{eq:D_SE_fermion_2}
    \end{equation}
    where $M$ is a $n\times N_e$ coefficient matrix.  Like in \cref{eq:L_SE_U_fermion}, $\boldsymbol{D}_{\text{SE}}^\mp$ could also be written in superoperator forms,
   \begin{equation}
           \boldsymbol{D}_{\text{SE}}^{\mp}  = \sum_{\alpha=1}^{4n} \boldsymbol{E}_{\alpha}^{\text{D}\mp}\otimes \boldsymbol{S}_{\alpha},
             \label{eq:D_SE_fermion}
   \end{equation}
  where both sides of \cref{eq:D_SE_fermion} is also understood to be superoperators acting on $\boldsymbol{L}^{(0)}(\mathcal H)$  (see \cref{rem:SE_decomposition}).   We refer to \cref{eq:E_alpha_fermionic_non_unitary} for expressions of environment superoperators $\boldsymbol{E}_{\alpha}^{\text{D}\mp}$ ($\alpha=1,\cdots,4n$), and \cref{lem:L_SE_D_SE_Expression} for proving that the operator and superoperator forms of $\boldsymbol{D}_{\text{SE}}^\mp$ are equivalent. For future convenience, let us define $M^-$, $M^+$ as sublocks of $M$:   $M = [M^-,M^+]$, where $M^-$ are of size $n\times N_1$, and $M^+$ are of size $n\times (N_{\text e}-N_1)$.

\subsection{Reduced system dynamics and error bound}

For the reduced system density operator of quasi-Lindblad dynamics, we have the following result for reduced system density operator, also as the corollary of \cref{thm:rhoS_Dyson_short}:
\begin{cor}[reduced system density operator of quasi-Lindblad dynamics] We have the following result for quasi-Lindblad dynamics:
\begin{itemize}
    \item In the spin-boson case, where $\boldsymbol L_{\text e} = \boldsymbol L_{\text e}^{\text U}+\boldsymbol{D}_{\text e}$ (see \cref{eq:spin_boson_env}), $\boldsymbol L_{\text{SE}} = \boldsymbol L_{\text{SE}}^{\text U}+\boldsymbol{D}_{\text{SE}}$ (see \cref{lem:Liouvillian_unitary_spin_boson}, \cref{eq:system_env_dissipation_spin_boson}), and with the initial vacuum environment \cref{eq:rhoE_init_empty}, we have that \cref{eq:rhoS_spin_boson} holds, where $c(t-t')$ is given by
\begin{equation}
   c(t-t') =  (g-\mathrm{i} M) \mathrm{e}^{(-\mathrm{i} H-\Gamma) (t-t')}\left(g^{\dagger}-\mathrm{i} M^{\dagger}\right), \quad t\geq t'.\label{eq:c_t_quasi_Lindblad_spin_boson}
\end{equation}
\item In the fermionic impurity case, where $\boldsymbol L_{\text e} = \boldsymbol L_{\text e}^{\text U}+\boldsymbol{D}_{\text e}$ (see \cref{eq:Le_U_fermionic,eq:D_e_fermion}), $\boldsymbol L_{\text{SE}} = \boldsymbol L_{\text{SE}}^{\text U}+\boldsymbol{D}_{\text{SE}}$ (see \cref{eq:L_SE_U_fermion,eq:D_SE_fermion}),with the initial environment density given by \cref{eq:rhoE_init_emp_fil}), we have that \cref{eq:rhoS_fermion} holds, where $\Delta^>(t)$ and $\Delta^<(t)$ are given by 
    \begin{equation}
        \begin{aligned}
            \Delta^>(t) &= (\nu^- - \mathrm iM^-)  \mathrm e^{(-\mathrm iH^--\Gamma^-)t}((\nu^-)^\dagger-\mathrm i (M^-)^\dagger),\\
        \Delta^<(t) &= (\nu^+ - \mathrm iM^+)\mathrm e^{(-\mathrm iH^+-\Gamma^+)t}((\nu^+)^\dagger - \mathrm i(M^+)^\dagger).
        \end{aligned}
        \label{eq:c_t_fermion_quasi_lind}
    \end{equation}
\end{itemize}
\label{cor:quasi_lind_dynamics}
\end{cor}
Similar to \cref{cor:spin_boson_dynamics} and \cref{cor:dynamics_fermion}, according to \cref{eq:rhoS_Dyson}, \cref{cor:quasi_lind_dynamics} holds as long as the following lemma for superoperator BCFs hold:
\begin{lem}[BCFs for quasi-Lindblad dynamics]
For BCFs of quasi-Lindblad dynamics, we have
\begin{itemize}
    \item For spin-boson model, \cref{eq:two_point_corr_unitary_spin_boson} holds, where $c(t-t')$ is given by \cref{eq:c_t_quasi_Lindblad_spin_boson};
    \item For fermionic impurity model, \cref{eq:two_point_corr_unitary_fermionic} holds, where $\Delta^<$, $\Delta^>$ are given by \cref{eq:c_t_fermion_quasi_lind}.
\end{itemize}
\label{lem:BCF_quasi_lind}
\end{lem}
We will prove \cref{lem:BCF_quasi_lind}
in \cref{sec:env_correlation}.
As we have emphasized above,
since the quasi-Lindblad dynamics break the contraction property, one can not compare two quasi-Lindblad dynamics directly. However, by introducing an intermediate contractive dynamics, we have the following result:
\begin{cor}[Error bound for quasi-Lindblad dynamics]
Let $\hat\rho_{\text S}(t), \hat\rho_{\text S}'(t)$ be the reduced system density operator of two quasi-Lindblad dynamics. In the case of spin-boson model, let $c_{ij}(t)$ and $c_{ij}'(t)$ be the 
correlation function corresponding to $\hat\rho_{\text S}(t)$, $\hat\rho_{\text S}'(t)$, as defined by \cref{eq:c_t_quasi_Lindblad_spin_boson}, and let $\widetilde c(t)$
be the BCFs of an arbitrary contractive (unitary or Lindblad) dynamics (\cref{eq:c_t_unitary_spin_boson} or \cref{eq:c_t_Lindblad_spin_boson}). Then \cref{eq:error_bound_application} holds,
where $\epsilon, \epsilon_1$ is defined as:
$$\epsilon = 4n  (\|c  - \widetilde c\|_{L^\infty([0,T])} + \| c '-\widetilde c\|_{L^\infty([0,T])}),\quad \epsilon_1 = 4n   (\|c  - \widetilde c \|_{L^1([0,T])}+\|c'  - \widetilde c '\|_{L^1([0,T])}).$$
Simiarly, in the fermionic impurity case,  let$\Delta^{\lessgtr}_{ij}(t)$ and ${\Delta^{\lessgtr}_{ij}}'(t)$ be the 
lesser and greater hybridization functions corresponding to $\hat\rho_{\text S}(t)$, $\hat\rho_{\text S}'(t)$, as defined by 
\cref{eq:c_t_fermion_quasi_lind}, and let $\widetilde\Delta^{\lessgtr}_{ij}(t)$ be the hybridization functions corresponding to any unitary or Lindblad dynamics (\cref{eq:c_t_fermion_unitary} or \cref{eq:c_t_fermion_Lindblad}). Then \cref{eq:error_bound_application} holds, where  $\epsilon, \epsilon_1$ is defined as:
$$
 \begin{aligned}
        \epsilon& = 4 n \left( \|\Delta^> - {\widetilde \Delta^{>}}\|_{L^\infty([0,T])} + \|\Delta^< - {\widetilde \Delta^{<}}\|_{L^\infty([0,T])}+\|{\Delta^{>}}' - {\widetilde \Delta^{>}}\|_{L^\infty([0,T])} + \|{\Delta^{<}}' - {\widetilde \Delta^{<}}\|_{L^\infty([0,T])}\right), \\ 
        \epsilon_1 & =
        4 n\left( \|\Delta^> - {\widetilde \Delta^{>}} \|_{L^1([0,T])} + \|\Delta^< - {\widetilde \Delta^{<}} \|_{L^1([0,T])}+\|{\Delta^>}' - {\widetilde \Delta^{>}} \|_{L^1([0,T])} + \|{\Delta^<}'- {\widetilde \Delta^{<}} \|_{L^1([0,T])}\right).
     \end{aligned}
$$
\end{cor}
\begin{proof}
    Let us prove the case of spin-boson models while the case of fermionic impurity follows similarly. 
    Let $\widetilde \rho_{\text S}(t)$ be the reduced system density operator corresponding to $\widetilde c(t)$.
    Since \cref{thm:main} could be used for the contractive dynamics and the quasi-Lindblad dynamics, we have:
    $$
    \|\hat\rho_{\text S}(t) - \widetilde{\rho}_{\text S}(t)\|_{\text{tr}}\leq \mathrm e^{\widetilde{\epsilon}t^2/2}-1,\quad
    \|\hat\rho_{\text S}'(t) - \widetilde{\rho}_{\text S}(t)\|_{\text{tr}}\leq \mathrm e^{\widetilde{\epsilon}'t^2/2}-1,\quad 
    $$
    where  $\widetilde\epsilon = 4n  (\|c  - \widetilde c\|_{L^\infty([0,T])} $ and $\widetilde\epsilon'=\| c '-\widetilde c\|_{L^\infty([0,T])}$.
    Thus, 
    $$
    \begin{aligned}
    &\|\hat\rho_{\text S}(t) - \hat\rho_{\text S}'(t)\|_{\text{tr}}\leq 
    \|\hat\rho_{\text S}(t) - \widetilde{\rho}_{\text S}(t)\|_{\text{tr}}+
    \|\hat\rho_{\text S}'(t) - \widetilde{\rho}_{\text S}(t)\|_{\text{tr}}
    \\\leq & \mathrm e^{\widetilde{\epsilon}t^2/2}-1 + \mathrm e^{\widetilde{\epsilon}'t^2/2}-1 = \sum_{n=1}^{\infty}(\widetilde{\epsilon}t^2/2)^n + (\widetilde{\epsilon}'t^2/2)^n/n!\leq \sum_{n=1}^{\infty}((\widetilde{\epsilon}+\widetilde{\epsilon}')t^2/2)^n = \mathrm e^{{\epsilon}t^2/2}-1.
    \end{aligned}
       $$
       The error bound in the $L^1$ sense holds in a similar way.
\end{proof}
\section*{Conflict of interest}
The authors declare that they have no conflict of interest.
\bibliographystyle{alpha}
\bibliography{ref}
\appendix

\section{\texorpdfstring{$\mathbb Z_2$}{Z2}-graded tensor product for fermionic systems}
\label{sec:Z2}
In fermionic cases, 
the Hilbert space $\mathcal H_{\text S}$ and $\mathcal H_{\text E}$ could be easily identified as $\mathbb Z_2$-graded vector space, i.e.,
$$
\mathcal H_{\text S} = \mathcal H_{\text S}^{(0)} \oplus \mathcal H_{\text S}^{(1)} ,\quad \mathcal H_{\text E} = \mathcal H_{\text E}^{(0)} \oplus \mathcal H_{\text E}^{(1)},
$$
where $\mathcal H_{\text S}^{(0)}$, $\mathcal H_{\text S}^{(1)}$ are defined as
$$
\mathcal H_{\text S}^{(0)} = \{\psi_{\text s}\in \mathcal H_{\text S} | \hat P_{\text s}|\psi_{\text s}\rangle = |\psi_{\text s}\rangle\},\quad \mathcal H_{\text S}^{(1)} = \{\psi_{\text s}\in \mathcal H_{\text S} | \hat P_{\text s}|\psi\rangle_{\text s} = -|\psi\rangle_{\text s}\},
$$
and $\mathcal H_{\text E}^{(0)}$, $\mathcal H_{\text E}^{(1)}$ are similarly defined. Intuitively, $\mathcal H_{\text S,\text E}^{(0)}$ is the subspace of wavefunctions that has even particle numbers, while $\mathcal H_{\text S,\text E}^{(1)}$ is the subspace of wavefunctions that has odd particle numbers.
An operator $\hat O_{\text s}$ acting on $\mathcal H_{\text S}$ is said to be of even parity, if $\hat O_{\text s}$ maps $\mathcal H_{\text S}^{(j)}$ to $\mathcal H_{\text S}^{(j)}$, for $j=0,1$, and is said to be of odd parity if  $\hat O_{\text s}$ maps $\mathcal H_{\text S}^{(j)}$ to $\mathcal H_{\text S}^{((j')}$ for $j=0,1$, where $j' = (j+1)\mod 2$. If $\hat O_{\text s}$ is neither even nor odd,
it is said to be of indefinite parity otherwise.
Similarly one can define parity of operators acting on $\mathcal H_{\text E}$.
The fermionic creation and annihilation operators, $\hat c_i,\hat c_i^\dagger\in B(\mathcal H_{\text E})$, $\hat a_j,\hat a_j^\dagger\in B(\mathcal H_{\text S})$ are of odd parity, while the number operators are of even parity.
Note that in this work, all operators that have been used are of definite parity.

The standard $\mathbb Z_2$-graded tensor product, which preserves the CAR,  is defined as follows \cite{Berezin2013}:
\begin{defn}[$\mathbb Z_2$-graded tensor product]
For any $\hat O_{\text e}$ acting on $\mathcal H_{\text E})$ and $\hat O_{\text s}$ acting on $\mathcal H_{\text S}$ of definitive parity, the $\mathbb Z_2$-graded tensor product of $\hat O_{\text e}$ and $\hat O_{\text s}$, denoted as $\hat O_{\text e}\otimes_g \hat O_{\text s}$, is an operator acting on $\mathcal H$ ($\mathcal H = \mathcal H_{\text E}\otimes \mathcal H_{\text S}$), defined as follows:
\begin{equation}
\left(\hat{O}_{\mathrm{e}} \otimes_g \hat{O}_{\mathrm{s}}\right)\left( \psi_{\mathrm{e}} \otimes \psi_{\mathrm{s}} \right)=(-1)^{\sigma\left(\hat{O}_{\mathrm{s}}\right) \sigma'\left( \psi_{\mathrm{e}} \right)}\left(\left(\hat{O}_{\mathrm{e}} \psi_{\mathrm{e}} \right) \otimes\left(\hat{O}_{\mathrm{s}} \psi_{\mathrm{s}} \right)\right).
\end{equation}
Here $\sigma(\hat O_{\text s})=0$ if $\hat O_{\text s}$ is even and $\sigma(\hat O_{\text s})=1$ if $\hat O_{\text s}$ is odd, and $\sigma'(\psi_{\text e})$ are defined similarly.
\label{defn:Z2}
\end{defn}

Thus, the generalization of environment and system operators is defined using $\mathbb Z_2$-graded tensor product with identity operator:
$$
\hat O_{\text e}|_{\mathcal H}=\hat O_{\text e}\otimes_g \hat 1_{\text s},\quad 
\hat O_{\text s}|_{\mathcal H}=\hat 1_{\text e}\otimes_g \hat O_{\text s}.
$$
Using \cref{defn:Z2}, we have $\hat O_{\text e}|_{\mathcal H} = \hat O_{\text e}\otimes \hat 1_{\text s}$ for both even and odd $\hat O_{\text e}$, while $\hat O_{\text s}|_{\mathcal H} =\hat P_{\text e}\otimes \hat O_{\text s}$ if $\hat O_{\text s}$ is odd, and $\hat O_{\text s}|_{\mathcal H} =\hat 1_{\text e}\otimes \hat O_{\text s}$ if $\hat O_{\text s}$ is even.
From this, the fermionic extension in the main text \cref{eq:fermionic_extension} naturally follows.

\section{Proof of Eq. \ref{eq:the_thing_need_to_be_proved} in the fermionic case}
\label{sec:proof_correlator_return_to_extended_system_fermionic}
In this appendix, we prove \cref{eq:the_thing_need_to_be_proved} in the fermionic case, which we restate here for convenience:
\begin{equation}
    \begin{aligned}
      &  \operatorname{tr}_{\mathrm{E}}\left(\mathrm{e}^{\boldsymbol{L}_0 t} \sum_{m=0}^{\infty} \frac{1}{m!} \int_0^t \cdots \int_0^t \mathrm{~d} s_1 \cdots \mathrm{d} s_m \boldsymbol{\mathcal T}\left(\prod_{i=1}^{2 n} \widetilde{\boldsymbol{S}}_{\alpha_i}\left(t_i\right) \prod_{j=1}^m \boldsymbol{L}_{\mathrm{SE}}\left(s_j\right)\right) \hat{\rho}(0)\right)\\
    &  \begin{aligned}
        =(-1)^{\sigma_t}&\sum_{m=0}^\infty\frac{\mathrm e^{\boldsymbol{L}_{\text s}t}}{m!}\int_0^t \mathrm ds_1\cdots\int_0^t\mathrm ds_m \\
        &\qquad \qquad
        \sum_{\beta_1, \cdots, \beta_{m}=1}^N  C^{(m)}_{\beta_1,\cdots,\beta_{m}}(s_1,\cdots,s_{m})
        \mathbf T^{\pm}\left(\prod_{i=1}^{2n} {\boldsymbol{S}}_{\alpha_i}(t_i)
        \prod_{j=1}^{m} \boldsymbol{S}_{\beta_j}(s_j)\right)\hat\rho_{\text S}(0),
        \label{eq:return_to_extended_system_fermionic}
        \end{aligned}
    \end{aligned}
\end{equation}
where $\sigma_t$ is the permutation of $1,\cdots, 2n$ such that 
$t_{\sigma_t(1)}\geq t_{\sigma_t(2)}\leq \cdots \geq t_{\sigma_t(2n)}$. In other words, we need to introduce a definition of $\widetilde{\boldsymbol S}_{\alpha}$ in the fermionic case, such that $\|\widetilde{\boldsymbol S}_{\alpha}\|_{\text{tr}}=1$ and \cref{eq:return_to_extended_system_fermionic} holds.
To achieve this, let us first introduce the environment parity superoperator $\boldsymbol P_{\text e}$, defined as follows:
\begin{equation}
    \boldsymbol{P}_{\text e} : \hat\rho \rightarrow \hat P_{\text e}\hat\rho\hat P_{\text e}.
\end{equation}
where $\hat P_{\text e} $ is the environment parity operators as defined in \cref{eq:parity_operator}.
Therefore, we have for any $\hat \rho \in B_1(\mathcal H_{\text E})$,
$$
\boldsymbol c_k \boldsymbol P_{\text e} \hat\rho = \boldsymbol c_k (\hat P_{\text e}\hat\rho\hat P_{\text e})
 = \hat P_{\text e}\hat c_k\hat P_{\text e}\hat\rho\hat P_{\text e} = -\hat c_k \hat \rho\hat P_{\text e} = - \hat P_{\text e}(\hat P_{\text e}\hat c_k\hat\rho) \hat P_{\text e} = -\boldsymbol P_{\text e}\boldsymbol c_k\hat\rho,
$$
$$
\widetilde{\boldsymbol c}_k \boldsymbol P_{\text e} \hat\rho = \widetilde{\boldsymbol c}_k (\hat P_{\text e}\hat\rho\hat P_{\text e})
 = \hat P_{\text e}\hat P_{\text e}\hat\rho\hat P_{\text e}\hat c_k^\dagger 
 = - \hat P_{\text e}\hat P_{\text e}\hat \rho \hat c_k^\dagger \hat P_{\text e} = -\boldsymbol P_{\text e}\widetilde{\boldsymbol c}_k\hat\rho,
$$
where $\boldsymbol c_k$ and $\widetilde{\boldsymbol c}_k$ are defined in \cref{eq:fermionic_env_superoperators}.
In other words, 
$\{ \boldsymbol c_k, \boldsymbol P_{\text e}\} = 0$ and $\{\widetilde{\boldsymbol c}_k, \boldsymbol P_{\text e}\} = 0$.
Similarly, 
$\{ \boldsymbol c_k^\dagger, \boldsymbol P_{\text e}\} = 0$ and $\{\widetilde{\boldsymbol c}_k^\dagger, \boldsymbol P_{\text e}\} = 0$.
Since $\boldsymbol L_{\text e}$ is quadratic in $\boldsymbol c_k, \boldsymbol c_k^\dagger, \widetilde{\boldsymbol c}_k, \widetilde{\boldsymbol c}_k^\dagger$, we have $\boldsymbol L_{\text e}\boldsymbol P_{\text e} = \boldsymbol P_{\text e}\boldsymbol L_{\text e}$.
Then, let us define the superoperators $\widetilde{\boldsymbol S}_{\alpha}$ as follows:
\begin{equation}
    \widetilde{\boldsymbol S}_{\alpha} = \boldsymbol P_{\text e}\otimes\boldsymbol S_{\alpha}.
\end{equation}
Since $\|\boldsymbol P_{\text e}\|_{\text{tr}}=1$, we have $\|\widetilde{\boldsymbol S}_{\alpha}\|_{\text{tr}}=1$.
Without loss of generality, let us assume that $t_1\geq t_2\geq \cdots \geq t_{2n}$.
For $s_1,\cdots,s_m\in[0,T]$, let $\sigma_s$
be the permutation of $1,\cdots, m$ such that $s_{\sigma_s(1)}\geq s_{\sigma_s(2)}\geq \cdots \geq s_{\sigma_s(m)}$.
For $t_1,\cdots, t_{2n} \in [0,t]$
and $s_1,\cdots, s_{m}\in [0,t]$, let $w_i = t_i$ ($i=1,\cdots,2n$) and $w_{2n+i} = s_{i}$ ($i=1,\cdots,m$). Let 
us define permutation $\widetilde\sigma\in S_{2n+m}$ such that $w_{\widetilde\sigma(1)}\geq w_{\widetilde\sigma(2)}\geq \cdots \geq w_{\widetilde\sigma(2n+m)}$.
Then, we have
\begin{equation}
    \begin{aligned}
        & \mathrm e^{\boldsymbol L_0t}\boldsymbol{\mathcal T}\left( \prod_{i=1}^{2n} \widetilde{\boldsymbol S}_{\alpha_i}\left(t_i\right) \prod_{j=1}^m \boldsymbol L_{\text{SE}}\left(s_j\right)\right) \hat\rho(0)
       \\ =&\mathrm e^{\boldsymbol L_0t}\sum_{\beta_1,\cdots\beta_m=1}^N(-1)^{\widetilde\sigma}\boldsymbol{\mathcal T}\left(\prod_{i=1}^{2n}\boldsymbol P_{\text e}(t_i)\prod_{j=1}^m \boldsymbol E_{\beta_j}\left(s_j\right)\right)\otimes(-1)^{\widetilde\sigma}\boldsymbol{\mathcal T}\left( \prod_{i=1}^{2n} \widetilde{\boldsymbol S}_{\alpha_i}\left(t_i\right) \prod_{j=1}^m \boldsymbol S_{\beta_j}\left(s_j\right)\right) \hat\rho(0)\\
       =& \mathrm e^{\boldsymbol L_0t}\sum_{\beta_1,\cdots\beta_m=1}^N\mathbf T^-\left(\prod_{i=1}^{2n}\boldsymbol P_{\text e}(t_i)\prod_{j=1}^m \boldsymbol E_{\beta_j}\left(s_j\right)\right)\otimes\mathbf T^-\left( \prod_{i=1}^{2n} \widetilde{\boldsymbol S}_{\alpha_i}\left(t_i\right) \prod_{j=1}^m \boldsymbol S_{\beta_j}\left(s_j\right)\right) \hat\rho(0)\\
         =& \sum_{\beta_1,\cdots\beta_m=1}^N\mathrm e^{\boldsymbol{L}_{\text e} t}\mathbf T^-\left(\prod_{i=1}^{2n}\boldsymbol P_{\text e}(t_i)\prod_{j=1}^m \boldsymbol E_{\beta_j}\left(s_j\right)\right)\hat\rho_{\text e}^{(0)}\otimes\mathrm e^{\boldsymbol{L}_{\text s} t}\mathbf T^-\left(
            \prod_{i=1}^{2n} \widetilde{\boldsymbol S}_{\alpha_i}\left(t_i\right) \prod_{j=1}^m \boldsymbol S_{\beta_j}\left(s_j\right)\right) \hat\rho_{\text s}^{(0)}.
    \end{aligned}
\end{equation}
Here $\boldsymbol P_{\text e}(t) = \mathrm e^{-\boldsymbol L_{\text e}t}
\boldsymbol P_{\text e}\mathrm e^{\boldsymbol L_{\text e}t} = \boldsymbol P_{\text e}$ since $\boldsymbol L_{\text e}$ and $\boldsymbol P_{\text e}$ commute, and $t_i$ serves as a time argument for performing the time-ordering operation.
Let us define $\sigma_s \in S_m$ such that $s_{\sigma_s(1)}\geq s_{\sigma_s(2)}\geq \cdots \geq s_{\sigma_s(m)}$.
Furthermore, let us define $\sigma_{st}$ as the permutation of $1,\cdots, 2n+m$ such that if letting $w_1 = t_1,\cdots,w_{2n} =t_{2n}, w_{2n+1} = s_1,\cdots,w_{2n+m} = s_m$,
then $w_{\sigma_{st}(1)}\geq w_{\sigma_{st}(2)}\geq \cdots \geq w_{\sigma_{st}(2n+m)}$.
Since  $t_1\geq \cdots t_{2n}$, we have 
$(-1)^{\widetilde\sigma} = (-1)^{\sigma_s}(-1)^{\sigma_{st}}$.
Let $n_0$ be the number of $i$ such that $t_i \in [s_{\sigma_s(1)},t]$, 
$n_m$ be the number of $i$ such that $t_i \in [0,s_{\sigma_s(m)})$, and $n_j$ be the number of $i$ such that $t_i\in [s_{\sigma_s(j)},s_{\sigma_s(j+1)})$ for $j=1,\cdots,m-1$.
Thus we have 
$$
\begin{aligned}
   & \mathrm e^{\boldsymbol{L}_{\text e} t}\mathbf T^-\left(
    \prod_{i=1}^{2n}\boldsymbol P_{\text e}(t_i)\prod_{j=1}^m \boldsymbol E_{\beta_j}\left(s_j\right)\right)\hat\rho_{\text e}^{(0)} \\= &\mathrm e^{\boldsymbol{L}_{\text e} t}
    (-1)^{\widetilde\sigma} (\boldsymbol P_{\text e})^{n_0}
    \boldsymbol E_{\beta_{\sigma_s(1)}}\left(s_{\sigma_s(1)}\right)(\boldsymbol P_{\text e})^{n_1}\cdots \boldsymbol E_{\beta_{\sigma_s(m)}}\left(s_{\sigma_s(m)}\right)(\boldsymbol P_{\text e})^{n_m}
    \hat\rho_{\text e}^{(0)} \\
    =& \mathrm e^{\boldsymbol{L}_{\text e} t}(-1)^{\sigma_s} (\boldsymbol P_{\text e})^{2n} \boldsymbol E_{\beta_{\sigma_s(1)}}\left(s_{\sigma_s(1)}\right)\cdots \boldsymbol E_{\beta_{\sigma_s(m)}}\left(s_{\sigma_s(m)}\right)\hat\rho_{\text e}^{(0)}
    \\=&\mathrm e^{\boldsymbol{L}_{\text e} t} (-1)^{\sigma_s}  \boldsymbol E_{\beta_{\sigma_s(1)}}\left(s_{\sigma_s(1)}\right)\cdots \boldsymbol E_{\beta_{\sigma_s(m)}}\left(s_{\sigma_s(m)}\right)\hat\rho_{\text e}^{(0)}
    =\mathrm e^{\boldsymbol{L}_{\text e} t}\mathbf T^- \left( 
        \prod_{j=1}^m \boldsymbol E_{\beta_{j}}\left(s_{j}\right)
    \right)\hat\rho_{\text e}^{(0)}
    .
\end{aligned}
$$
Here we have used the anti-commutation relation $\{\boldsymbol P_{\text e}, \boldsymbol E_{\beta}\} = 0$ for all $\beta = 1,\cdots, N$ and $(\boldsymbol P_{\text e})^2 = \boldsymbol 1_{\text e}$.
Then
we have 

$$
\operatorname{tr}\left(\mathrm{e}^{\boldsymbol{L}_{\mathrm{e}} t} \mathbf{T}^{-}\left(\prod_{i=1}^{2 n} \boldsymbol{P}_{\mathrm{e}}\left(t_i\right) \prod_{j=1}^m \boldsymbol{E}_{\beta_j}\left(s_j\right)\right) \hat{\rho}_{\mathrm{e}}^{(0)}\right)
 = \operatorname{tr}\left(\mathrm{e}^{\boldsymbol{L}_{\mathrm{e}} t}\mathbf{T}^{-}\left(\prod_{j=1}^m \boldsymbol{E}_{\beta_j}\left(s_j\right)\right) \hat{\rho}_{\mathrm{e}}^{(0)}\right)
  = C^{(m)}_{\beta_1,\cdots,\beta_{m}}(s_1,\cdots,s_{m}).
$$
Therefore, \cref{eq:return_to_extended_system_fermionic} holds. 
\section{Liouvillians and BCFs}
\label{sec:env_correlation}
The main goal of this section is to give a proof for \cref{lem:corr_spin_boson} and \cref{lem:corr_fermionic}, from which the reduced system density dynamics (see \cref{cor:spin_boson_dynamics},\cref{cor:dynamics_fermion}) and the error bound (see \cref{cor:error_spin_boson}, \cref{cor:error_fermion}) follow naturally.
We first introduce several useful lemmas. Then, in \cref{sec:spin_boson_proof}, we prove \cref{lem:corr_spin_boson} for unitary, Lindblad, and quasi-Lindblad cases. In \cref{sec:fermion_proof}, we first give the explicit expressions for $\boldsymbol{E}_{\alpha}^{\text D \mp}$  in \cref{eq:D_SE_fermion}, then we prove \cref{lem:corr_fermionic} also for each cases.
\subsection*{Some useful lemmas}
Let us introduce the shorthand notation $\langle \boldsymbol O\rangle_{\text e} :=
\operatorname{tr}(\boldsymbol O\hat\rho_{\text E}(0))$ for $\boldsymbol O$ acting on $B_1(\mathcal H_{\text e})$, and $\langle \hat O\rangle_{\text e}: = \operatorname{tr}(\hat O\hat\rho_{\text E}(0))$ for $\hat O$ acting on $\mathcal H_{\text e}$.
For any two environment  superoperators $\boldsymbol O_1, \boldsymbol O_2$, and $t>t_1>t_2>0$,  consider the following two-point correlation functions:
\begin{equation}
 \operatorname{tr}(\mathrm e^{\boldsymbol L_{\text e} (t-t_1)}\boldsymbol O_1\mathrm e^{\boldsymbol L_{\text e} (t_1-t_2)}\boldsymbol O_2\mathrm e^{\boldsymbol L_{\text e} t_2}\hat\rho_{\text E}(0)).
\end{equation}
Since $\hat\rho_{\text E}(0)$ is the stationary state of $\boldsymbol L_{\text e}$ and $\mathrm e^{\boldsymbol L_{\text e} t}$ is a trace-preserving map, then the above correlation function is equal to
$$
\langle \boldsymbol O_1\mathrm e^{\boldsymbol L_{\text e} (t_1-t_2)}\boldsymbol O_2\rangle_{\text e} = 
 \operatorname{tr}(\boldsymbol O_1\mathrm e^{\boldsymbol L_{\text e} (t_1-t_2)}\boldsymbol O_2\hat\rho_{\text E}(0)),
$$
which only relies on $(t_1-t_2)$. 
Thus let us take $(t_1,t_2)=(t,0)$. Then taking derivative with respect to $t$, we have
$$
\begin{aligned}
    \frac{\mathrm d}{\mathrm dt}\langle \boldsymbol O_1\mathrm e^{\boldsymbol L_{\text e} t}\boldsymbol O_2\rangle_{\text e}
= \operatorname{tr}(\boldsymbol O_1\mathrm e^{\boldsymbol L_{\text e} t}\boldsymbol L_{\text e}\boldsymbol O_2\hat\rho_{\text E}(0))
&= \operatorname{tr}(\boldsymbol O_1\mathrm e^{\boldsymbol L_{\text e} t}[\boldsymbol L_{\text e}, \boldsymbol O_2]\hat\rho_{\text E}(0))+
\operatorname{tr}( \boldsymbol O_1\mathrm e^{\boldsymbol L_{\text e} t}\boldsymbol O_2\boldsymbol L_{\text e}\hat\rho_{\text E}(0)) \\
&=\operatorname{tr}(\boldsymbol O_1\mathrm e^{\boldsymbol L_{\text e} t}[\boldsymbol L_{\text e}, \boldsymbol O_2]\hat\rho_{\text E}(0)),
\end{aligned}
$$
where we have used again that $\hat\rho_{\text E}(0)$ is the stationary state of $\boldsymbol L_{\text e}$.

In the unitary case, the following results are well-known and useful (for example, see \cite{NegeleOrland1998}):
\begin{lem}
    In the bosonic case, let $\hat H_{\text e} = \sum_{k,k'=1}^{N_{\text e}}H_{kk'}\hat b_k^\dagger\hat b_{k'}$. With environment density operator being \cref{eq:rhoE_init},
     we have
    $$
    \langle\hat b_l^\dagger\hat b_k\rangle_{\text e} = \left((\mathrm e^{\beta H}- I)^{-1}\right)_{kl},\quad 
    \langle\hat b_k\hat b_l^\dagger\rangle_{\text e} = ((I-\mathrm e^{-\beta H})^{-1})_{kl},
    $$
    Similarly, in the fermionic cases, by replacing $\hat b_k, \hat b_k^\dagger$ with fermionic operators $\hat c_k,\hat c_k^\dagger$, we have 
    $$
    \langle\hat c_l^\dagger\hat c_k\rangle_{\text e} = \left((\mathrm e^{\beta H}+ I)^{-1}\right)_{kl},\quad 
    \langle\hat c_k\hat c_l^\dagger\rangle_{\text e} = ((I+\mathrm e^{-\beta H})^{-1})_{kl}.
    $$
\end{lem}
For both  bosonic and fermionic $\hat H_{\text e}$, we 
 will also use the following commutator relations:
\begin{lem}
    In the unitary case, with $\boldsymbol L_{\text e}^{\text U}$ in \cref{eq:spin_boson_env} and \cref{eq:Le_U_fermionic}, we have:
    \begin{equation}
        [\widetilde{\mathbf b}_l, \boldsymbol L_{\text e}^{\text U}] = \sum_{k=1}^{N_{\text e}} \mathrm i H_{kl}\widetilde{\mathbf b}_{k},\quad [\widetilde{\mathbf b}_l^\dagger, \boldsymbol L_{\text e}^{\text U}] = \sum_{k=1}^{N_{\text e}}-\mathrm i H_{lk} \widetilde{\mathbf b}_k^\dagger.
        \label{eq:commutator_Le_H}
    \end{equation}
    where $\mathbf b_{k}, \mathbf b_{k}^\dagger, \widetilde{\mathbf b}_{k}, \widetilde{\mathbf b}_{k}^\dagger$ are understood to be $\boldsymbol b_k, \boldsymbol b_k^\dagger, \widetilde{\boldsymbol b}_k, \widetilde{\boldsymbol b}_k^\dagger$ in the bosonic case, and 
    $\boldsymbol c_k, \boldsymbol c_k^\dagger, \widetilde{\boldsymbol c}_k, \widetilde{\boldsymbol c}_k^\dagger$ in the fermionic case.
    We also have:
    $$[\mathbf b_l, \widetilde{\mathbf b}_k{\mathbf b}_{k'}] = 0 ,\quad  [\widetilde{\mathbf b}_l, \widetilde{\mathbf b}_{k}\mathbf b_{k'}] = 0,\quad [\mathbf b_l^\dagger, \widetilde{\mathbf b}_k\mathbf b_{k'}] = -\delta_{k'l}\widetilde{\mathbf b}_{k},\quad[\widetilde{\mathbf b}_l^\dagger, \widetilde{\mathbf b}_k\mathbf b_{k'}] = \mp\delta_{kl}\mathbf b_{k'}.$$
    $$
    [\mathbf b_l^\dagger, \widetilde{\mathbf b}_{k'}^\dagger{\mathbf b}_{k}^\dagger] = 0,\quad [\widetilde{\mathbf b}_l^\dagger, \widetilde{\mathbf b}_{k'}^\dagger{\mathbf b}_{k}^\dagger] = 0,\quad 
    [\mathbf b_l, \widetilde{\mathbf b}_{k'}^\dagger{\mathbf b}_{k}^\dagger] = \pm \delta_{kl}\widetilde{\mathbf b}_{k'}^\dagger,\quad [\widetilde{\mathbf b}_l, \widetilde{\mathbf b}_{k'}^\dagger{\mathbf b}_{k}^\dagger] = \delta_{k'l}\mathbf b_{k}^\dagger. 
    $$
    \label{lem:commutator_L}
\end{lem}
\begin{proof}
    Let us introduce the notation $[\cdot,\cdot]_{\mp}$, where $[\cdot,\cdot]_{-} = [\cdot,\cdot]$ and $[\cdot,\cdot]_{+} = \{\cdot,\cdot\}$.
    Using the CCR relation \cref{lem:CCR} and the CAR relation \cref{lem:CAR}, we have $[{\mathbf b}_l, \widetilde{\mathbf b}_k^\dagger \widetilde{\mathbf b}_{k'}] = 0$, $[{\mathbf b}_l^\dagger, \widetilde{\mathbf b}_k^\dagger \widetilde{\mathbf b}_{k'}] = 0$, $[\widetilde{\mathbf b}_l, {\mathbf b}_k^\dagger {\mathbf b}_{k'}] = 0$, $[\widetilde{\mathbf b}_l^\dagger, {\mathbf b}_k^\dagger {\mathbf b}_{k'}] = 0$, $[{\mathbf b}_l, \widetilde{\mathbf b}_k{\mathbf b}_{k'}] = 0$, $[\widetilde{\mathbf b }_l, \widetilde{\mathbf b}_k{\mathbf b}_{k'}] = 0$, $[{\mathbf b}_l^\dagger, \widetilde{\mathbf b}_{k'}^\dagger {\mathbf b}_{k}^\dagger] = 0$, $[\widetilde{\mathbf b}_l^\dagger, \widetilde{\mathbf b}_{k'}^\dagger {\mathbf b}_{k}^\dagger] = 0$,
    and
    $$
    \begin{aligned}
       &[ {\mathbf b}_l, {\mathbf b}_k^\dagger {\mathbf b}_{k'}] = [ {\mathbf b}_l, {\mathbf b}_k^\dagger]_{\mp} {\mathbf b}_{k'}\pm{\mathbf b}_k^\dagger[ {\mathbf b}_l,  {\mathbf b}_{k'}]_{\mp} = \delta_{lk}{\mathbf b}_{k'},\quad 
    [\boldsymbol{\mathbf b}_l^\dagger, 
    \boldsymbol{\mathbf b}_k^\dagger \boldsymbol{{\mathbf b}}_{k'}] = [ \boldsymbol{\mathbf b}_l^\dagger, \boldsymbol{\mathbf b}_k^\dagger]_{\mp} \boldsymbol{\mathbf b}_{k'}\pm\boldsymbol{\mathbf b}_k^\dagger[ \boldsymbol{\mathbf b}_l^\dagger,  \boldsymbol{\mathbf b}_{k'}]_{\mp} = -\delta_{lk'}\boldsymbol{\mathbf b}_{k}^\dagger. \\ 
    &[\widetilde{\mathbf b}_l, \widetilde{\mathbf b}_k^\dagger \widetilde{\mathbf b}_{k'}] = [\widetilde{\mathbf b}_l, \widetilde{\mathbf b}_k^\dagger]_\mp \widetilde{\mathbf b}_{k'}\pm\widetilde{\mathbf b}_k^\dagger[\widetilde{\mathbf b}_l,  \widetilde{\mathbf b}_{k'}] = \delta_{lk}\widetilde{{\mathbf b}}_{k'},\quad 
    [\widetilde{\mathbf b}_l^\dagger, \widetilde{\mathbf b}_k^\dagger \widetilde{\mathbf b}_{k'}] = [\widetilde{\mathbf b}_l^\dagger, \widetilde{\mathbf b}_k^\dagger]_\mp \widetilde{\mathbf b}_{k'}\pm\widetilde{\mathbf b}_k^\dagger[\widetilde{\mathbf b}_l^\dagger,  \widetilde{\mathbf b}_{k'}] = -\delta_{lk'}\widetilde{\mathbf b}_{k}^\dagger.
    \end{aligned}
    $$
    With $\boldsymbol L_{\text e}^{\text U} = \sum_{k,k'=1}^{N_{\text e}}-\mathrm i H_{kk'}{\mathbf b}_k^\dagger{\mathbf b}_{k'}+\mathrm i H_{kk'}\widetilde{\mathbf b}_{k'}^\dagger\widetilde{\mathbf b}_k$, \cref{eq:commutator_Le_H} is proved.
    Furthermore, we have 
    $$
    [\mathbf b_l^\dagger, \widetilde{\mathbf b}_k \mathbf b_{k'}] =  \pm\widetilde{\mathbf b}_k\mathbf b_l^\dagger\mathbf b_{k'} -  \widetilde{\mathbf b}_k \mathbf b_{k'}\mathbf b_l^\dagger = -\widetilde{\mathbf b}_k[\mathbf b_{k'}, \mathbf b_l^\dagger]_{\mp} = -\widetilde{\mathbf b}_k\delta_{k'l},
    $$
    $$
    [\widetilde{\mathbf b}_l^\dagger, \widetilde{\mathbf b}_k \mathbf b_{k'}] = \widetilde{\mathbf b}_l^\dagger \widetilde{\mathbf b}_k \mathbf b_{k'} \mp \widetilde{\mathbf b}_k \widetilde{\mathbf b}_l^\dagger\mathbf b_{k'}=\mp [\widetilde{\mathbf b}_k, \widetilde{\mathbf b}_l^\dagger]_{\mp}\mathbf b_{k'} = \mp \delta_{kl} {\mathbf b}_{k'},
    $$
    $$
    [\mathbf b_l, \widetilde{\mathbf b}_{k'}^\dagger \mathbf b_{k}^\dagger] = \pm \widetilde{\mathbf b}_{k'}^\dagger\mathbf b_l \mathbf b_{k}^\dagger -  \widetilde{\mathbf b}_{k'}^\dagger \mathbf b_{k}^\dagger\mathbf b_l = \pm\widetilde{\mathbf b}_{k'}^\dagger[\mathbf b_l ,\mathbf b_{k}^\dagger ]_{\mp} =\pm\widetilde{\mathbf b}_{k'}^\dagger\delta_{lk}, 
    $$
    $$
    [\widetilde{\mathbf b}_l, \widetilde{\mathbf b}_{k'}^\dagger \mathbf b_{k}^\dagger] = \widetilde{\mathbf b}_l \widetilde{\mathbf b}_{k'}^\dagger \mathbf b_{k}^\dagger \mp \widetilde{\mathbf b}_{k'}^\dagger \widetilde{\mathbf b}_l\mathbf b_{k}^\dagger = [\widetilde{\mathbf b}_{l}, \widetilde{\mathbf b}_{k'}^\dagger]_{\mp}\mathbf b_{k}^\dagger = \delta_{k'l}\mathbf b_{k}^\dagger.
    $$
\end{proof}
\subsection{Spin-boson model}
\label{sec:spin_boson_proof}
We have established the commutator relations in the unitary case in Lemma \ref{lem:commutator_L}. Let us also state the commutator result for the Lindbladian case.

\begin{lem}
    With $\boldsymbol{L}_{\text e} = \boldsymbol L_{\text e}^{\text U} + \boldsymbol D_{\text e}$ defined in \cref{eq:spin_boson_env}, we have:
 \begin{equation}
        \begin{aligned}
      \relax  [\boldsymbol b_l, \boldsymbol L_{\text e}] &= \sum_{k=1}^{N_{\text e}}(-\mathrm i H_{lk}-\Gamma_{lk})\boldsymbol b_{k},\quad [\boldsymbol b_l^\dagger, \boldsymbol L_{\text e}] = \sum_{k=1}^{N_{\text e}} (\mathrm i H_{kl}+\Gamma_{kl})\boldsymbol  b_k^\dagger - 2\Gamma_{kl}\widetilde{\boldsymbol b}_k, \\
        [\widetilde{\boldsymbol b}_l, \boldsymbol L_{\text e}] &= \sum_{k=1}^{N_{\text e}} (\mathrm i H_{kl}-\Gamma_{kl})\widetilde{\boldsymbol b}_{k},\quad [\widetilde{\boldsymbol b}_l^\dagger, \boldsymbol L_{\text e}] = \sum_{k=1}^{N_{\text e}}(-\mathrm i H_{lk}+\Gamma_{lk} )\widetilde{\boldsymbol b}_k^\dagger -2\Gamma_{lk}{\boldsymbol b}_k.
        \end{aligned}\label{eq:commutator_Le_D}
    \end{equation}
    \label{lem:commutator_L_D}
\end{lem}
The proof of this lemma follows naturally from Lemma \ref{lem:commutator_L}.

Now let us prove \cref{lem:corr_spin_boson}. Let us first consider the unitary case. We have the following lemma:
\begin{lem}
    Let $\mathcal C_{\boldsymbol b, \boldsymbol b^\dagger}(t)$ be a matrix-valued function, with its $(k,l)$-th element being $\langle{\boldsymbol b_k\mathrm e^{\boldsymbol L_{\text e}t} \boldsymbol b_l^\dagger}\rangle_{\text e}$ for $t\geq 0$. Similarly, let us define $\mathcal C_{\boldsymbol O, \boldsymbol O'}(t)$ where $\boldsymbol O,\boldsymbol O' \in \boldsymbol b, \boldsymbol b^\dagger, \widetilde{\boldsymbol b},\widetilde{\boldsymbol b^\dagger}$. Then, in the unitary case (see \cref{eq:spin_boson_env,eq:Le_U_fermionic}), the nonzero components of $\mathcal C_{\boldsymbol O,\boldsymbol O'}(t)$ ($t\geq 0$) are:
    \begin{equation}
      \begin{aligned}
        \mathcal C_{\boldsymbol b,\boldsymbol b^\dagger}(t) &=\mathrm{C}_{\widetilde{\boldsymbol{b}}^{\dagger}, \boldsymbol{b}^{\dagger}}(t)= G^>(t),\quad 
        \mathcal C_{\boldsymbol b^\dagger, \boldsymbol b}(t) = C_{\widetilde{\boldsymbol b}, \boldsymbol{b}}(t)=(G^<(t))^{*},\\
        \mathcal C_{\widetilde{\boldsymbol b}^{\dagger}, \widetilde{\boldsymbol b}}(t) &= \mathcal C_{\boldsymbol b, \widetilde{\boldsymbol b}}(t) = G^<(t),\quad
        \mathcal C_{\boldsymbol b^\dagger, \widetilde{\boldsymbol b}^{\dagger}}(t) = \mathcal C_{\widetilde{\boldsymbol b}, \widetilde{\boldsymbol{b}}^\dagger}(t) = (G^>(t))^{*},
      \end{aligned}
    \end{equation}
    where we have defined in \cref{eq:Greens_unitary_spin_boson} that,
    \begin{equation}
        G^<(t) = \left(\mathrm{e}^{\beta H}-I\right)^{-1} \mathrm{e}^{-\mathrm{i} H t},\quad G^>(t) = \left(I-\mathrm{e}^{-\beta H}\right)^{-1} \mathrm{e}^{-\mathrm{i} H t}.
    \end{equation}
    \label{lem:C_OO_unitary_spin_boson}
    \end{lem}
\begin{proof}
    Since $(\mathcal C_{\boldsymbol b, \boldsymbol b}(0))_{kl} = \operatorname{tr}(\boldsymbol b_k\boldsymbol b_l\hat\rho_{\text E}(0))=\operatorname{tr}(\hat b_k\hat b_l\hat\rho_{\text E}(0))=0$, and using \cref{lem:commutator_L}, we have 
    $\frac{\mathrm d}{\mathrm dt}(\mathcal C_{\boldsymbol b, \boldsymbol b})_{kl}=\sum_{k'=1}^{N_{\text e}}\mathrm i H_{lk'}(\mathcal C_{\boldsymbol b, \boldsymbol b})_{kk'}$. Therefore, $C_{\boldsymbol b, \boldsymbol b}(t) = 0$ for all $t>0$.
    Similarly, we have $\mathcal C_{\boldsymbol b^\dagger, \boldsymbol b^\dagger}(t) = 0$, $\mathcal C_{\widetilde{\boldsymbol b}, \widetilde{\boldsymbol b}}(t) = 0$, $\mathcal C_{\widetilde{\boldsymbol b}^\dagger, \widetilde{\boldsymbol b}^\dagger}(t) = 0$, $\mathcal C_{\boldsymbol b, \widetilde{\boldsymbol b}^{\dagger}}(t) = 0$, $\mathcal C_{\boldsymbol b^\dagger, \widetilde{\boldsymbol b}}(t) = 0$, 
    $ \mathcal C_{\widetilde{\boldsymbol b}, \boldsymbol b^\dagger}(t) = 0$, and $\mathcal C_{\widetilde{\boldsymbol b}^\dagger, \boldsymbol b}(t) = 0$. We also have:
\begin{equation}
    \begin{aligned}
        (\mathcal C_{\boldsymbol b, \boldsymbol b^\dagger}(0))_{kl} &= \operatorname{tr}( \hat b_k\hat b_l^\dagger\hat\rho_{\text E}(0))= \langle\hat b_k\hat b_l^\dagger\rangle_{\text e}=\operatorname{tr}(\hat b_l^\dagger\hat\rho_{\text E}(0)\hat b_k)= (\mathcal C_{\widetilde{\boldsymbol b}^\dagger, \boldsymbol b^\dagger}(0))_{kl} ,\\ (\mathcal C_{\boldsymbol b^\dagger, \boldsymbol b}(0) )_{kl}&= \operatorname{tr}( \hat b_k^\dagger\hat b_l\hat\rho_{\text E}(0)) = \langle\hat b_k^\dagger\hat b_l\rangle_{\text e}= \operatorname{tr}(\hat b_l\hat\rho_{\text E}(0)\hat b_k^\dagger)=(\mathcal C_{\widetilde{\boldsymbol b}, \boldsymbol b}(0))_{kl},\\
    (\mathcal C_{\widetilde{\boldsymbol b}, \widetilde{\boldsymbol b}^\dagger}(0))_{kl} &= \operatorname{tr}(\hat\rho_{\text E}(0)\hat b_l \hat b_k^\dagger)= \langle\hat b_l\hat b_k^\dagger\rangle_{\text e}= \operatorname{tr}(\hat b_k^\dagger\hat\rho_{\text E}(0)\hat b_l) = (\mathcal C_{\boldsymbol b^\dagger, \widetilde{\boldsymbol b}^\dagger}(0))_{kl} ,\\
    (\mathcal C_{\boldsymbol b, \widetilde{\boldsymbol b}}(0))_{kl}&= 
    \operatorname{tr}(\hat b_k\hat\rho_{\text E}(0)\hat b_l^\dagger)= \langle\hat b_l^\dagger\hat b_k\rangle_{\text e}=\operatorname{tr}(\hat\rho_{\text E}(0)\hat b_l^\dagger \hat b_k)=(\mathcal C_{\widetilde{\boldsymbol b}^\dagger, \widetilde{\boldsymbol b}}(0))_{kl} .
    \end{aligned}
    \label{eq:spin_bosn_C_initial}
\end{equation}
    Or in matrix form,
    $$
    \begin{aligned}
        \mathcal C_{\boldsymbol{b},\boldsymbol{b}^\dagger} (0)=\mathcal C_{\widetilde{\boldsymbol{b}}^\dagger,\boldsymbol{b}^\dagger} (0) = (I-\mathrm e^{-\beta H})^{-1},\quad
    \mathcal C_{\boldsymbol b^\dagger, \boldsymbol b}(0) = \mathcal C_{\widetilde{\boldsymbol b},\boldsymbol b}(0) = \left(\mathrm{e}^{\beta H^{\mathrm T}}-I\right)^{-1},\\
    \mathcal C_{\widetilde{\boldsymbol b}^\dagger,\widetilde{\boldsymbol b}}(0)= \mathcal C_{\boldsymbol b,\widetilde{\boldsymbol b}}(0)= \left(\mathrm e^{\beta H} - I\right)^{-1},\quad \mathcal C_{{\boldsymbol b}^\dagger, \widetilde{\boldsymbol b}^\dagger}(0) = \mathcal C_{\widetilde{\boldsymbol b}, \widetilde{\boldsymbol b}^\dagger}(0) = \left(I - \mathrm e^{-\beta H^{\mathrm T}}\right)^{-1},
    \end{aligned}
    $$
    and 
    $$
        \frac{\mathrm d}{\mathrm dt}(\mathcal C_{\boldsymbol{b}, \boldsymbol b^\dagger})_{kl}= \operatorname{tr}\left(\boldsymbol{b}_k \mathrm{e}^{\boldsymbol{L}_{\mathrm{e}} t}\left[\boldsymbol{L}_{\mathrm{e}}, \boldsymbol{b}_l^\dagger\right] \hat{\rho}_{\mathrm{E}}(0)\right)
    = \sum_{k'=1}^{N_{\text e}}-\mathrm i H_{k'l}(\mathcal C_{\boldsymbol b, \boldsymbol b})_{kk'},
    $$
    $$
    \frac{\mathrm d}{\mathrm dt}(\mathcal C_{\boldsymbol b^\dagger, \boldsymbol b})_{kl} = \operatorname{tr}
    \left(\boldsymbol b_k^\dagger \mathrm e^{\boldsymbol L_{\text e} t}\left[\boldsymbol L_{\text e}, \boldsymbol b_l\right] \hat\rho_{\text E}(0)\right) = \sum_{k'=1}^{N_{\text e}}\mathrm i H_{lk'}(\mathcal C_{\boldsymbol b^\dagger, \boldsymbol b})_{kk'},
    $$
    $$
    \frac{\mathrm d}{\mathrm dt}(\mathcal C_{\widetilde{\boldsymbol b}^\dagger, \widetilde{\boldsymbol b}})_{kl} = \operatorname{tr}\left(
        \widetilde{\boldsymbol b}_k^\dagger \mathrm e^{\boldsymbol L_{\text e} t}\left[\boldsymbol L_{\text e}, \widetilde{\boldsymbol b}_l\right] \hat\rho_{\text E}(0)
    \right) = \sum_{k'=1}^{N_{\text e}}-\mathrm i H_{k'l}(\mathcal C_{\widetilde{\boldsymbol b}^\dagger, \widetilde{\boldsymbol b}})_{kk'},
    $$
    $$
    \frac{\mathrm d}{\mathrm dt}(\mathcal C_{\widetilde{\boldsymbol{b}}, \widetilde{\boldsymbol b}^\dagger})_{kl}= \operatorname{tr}\left( \widetilde{\boldsymbol b}_k \mathrm e^{\boldsymbol L_{\text e} t}\left[\boldsymbol L_{\text e}, \widetilde{\boldsymbol b}_l^\dagger\right] \hat{\rho}_{\text E}(0)\right) = 
    \sum_{k'=1}^{N_{\text e}}\mathrm i H_{lk'}(\mathcal C_{\widetilde{\boldsymbol b}, \widetilde{\boldsymbol b}^\dagger})_{kk'}.
    $$
    Or in matrix form 
    $$
    \begin{aligned}
        \frac{\mathrm d}{\mathrm dt}\mathcal C_{\boldsymbol{b},\boldsymbol{b}^\dagger} = -\mathrm i\mathcal C_{\boldsymbol{b},\boldsymbol{b}^\dagger}H,\quad 
        \frac{\mathrm d}{\mathrm dt}\mathcal C_{\boldsymbol{b}^\dagger,\boldsymbol{b}} = \mathrm i \mathcal C_{\boldsymbol{b}^\dagger,\boldsymbol{b}}H^{\mathrm T},\quad 
        \frac{\mathrm d}{\mathrm dt}\mathcal C_{\widetilde{\boldsymbol{b}}^\dagger,\widetilde{\boldsymbol{b}}} = -\mathrm i\mathcal C_{\widetilde{\boldsymbol{b}}^\dagger,\widetilde{\boldsymbol{b}}}H,\quad
        \frac{\mathrm d}{\mathrm dt}\mathcal C_{\widetilde{\boldsymbol{b}},\widetilde{\boldsymbol{b}}^\dagger} = \mathrm i\mathcal C_{\widetilde{\boldsymbol{b}},\widetilde{\boldsymbol{b}}^\dagger}H^{\mathrm T}.
    \end{aligned}
    $$
    And similarly we have
    $$
    \begin{aligned}
        \frac{\mathrm d}{\mathrm dt}\mathcal C_{\widetilde{\boldsymbol{b}}^\dagger,\boldsymbol{b}^\dagger} = -\mathrm i\mathcal C_{\widetilde{\boldsymbol{b}}^\dagger,\boldsymbol{b}^\dagger}H,\quad 
        \frac{\mathrm d}{\mathrm dt}\mathcal C_{\widetilde{\boldsymbol{b}},\boldsymbol{b}} = \mathrm i \mathcal C_{\widetilde{\boldsymbol{b}},\boldsymbol{b}}H^{\mathrm T},\quad 
        \frac{\mathrm d}{\mathrm dt}\mathcal C_{{\boldsymbol{b}},\widetilde{\boldsymbol{b}}} = -\mathrm i\mathcal C_{{\boldsymbol{b}},\widetilde{\boldsymbol{b}}}H,\quad
        \frac{\mathrm d}{\mathrm dt}\mathcal C_{{\boldsymbol{b}^\dagger},\widetilde{\boldsymbol{b}}^\dagger} = \mathrm i\mathcal C_{{\boldsymbol{b}^\dagger},\widetilde{\boldsymbol{b}}^\dagger}H^{\mathrm T}.
    \end{aligned}
    $$ 
    And thus, the lemma is proved.
\end{proof}
\begin{proof}[Proof of \cref{lem:corr_spin_boson}, in the unitary case] 
Using \cref{lem:Liouvillian_unitary_spin_boson},
    we have
$$
C_{2i-1,2j-1}(t)  = -
\sum_{k,l=1}^{N_{\text e}}g_{ik}g_{jl}^*
(\mathcal C_{\boldsymbol b, \boldsymbol b^\dagger}(t))_{kl} + g_{ik}^*g_{jl}(\mathcal C_{\boldsymbol{b}^\dagger, \boldsymbol b}(t))_{kl} = -(g G^>(t)g^\dagger + (g G^<(t)g^\dagger)^{*})_{ij},
$$
$$
C_{2i-1,2j}(t) =
\sum_{k,l=1}^{N_{\text e}}g_{ik}g_{jl}^*
(\mathcal C_{\boldsymbol b, \widetilde{\boldsymbol b}}(t))_{kl} + g_{ik}^*g_{jl}(\mathcal C_{\boldsymbol{b}^\dagger, \widetilde{\boldsymbol b}^\dagger}(t))_{kl} = (g G^<(t)g^\dagger + (g G^>(t)g^\dagger)^{*})_{ij},
$$
$$
C_{2i,2j-1}(t) = \sum_{k,l=1}^{N_{\text e}}g_{ik}g_{jl}^*
(\mathcal C_{\widetilde{\boldsymbol b}^\dagger, \boldsymbol b^\dagger}(t))_{kl} + g_{ik}^* g_{jl} (\mathcal C_{\widetilde{\boldsymbol b}, \boldsymbol b}(t))_{kl} 
= (g G^>(t)g^\dagger + (g G^<(t)g^\dagger)^{*})_{ij},
$$
$$
C_{2i,2j}(t) = -\sum_{k,l=1}^{N_{\text e}}g_{ik}g_{jl}^*
(\mathcal C_{\widetilde{\boldsymbol b}^\dagger, \widetilde{\boldsymbol b}}(t))_{kl} + g_{ik}^* g_{jl} (\mathcal C_{\widetilde{\boldsymbol b}, \widetilde{\boldsymbol b}^\dagger}(t))_{kl} = -(g G^<(t)g^\dagger + (g G^>(t)g^\dagger)^*)_{ij}.
$$
Let $c_{ij}(t) = -C_{2i-1,2j-1}(t) $, then $C_{2i,2j-1}(t) = c_{ij}(t)$, $C_{2i,2j}(t) = -(c_{ij}(t))^*$, and $C_{2i-1,2j}(t) = (c_{ij}(t))^*$.  
\end{proof}
 Next, we turn to the Lindblad case:
\begin{proof}[Proof of \cref{lem:corr_spin_boson}, in the Lindblad case]
    Let us define $\mathcal C_{\boldsymbol O,\boldsymbol O'}$ the same as in \cref{lem:C_OO_unitary_spin_boson} for $\boldsymbol O,\boldsymbol O' = \boldsymbol b, \boldsymbol b^\dagger, \widetilde{\boldsymbol b},\widetilde{\boldsymbol b}^\dagger$. Since $\hat{\rho}_{\mathrm{E}}^{(0)}=|0\rangle_{\mathrm{E}}\left\langle\left. 0\right|_{\mathrm{E}}\right.$, then 
    $$
    \langle b_k\hat b_l^{\dagger}\rangle_{\text e} = \delta_{kl},\quad 
    \langle\hat b_l^{\dagger} b_k\rangle_{\text e} = 0.
    $$
    Thus using \cref{eq:spin_bosn_C_initial}, we know that  the only nonzero ones among all $C_{\boldsymbol O, \boldsymbol O'}(0)$, ($\boldsymbol O,\boldsymbol O' = \boldsymbol b, \boldsymbol b^\dagger, \widetilde{\boldsymbol b},\widetilde{\boldsymbol b}^\dagger$) are:
    $$
    {(\mathcal C_{\boldsymbol{b},\boldsymbol{b}^\dagger}(0))}_{kl} = (\mathcal C_{\widetilde{\boldsymbol{b}}^\dagger,\boldsymbol{b}^\dagger}(0))_{kl} =  (\mathcal C_{\boldsymbol{b}^\dagger,\widetilde{\boldsymbol{b}}^\dagger}(0))_{kl} = (\mathcal C_{\widetilde{\boldsymbol{b}},\widetilde{\boldsymbol{b}}^\dagger}(0))_{kl} = \delta_{kl}. 
    $$
    Using the result for $[\boldsymbol b_l, \boldsymbol L_{\text e}]$ and $[\widetilde{\boldsymbol b}_l, \boldsymbol L_{\text e}]$ in the above lemma, similar to \cref{lem:C_OO_unitary_spin_boson}, we have:
    $$
    \frac{\mathrm d}{\mathrm dt}\mathcal C_{\boldsymbol O,\boldsymbol b } = C_{\boldsymbol O,\boldsymbol b }(\mathrm iH^{\text T}+\Gamma^{\text T}),\quad 
    \frac{\mathrm d}{\mathrm dt}\mathcal C_{\boldsymbol O,\widetilde{\boldsymbol b} } = C_{\boldsymbol O,\widetilde{\boldsymbol b} }(-\mathrm iH+\Gamma),\quad \text{for }\boldsymbol O = \boldsymbol b, \boldsymbol b^\dagger, \widetilde{\boldsymbol b},\widetilde{\boldsymbol b}^\dagger.
    $$
    Note that $\mathcal C_{\boldsymbol O,\boldsymbol b}(0) = \mathcal C_{\boldsymbol O,\widetilde{\boldsymbol b}}(0) = 0$, thus $\mathcal C_{\boldsymbol O,\boldsymbol b}(t) = \mathcal C_{\boldsymbol O,\widetilde{\boldsymbol b}}(t) = 0$ for all $t$ and $\boldsymbol O = \boldsymbol b, \boldsymbol b^\dagger, \widetilde{\boldsymbol b},\widetilde{\boldsymbol b}^\dagger$. 
    For the rest, we have 
    $$
    \frac{\mathrm d}{\mathrm dt}\mathcal C_{\boldsymbol O, \boldsymbol b^\dagger} = \mathcal C_{\boldsymbol O, \boldsymbol b^\dagger}(-\mathrm iH-\Gamma) + 2\mathcal C_{\boldsymbol O, \widetilde{\boldsymbol b}} \Gamma,\quad
    \frac{\mathrm d}{\mathrm dt}\mathcal C_{\boldsymbol O, \widetilde{\boldsymbol b}^\dagger} = \mathcal C_{\boldsymbol O, \widetilde{\boldsymbol b}^\dagger}(\mathrm iH^{\text T}-\Gamma^{\text T}) + 2\mathcal C_{\boldsymbol O, \boldsymbol b} \Gamma^{\text T}.
    $$
    Using $\mathcal C_{\boldsymbol O, \boldsymbol b}(t) = \mathcal C_{\boldsymbol O, \widetilde{\boldsymbol b}}(t) = 0$ and the initial conditions, we know that $\mathcal C_{\boldsymbol{b}^\dagger, \boldsymbol b^\dagger}(t) = \mathcal C_{\widetilde{\boldsymbol{b}}, {\boldsymbol{b}}^\dagger}(t) = 
    \mathcal C_{\boldsymbol{b}, \widetilde{\boldsymbol{b}}^\dagger}(t) = \mathcal C_{\widetilde{\boldsymbol{b}}^\dagger, \widetilde{\boldsymbol{b}}^\dagger}(t) = 0$ for all $t$, and 
   \begin{equation}
        \mathcal C_{\boldsymbol{b}, \boldsymbol{b}^\dagger}(t) = \mathcal C_{\widetilde{\boldsymbol{b}}^\dagger, \boldsymbol{b}^\dagger}(t) = \mathrm e^{(-\mathrm iH-\Gamma )t},\quad 
    \mathcal C_{\boldsymbol b^\dagger, \widetilde{\boldsymbol b}^\dagger}(t) = \mathcal C_{\widetilde{\boldsymbol b}, \widetilde{\boldsymbol b}^\dagger}(t) = \mathrm e^{(\mathrm iH^{\text T}-\Gamma^{\text T})t} = (\mathrm e^{(-\mathrm iH-\Gamma )t})^*.
    \label{eq:C_OO_Lindblad_spin_boson}
   \end{equation}
    Thus
    $$
\begin{aligned}
    C_{2i-1,2j-1}(t) = -g \mathcal C_{\boldsymbol{b}, \boldsymbol{b}^\dagger}(t)g^\dagger = -g \mathrm e^{(-\mathrm iH-\Gamma )t}g^\dagger,\quad C_{2i-1,2j}(t) = g^*\mathcal C_{ \boldsymbol{b}^\dagger,\widetilde{\boldsymbol{b}}^\dagger}(t) g^{\text T} = (g \mathrm e^{(- \mathrm iH -\Gamma)t}g^\dagger)^{*},\\
    C_{2i,2j-1}(t)=g\mathcal C_{\widetilde{\boldsymbol{b}}^\dagger,\boldsymbol{b}^\dagger}(t) g^{\dagger} = g \mathrm e^{(-\mathrm iH-\Gamma )t}g^\dagger,\quad 
    C_{2i,2j}(t) = -g^*\mathrm{C}_{\widetilde{\boldsymbol{b}}, \widetilde{\boldsymbol b}^{\dagger}}(t)g^{\text T} = -(g \mathrm e^{( -\mathrm iH -\Gamma)t}g^\dagger)^*.
\end{aligned}
    $$
    Let $c(t) =g \mathrm e^{(-\mathrm iH-\Gamma )t}g^\dagger$, then $C_{2i-1,2j-1}(t) = -c_{ij}(t)$, 
    $C_{2i,2j-1}(t) = c_{ij}(t)$, 
    $C_{2i-1,2j}(t) = (c_{ij}(t))^*$,
    $C_{2i,2j}(t) = -(c_{ij}(t))^*$. 
\end{proof}
Finally, let us prove this result for quasi-Lindblad dynamics:
\begin{proof}[Proof of \cref{lem:BCF_quasi_lind}, in the bosonic case]
    Recall that in the quasi-Lindblad case, $\boldsymbol E_\alpha = \boldsymbol E_\alpha^{\text U} + \boldsymbol E_\alpha^{\text D}$, where $\boldsymbol E_\alpha^{\text U}$ and $\boldsymbol E_\alpha^{\text D}$ are defined in \cref{lem:Liouvillian_unitary_spin_boson} 
 and \cref{eq:system_env_dissipation_spin_boson} respectively.
    Since $\boldsymbol{L}_{\text e}$ is the same in the Lindblad and quasi-Lindblad case, 
  the result of $\mathcal C_{\boldsymbol O,\boldsymbol O'}(t)$ (see \cref{eq:C_OO_Lindblad_spin_boson}) still holds, and we have 
$$
C_{2i-1,2j-1}(t) = ((-\mathrm ig - M) \mathcal C_{\boldsymbol b, \boldsymbol b^\dagger}(t)(-\mathrm i g^\dagger - M^\dagger))_{ij} =  - ((g-\mathrm iM) \mathrm e^{(-\mathrm iH-\Gamma)t}(g^{\dagger}-\mathrm iM^{\dagger}) )_{ij}
$$
$$
\begin{aligned}
    C_{2i-1,2j }(t) &= ((-\mathrm ig^*-M^*)\mathcal C_{\boldsymbol b^\dagger, \widetilde{\boldsymbol{b}}^\dagger}(t)(\mathrm ig^{\text T}-M^{\text T}) + 2M^* \mathcal C_{\widetilde{\boldsymbol{b}} , \widetilde{\boldsymbol{b}}^\dagger}(t)(\mathrm ig^{\text T}-M^{\text T}) )_{ij}\\
    &= ((-\mathrm ig^*+M^*)\mathrm{e}^{\left(\mathrm{i} H^{\mathrm{T}}-\Gamma^{\mathrm{T}}\right) t}(\mathrm ig^{\text T}-M^{\text T}) )_{ij} = (((g-\mathrm iM) \mathrm e^{(-\mathrm iH-\Gamma)t}(g^{\dagger}-\mathrm iM^{\dagger}) )_{ij})^*,
\end{aligned}
$$
$$
\begin{aligned}
    C_{2i,2j-1}(t) &=( (\mathrm ig - M)\mathcal C_{\widetilde{\boldsymbol{b}}^\dagger, \boldsymbol{b}^\dagger}(t)(-\mathrm ig^\dagger - M^\dagger)  + 2M \mathcal C_{\boldsymbol{b}, \widetilde{\boldsymbol{b}}}(t)(-\mathrm ig^\dagger - M^\dagger) )_{ij} \\
    &= \left((g-\mathrm{i} M) \mathrm{e}^{(-\mathrm{i} H-\Gamma) t}\left(g^{\dagger}-\mathrm{i} M^{\dagger}\right)\right)_{i j}
\end{aligned}
$$
$$
C_{2i,2j}(t) = ((\mathrm ig^* - M^*)\mathcal C_{\widetilde{\boldsymbol{b}}, \widetilde{\boldsymbol{b}}^\dagger}(t)(\mathrm ig^{\text T}-M^{\text T}) )_{ij} = -(((g-\mathrm iM) \mathrm e^{(-\mathrm iH-\Gamma)t}(g^{\dagger}-\mathrm iM^{\dagger}) )_{ij})^*.
$$
Let $c(t) =(g-\mathrm{i} M) \mathrm{e}^{(-\mathrm{i} H-\Gamma) t}\left(g^{\dagger}-\mathrm{i} M^{\dagger}\right)$, then $C_{2i-1,2j-1}(t) = -c_{ij}(t)$, 
$C_{2i,2j-1}(t) = c_{ij}(t)$, 
$C_{2i-1,2j}(t) = (c_{ij}(t))^*$,
$C_{2i,2j}(t) = -(c_{ij}(t))^*$. 
\end{proof}

\subsection{Fermionic impurity model}
\label{sec:fermion_proof}
Let us first give the explicit expressions for $\boldsymbol{E}_{\alpha}^{\text D \mp}$ in \cref{eq:D_SE_fermion}:
  \begin{equation}
     \begin{aligned}
         \boldsymbol{E}_{4j-3}^{\text{D}-} &=   \sum_{k=1}^{N_1} -  M_{jk} \boldsymbol{c}_k,\quad  \boldsymbol{E}_{4j-2}^{\text{D}-} =  \sum_{k=1}^{N_1} -M_{jk}^*\boldsymbol{c}_k^{\dagger}+ 2 M_{jk}^*\widetilde{\boldsymbol c}_k,\quad \boldsymbol{E}_{4j-1}^{\text{D}-} = \sum_{k=1}^{N_1}- M_{jk}^* \widetilde{\boldsymbol{c}}_k,\\
         \boldsymbol{E}_{4 j}^{\text{D}-}&= \sum_{k=1}^{N_1}-M_{jk} \widetilde{\boldsymbol{c}}_k^{\dagger}- 2M_{jk}\boldsymbol c_k,\quad 
         \boldsymbol{E}_{4 j-3}^{\text{D}+} = \sum_{l=N_1+1}^N M_{jl} \boldsymbol{c}_l + 2  M_{jl}\widetilde{\boldsymbol{c}}_{l}^\dagger,\quad \boldsymbol{E}_{4 j-2}^{\text{D}+} = \sum_{l=N_1+1}^{N_{\text e}} M_{jl}^* \boldsymbol c_{l}^\dagger,\\
         \boldsymbol{E}_{4 j-1}^{\text{D}+}&=\sum_{l=N_1+1}^{N_{\text e}}(M_{jl})^* \widetilde{\boldsymbol{c}}_l-2(M_{jl})^*\boldsymbol c_k^\dagger,\quad 
         \boldsymbol{E}_{4 j }^{\text{D}+} = \sum_{l=N_1+1}^{N_{\text e}}M_{jl}\widetilde{\boldsymbol{c}}_{l}^\dagger.
     \end{aligned}
     \label{eq:E_alpha_fermionic_non_unitary} 
      \end{equation}
Now let us prove that with $\boldsymbol{E}_{\alpha}^{\text U}$ in 
\cref{eq:E_alpha_fermionic} and $\boldsymbol{E}_{\alpha}^{\text D \mp}$ in \cref{eq:E_alpha_fermionic_non_unitary}, the operator form and the superoperator form of $\boldsymbol{L}_{\text{SE}}^{\text U}$ and $\boldsymbol{D}_{\text{SE}}^\mp$ are equivalent. The key is to take advantage of the fermionic superselection rule to simplify the formulation of $\boldsymbol{L}_{\text{SE}}$ in the superoperator formalism.
\begin{lem}
Due to the fermionic superselection rule (see \cref{rem:SE_decomposition}), $\boldsymbol{L}_{\text{SE}}^{\text U}$ in \cref{eq:L_SE_U_fermion} and \cref{eq:L_SE_U_fermion_O} are equivalent, and $\boldsymbol D_{\text{SE}}^{\mp}$ in \cref{eq:D_SE_fermion} and \cref{eq:D_SE_fermion_2} are equivalent.
    \label{lem:L_SE_D_SE_Expression}
\end{lem}
\begin{proof}       Using \cref{eq:fermionic_extension}, we have
$$
\hat H_{\text{SE}} = \sum_{i=1}^n\sum_{k=1}^{N_{\text e}} \nu_{ik}
    \left(\hat P_{\text e}\hat c_k\right)
    \otimes \hat a_i^\dagger + \nu_{ik}^*\left(\hat c_k^\dagger\hat P_{\text e}\right)\otimes \hat a_i.
$$
Then        $$\hat H_{\text{SE}}\hat\rho = 
        \sum_{i=1}^n\sum_{k=1}^{N_{\text e}} 
     (\nu_{ik}((\hat P_{\text e}\hat c_k) \otimes \hat a_i^\dagger)\hat\rho -
     \nu_{ik}^* ( (\hat P_{\text e}\hat c_k^\dagger)\otimes \hat a_i)\hat\rho)
        = \sum_{i=1}^n\sum_{k=1}^{N_{\text e}} 
      (\nu_{ik}\boldsymbol c_k\otimes \boldsymbol a_i^\dagger + \nu_{ik}^*\boldsymbol c_k^\dagger\otimes \boldsymbol a_i)\hat\rho.$$
      From this, one can find the values of $\boldsymbol E^U_{4i-3}$, $\boldsymbol E^U_{4i-2}$, $\boldsymbol S_{4i-3}$, $\boldsymbol S_{4i-2}$. 
      Multiplying $\hat H_{\text{SE}}$ from the right, we have
    $$
    \hat\rho\hat H_{\text{SE}} = \sum_{i=1}^n\sum_{k=1}^{N_{\text e}} \nu_{ik} \hat\rho((-\hat c_k\hat P_{\text e})\otimes \hat a_i^\dagger) + \nu_{ik}^*\hat\rho((\hat c_k^\dagger\hat P_{\text e})\otimes \hat a_i).
    $$
    As a result of the  fermionic superselection rule \cref{rem:SE_decomposition}, we only need to consider $\hat\rho$  in the physical space $\mathcal P(\mathcal H_{\text E}\otimes \mathcal H_{\text S}))$, namely, $[\hat P,\hat\rho]=0$. Since $[\hat P, \hat H_{\text{SE}}]=0$, we have $[\hat P,\hat\rho\hat H_{\text{SE}}] = [\hat P, \hat\rho]\hat H_{\text{SE}} + \hat \rho [\hat P,\hat H_{\text{SE}}] = 0$. In other words, 
    $\hat P \hat\rho\hat H_{\text{SE}} = \hat\rho\hat H_{\text{SE}}\hat P$. 
    Multiplying $\hat P$ from the right, we have  $\hat\rho \hat H_{\text{SE}} = \hat P \hat\rho \hat H_{\text{SE}}\hat P$, and therefore, 
    $$
    \begin{aligned}
        \hat\rho \hat H_{\text{SE}}
    & = \sum_{i=1}^n\sum_{k=1}^{N_{\text e}} \nu_{ik} (\hat P_{\text e}\otimes\hat P_{\text s})\hat\rho((-\hat c_k\hat P_{\text e})\otimes \hat a_i^\dagger)(\hat P_{\text e}\otimes\hat P_{\text s}) + \nu_{ik}^*(\hat P_{\text e}\otimes\hat P_{\text s})\hat\rho((\hat c_k^\dagger\hat P_{\text e})\otimes \hat a_i)(\hat P_{\text e}\otimes\hat P_{\text s}) \\
     & = \sum_{i=1}^n\sum_{k=1}^{N_{\text e}} \nu_{ik} (\hat P_{\text e}\otimes\hat P_{\text s})\hat\rho(\hat c_k\otimes (-\hat a_i^\dagger\hat P_{\text s})) + 
     \nu_{ik}^*(\hat P_{\text e}\otimes\hat P_{\text s})\hat\rho(\hat c_k^\dagger\otimes (\hat a_i \hat P_{\text s})) \\
     &= \sum_{i=1}^n\sum_{k=1}^{N_{\text e}}  (\nu_{ik}\widetilde{\boldsymbol c}_k^\dagger\otimes \widetilde{\boldsymbol a}_i + \nu_{ik}^*\widetilde{\boldsymbol c}_k\otimes \widetilde{\boldsymbol a}_i^\dagger)\hat\rho.
    \end{aligned}
    $$
    Combining the above equalities one can see $\boldsymbol{L}_{\text{SE}}^{\text U}$ in \cref{eq:L_SE_U_fermion} and \cref{eq:L_SE_U_fermion_O} are equivalent. 
    Furthermore,
     with \cref{defn:fermionic_superoperators}, we have
    $$
    \hat a_i|_{\mathcal H}\hat\rho \hat c_k^\dagger|_{\mathcal H} = (\hat P_{\text e}\otimes \hat a_i)\hat\rho (\hat c_k^\dagger \otimes \hat I_{\text s})  = (\widetilde{\boldsymbol c}_k \otimes \boldsymbol a_i)\hat\rho,\quad
    \hat a_i^\dagger|_{\mathcal H}\hat\rho \hat c_k|_{\mathcal H} = (\hat P_{\text e}\otimes \hat a_i^\dagger)\hat\rho (\hat c_k \otimes \hat I_{\text s})  = (\widetilde{\boldsymbol c}_k^\dagger \otimes \boldsymbol a_i^\dagger)\hat\rho. 
    $$
    And since $\hat\rho\in \mathcal P(\mathcal H_{\text E}\otimes \mathcal H_{\text S})$, we have
    $$
    \hat c_k^\dagger|_{\mathcal H}\hat\rho\hat a_i|_{\mathcal H} = (\hat c_k^\dagger \otimes \hat I_{\text s})\hat\rho (\hat P_{\text e}\otimes \hat a_i) = (\hat P_{\text e}\otimes \hat P_{\text s})(\hat c_k^\dagger \otimes \hat I_{\text s})\hat\rho (\hat P_{\text e}\otimes \hat a_i)(\hat P_{\text e}\otimes \hat P_{\text s}) = -(\boldsymbol c_k^\dagger \otimes \widetilde{\boldsymbol a}_i^\dagger)\hat\rho,
    $$
    $$
    \hat c_k|_{\mathcal H}\hat\rho\hat a_i^\dagger|_{\mathcal H} = (\hat c_k \otimes \hat I_{\text s})\hat\rho (\hat P_{\text e}\otimes \hat a_i^\dagger) = (\hat P_{\text e}\otimes \hat P_{\text s})(\hat c_k \otimes \hat I_{\text s})\hat\rho (\hat P_{\text e}\otimes \hat a_i^\dagger)(\hat P_{\text e}\otimes \hat P_{\text s}) = -(\boldsymbol c_k \otimes \widetilde{\boldsymbol a}_i)\hat\rho.
    $$
    Thus, $\boldsymbol D_{\text{SE}}^{\mp}$ in \cref{eq:D_SE_fermion} and \cref{eq:D_SE_fermion_2} are equivalent.
\end{proof}

Then
similar to \cref{lem:C_OO_unitary_spin_boson}, in the unitary fermionic impurity model, we have:
\begin{lem}
    Let us define $\mathcal C_{\boldsymbol O, \boldsymbol O'}(t)$ the same as in \cref{lem:C_OO_unitary_spin_boson}, where $\boldsymbol O,\boldsymbol O' = \boldsymbol c, \boldsymbol c^\dagger, \widetilde{\boldsymbol c},\widetilde{\boldsymbol c}^\dagger$. Then, in the unitary case, the nonzero components of $\mathcal C_{\boldsymbol O,\boldsymbol O'}(t)$ (for $t\geq 0$) are:
    \begin{equation}
      \begin{aligned}
        \mathcal C_{\boldsymbol c,\boldsymbol c^\dagger}(t) &=-\mathrm{C}_{\widetilde{\boldsymbol{c}}^{\dagger}, \boldsymbol{c}^{\dagger}}(t)= G^>(t),\quad 
        \mathcal C_{\boldsymbol c^\dagger, \boldsymbol c}(t) = C_{\widetilde{\boldsymbol c}, \boldsymbol{c}}(t)=(G^<(t))^{*},\\
        \mathcal C_{\widetilde{\boldsymbol c}^{\dagger}, \widetilde{\boldsymbol c}}(t) &=- \mathcal C_{\boldsymbol c, \widetilde{\boldsymbol c}}(t) = G^<(t),\quad
        \mathcal C_{\boldsymbol c^\dagger, \widetilde{\boldsymbol c}^{\dagger}}(t) = \mathcal C_{\widetilde{\boldsymbol c}, \widetilde{\boldsymbol{c}}^\dagger}(t) = (G^>(t))^{*}.
      \end{aligned}
    \end{equation}
    where we have defind in \cref{eq:Greens_unitary_fermionic} that,
 $$
        G^<(t) = \left(\mathrm{e}^{\beta H}+I\right)^{-1} \mathrm{e}^{-\mathrm{i} H t},\quad G^>(t) = \left(I+\mathrm{e}^{-\beta H}\right)^{-1} \mathrm{e}^{-\mathrm{i} H t}.
  $$ 
    \label{lem:C_OO_unitary_fermionic}
    \end{lem}
    The proof of \cref{lem:C_OO_unitary_fermionic} is similar to that of \cref{lem:C_OO_unitary_spin_boson}, and we only emphasize the key difference:
    \begin{proof}
        Using \cref{lem:commutator_L}, similar to \cref{lem:C_OO_unitary_spin_boson}, the only nonzero components would be $\mathcal C_{\boldsymbol c,\boldsymbol c^\dagger}(t)$, $\mathcal C_{\boldsymbol c^\dagger, \boldsymbol c}(t)$, $\mathcal C_{\widetilde{\boldsymbol c}^{\dagger}, \widetilde{\boldsymbol c}}(t)$, $\mathcal C_{\boldsymbol c, \widetilde{\boldsymbol c}}(t)$, $\mathcal C_{\boldsymbol c^\dagger, \widetilde{\boldsymbol c}^{\dagger}}(t)$,  $\mathcal C_{\widetilde{\boldsymbol c}, \widetilde{\boldsymbol c}^\dagger}(t)$, $\mathcal C_{\widetilde{\boldsymbol c}^\dagger, \boldsymbol c^\dagger}(t)$, and $\mathcal C_{\widetilde{\boldsymbol c}, \boldsymbol c}(t)$. Their initial values, however, are different from those in \cref{lem:C_OO_unitary_spin_boson} because of the parity operator in the definition of the fermionic superoperators \cref{defn:fermionic_superoperators}. We have:
       $$
    \begin{aligned}
        (\mathcal C_{\boldsymbol c, \boldsymbol c^\dagger}(0))_{kl} &= \operatorname{tr}(-\hat P_{\text e} \hat c_k\hat P_{\text e}\hat c_l^\dagger\hat\rho_{\text E}(0))=\operatorname{tr}( \hat c_k\hat c_l^\dagger\hat\rho_{\text E}(0)) 
        =
        -\operatorname{tr}(-\hat P_{\text e}\hat P_{\text e}\hat c_l^\dagger\hat\rho_{\text E}(0)\hat c_k)= -(\mathcal C_{\widetilde{\boldsymbol c}^\dagger, \boldsymbol c^\dagger}(0))_{kl} ,\\ 
        (\mathcal C_{\boldsymbol c^\dagger, \boldsymbol c}(0) )_{kl}&= \operatorname{tr}(-\hat P_{\text e} \hat c_k^\dagger\hat P_{\text e}\hat c_l\hat\rho_{\text E}(0)) = \operatorname{tr}(\hat c_k^\dagger\hat c_l\hat\rho_{\text E}(0)) = \operatorname{tr}(\hat P_{\text e}\hat P_{\text e}\hat c_l\hat\rho_{\text E}(0)\hat c_k^\dagger)=(\mathcal C_{\widetilde{\boldsymbol c}, \boldsymbol c}(0))_{kl},\\
    (\mathcal C_{\widetilde{\boldsymbol c}, \widetilde{\boldsymbol c}^\dagger}(0))_{kl} &= \operatorname{tr}(\hat P_{\text e}\hat P_{\text e}\hat\rho_{\text E}(0)\hat c_l \hat c_k^\dagger)= \operatorname{tr}(\hat\rho_{\text E}(0)\hat c_l \hat c_k^\dagger)= \operatorname{tr}(-\hat P_{\text e}\hat c_k^\dagger\hat P_{\text e}\hat\rho_{\text E}(0)\hat c_l) = (\mathcal C_{\boldsymbol c^\dagger, \widetilde{\boldsymbol c}^\dagger}(0))_{kl} ,\\
    (\mathcal C_{\widetilde{\boldsymbol c}^\dagger, \widetilde{\boldsymbol c}}(0))_{kl} &=\operatorname{tr}(\hat P_{\text e}\hat P_{\text e}\hat\rho_{\text E}(0)\hat c_l^\dagger \hat c_k) = \operatorname{tr}(\hat\rho_{\text E}(0)\hat c_l^\dagger \hat c_k) = \operatorname{tr}(-\hat P_{\text e}\hat c_k\hat P_{\text e}\hat\rho_{\text E}(0)\hat c_l^\dagger) = -(\mathcal C_{\boldsymbol c, \widetilde{\boldsymbol c}}(0))_{kl} .
    \end{aligned}
    $$ 
    Here, we have repeatedly used that $\hat P_{\text e} \hat c_k^{\dagger} = - \hat c_k^{\dagger}\hat P_{\text e}$, and $\hat P_{\text e}^2 = \hat I_{\text e}$. In matrix form, the above becomes:
    $$
    \begin{aligned}
        \mathcal C_{\boldsymbol{c},\boldsymbol{c}^\dagger} (0)=-\mathcal C_{\widetilde{\boldsymbol{c}}^\dagger,\boldsymbol{c}^\dagger} (0) = (I+\mathrm e^{-\beta H})^{-1},\quad
    \mathcal C_{\boldsymbol c^\dagger, \boldsymbol c}(0) = \mathcal C_{\widetilde{\boldsymbol c},\boldsymbol c}(0) = \left(\mathrm{e}^{\beta H^{\mathrm T}}+I\right)^{-1},\\
    \mathcal C_{\widetilde{\boldsymbol c}^\dagger,\widetilde{\boldsymbol c}}(0)= -\mathcal C_{\boldsymbol c,\widetilde{\boldsymbol c}}(0)= \left(\mathrm e^{\beta H} + I\right)^{-1},\quad \mathcal C_{{\boldsymbol c}^\dagger, \widetilde{\boldsymbol c}^\dagger}(0) = \mathcal C_{\widetilde{\boldsymbol c}, \widetilde{\boldsymbol c}^\dagger}(0) = \left(I + \mathrm e^{-\beta H^{\mathrm T}}\right)^{-1}.
    \end{aligned}
    $$ 
    The equation of motion for these correlation functions is the same as in \cref{lem:C_OO_unitary_spin_boson}, and thus the lemma is proved.
    \end{proof}
    Now let us prove \cref{lem:corr_fermionic} in the unitary case.
    \begin{proof}[Proof of \cref{lem:corr_fermionic}, in the unitary case.]Recall that $\boldsymbol E_{\alpha} = \boldsymbol E_{\alpha}^{\text U}$ is defined in \cref{eq:E_alpha_fermionic}. Then
    we have the following nonzero components of the correlation functions (for $t\geq 0$):
        $$
            C_{4i-3, 4j-2}(t) = -\sum_{k,l=1}^{N_{\text e}}\nu_{ik}\nu_{jl}^*(\mathcal C_{\boldsymbol c, \boldsymbol c^\dagger}(t))_{kl} = -(\nu G^>(t)\nu^\dagger)_{ij}, 
        $$
         $$   C_{4i-3, 4j-1}(t) = \sum_{k,l=1}^{N_{\text e}}\nu_{ik}\nu_{jl}^*(\mathcal C_{\boldsymbol c, \widetilde{\boldsymbol c}}(t))_{kl} = 
        - (\nu G^<(t)\nu^\dagger)_{ij},$$
        $$C_{4i-2, 4j-3}(t) = -\sum_{k,l=1}^{N_{\text e}}\nu_{ik}^*\nu_{jl}(\mathcal C_{\boldsymbol c^\dagger, \boldsymbol c}(t))_{kl} = -((\nu G^<(t)\nu^\dagger)_{ij})^{*}, 
        $$
        $$
            C_{4i-2, 4j}(t) = \sum_{k,l=1}^{N_{\text e}}\nu_{ik}^*\nu_{jl}(\mathcal C_{\boldsymbol c^\dagger, \widetilde{\boldsymbol c}^\dagger}(t))_{kl} = ((\nu G^>(t)\nu^\dagger)_{ij})^{*}, 
        $$
        $$
        C_{4i-1, 4j-3}(t) = \sum_{k,l=1}^{N_{\text e}}\nu_{ik}^*\nu_{jl}(\mathcal C_{\widetilde{\boldsymbol c}, \boldsymbol c}(t))_{kl} =
            ((\nu G^<(t) \nu^\dagger)_{ij})^{*},
        $$
        $$ C_{4i-1, 4j}(t) = -\sum_{k,l=1}^{N_{\text e}}\nu_{ik}^*\nu_{jl}(\mathcal C_{\widetilde{\boldsymbol c}, \widetilde{\boldsymbol c}^\dagger}(t))_{kl} = -((\nu G^>(t) \nu^\dagger)_{ij})^{*},
        $$
        $$
        C_{4i, 4j-2}(t) = \sum_{k,l=1}^{N_{\text e}}\nu_{ik}\nu_{jl}^*(\mathcal C_{\widetilde{\boldsymbol c}^\dagger, \boldsymbol c^\dagger}(t))_{kl} = -(\nu G^>(t) \nu^\dagger)_{ij},
        $$
        $$ 
            C_{4i, 4j-1}(t) =  - \sum_{k,l=1}^{N_{\text e}}\nu_{ik}\nu_{jl}^*(\mathcal C_{\widetilde{\boldsymbol c}^\dagger, \widetilde{\boldsymbol c}}(t))_{kl} = -(\nu G^<(t) \nu^\dagger)_{ij}.
        $$
    Let us define 
    $$
    \Delta_{ij}^>(t) = (\nu G^>(t) \nu^\dagger)_{ij},\quad \Delta_{ij}^<(t) = (\nu G^<(t) \nu^\dagger)_{ij}.
    $$
    Thus,
    $$
  \begin{aligned}
   & C_{4i-3,4j-2}(t) = C_{4i,4j-2}(t) = -\Delta_{ij}^>(t),\quad &&
    C_{4i-3,4j-1}(t) = C_{4i,4j-1}(t) = -\Delta_{ij}^<(t),\\
  &  C_{4i-2,4j}(t) = -C_{4i-1,4j}(t) = (\Delta_{ij}^>(t))^*,\quad&&
    C_{4i-1,4j-3}(t) = -C_{4i-2,4j-3}(t) = (\Delta_{ij}^<(t))^*.
  \end{aligned}
    $$  
    \end{proof}
 
For the Lindblad and quasi-Lindblad fermionic impurity model, as a result of \cref{lem:commutator_L}, with $\boldsymbol L_{\text e}=\boldsymbol L_{\text e}^{\text U}+\boldsymbol D_{\text e}$ (see \cref{eq:L_SE_U_fermion}, \cref{eq:D_SE_fermion}), we have the following commutator result.
For $l=1,\ldots, N_1$, we have
    $$
\begin{aligned}
& {\left[\boldsymbol{c}_l, \boldsymbol{L}_{\mathrm{e}}\right]=\sum_{k=1}^{N_1}\left(-\mathrm{i} H_{l k}^--\Gamma_{l k}^-\right) \boldsymbol{c}_k, \quad\left[\boldsymbol{c}_l^{\dagger}, \boldsymbol{L}_{\mathrm{e}}\right]=\sum_{k=1}^{N_1}\left(\mathrm{i} H_{k l}^-+\Gamma_{k l}^-\right) \boldsymbol{c}_k^{\dagger}-2 \Gamma_{k l}^- \widetilde{\boldsymbol{c}}_k}, \\
& {\left[\widetilde{\boldsymbol{c}}_l, \boldsymbol{L}_{\mathrm{e}}\right]=\sum_{l=1}^{N_1}\left(\mathrm{i} H_{k l}^--\Gamma_{k l}^-\right) \widetilde{\boldsymbol{c}}_k, \quad\left[\widetilde{\boldsymbol{c}}_l^{\dagger}, \boldsymbol{L}_{\mathrm{e}}\right]=\sum_{k=1}^{N_1}\left(-\mathrm{i} H_{l k}^-+\Gamma_{l k}^-\right) \widetilde{\boldsymbol{c}}_k^{\dagger}+2 \Gamma_{l k}^- \boldsymbol{c}_k},
\end{aligned}
$$
and for $l=N_1+1,\ldots, N_{\text e}$, we have 
$$
\begin{aligned}
& \left[\boldsymbol{c}_l, \boldsymbol{L}_{\mathrm{e}}\right]=\sum_{k=N_1+1}^{N}\left(-\mathrm{i} H_{l k}^++\Gamma_{l k}^+\right) \boldsymbol{c}_k+2 \Gamma_{l k}^+ \widetilde{\boldsymbol{c}}_k^\dagger, \quad\left[\boldsymbol{c}_l^{\dagger}, \boldsymbol{L}_{\mathrm{e}}\right]=\sum_{k=N_1+1}^{N}\left(\mathrm{i} H_{k l}^+-\Gamma_{k l}^-\right) \boldsymbol{c}_k^{\dagger}, \\
& \left[\widetilde{\boldsymbol{c}}_l, \boldsymbol{L}_{\mathrm{e}}\right]=\sum_{k=N_1+1}^{N}\left(\mathrm{i} H_{k l}^++\Gamma_{k l}^+\right) \widetilde{\boldsymbol{c}}_k-2 \Gamma_{k l}^+{\boldsymbol{c}}_k^\dagger, \quad\left[\widetilde{\boldsymbol{c}}_l^{\dagger}, \boldsymbol{L}_{\mathrm{e}}\right]=\sum_{k=N_1+1}^{N}\left(-\mathrm{i} H_{l k}^+-\Gamma_{l k}^+\right) \widetilde{\boldsymbol{c}}_k^{\dagger}.
\end{aligned}
$$
 Now we are ready to calculate the BCFs for Lindblad  and quasi-Lindblad dynamics.
\begin{proof}[Proof of \cref{lem:corr_fermionic}, in the Lindblad case]
    Let us define $\mathcal C_{\boldsymbol O,\boldsymbol O'}(t)$ the same as in \cref{lem:C_OO_unitary_fermionic} for $\boldsymbol O,\boldsymbol O' = \boldsymbol c, \boldsymbol c^\dagger, \widetilde{\boldsymbol c},\widetilde{\boldsymbol c}^\dagger$. 
    Let us define $\mathcal C^{--}_{\boldsymbol O,\boldsymbol O'}(t)$,
    $\mathcal C^{-+}_{\boldsymbol O,\boldsymbol O'}(t)$, 
    $\mathcal C^{+-}_{\boldsymbol O,\boldsymbol O'}(t)$, and $\mathcal C^{++}_{\boldsymbol O,\boldsymbol O'}(t)$ as the subblocks of $\mathcal C_{\boldsymbol O,\boldsymbol O'}(t)$:
    $$
    \mathcal C_{\boldsymbol O,\boldsymbol O'}(t) = \begin{pmatrix}
        \mathcal C^{--}_{\boldsymbol O,\boldsymbol O'}(t) & \mathcal C^{-+}_{\boldsymbol O,\boldsymbol O'}(t)\\
        \mathcal C^{+-}_{\boldsymbol O,\boldsymbol O'}(t) & \mathcal C^{++}_{\boldsymbol O,\boldsymbol O'}(t) 
    \end{pmatrix}.
    $$
    $\mathcal C^{--}_{\boldsymbol O,\boldsymbol O'}$, $\mathcal C^{-+}_{\boldsymbol O,\boldsymbol O'}$, $\mathcal C^{+-}_{\boldsymbol O,\boldsymbol O'}$, and $\mathcal C^{++}_{\boldsymbol O,\boldsymbol O'}$ are of size $N_1\times N_1$, $N_1\times (N_{\text e}-N_1)$, $(N_{\text e}-N_1)\times N_1$, and $(N_{\text e}-N_1)\times (N_{\text e}-N_1)$, respectively.
    Since initial density operator is $\hat{\rho}_{\text E}^{(0)} = (\otimes_{k=1}^{N_1}|0\rangle_k\langle 0|)\otimes (\otimes_{k=N_1+1}^{N_{\text e}}|1\rangle_k\langle 1|)$, thus
    $$
    \langle \hat c_k\hat c_l^\dagger\rangle_{\text e} = \left\{
        \begin{array}{cc}
            \delta_{kl}, & 1\leq k,l\leq N_1,\\
            0, & \text{otherwise},
        \end{array}
    \right.,\quad 
    \langle \hat c_k^\dagger\hat c_l\rangle_{\text e} = \left\{
        \begin{array}{cc}
            \delta_{kl}, & N_1+1\leq k,l\leq N_{\text e},\\
           0, & \text{otherwise},
        \end{array}
    \right. .
    $$
    Thus
    $\mathcal C^{-+}_{\boldsymbol O,\boldsymbol O'}(0) = \mathcal C^{+-}_{\boldsymbol O,\boldsymbol O'}(0) = 0$ for all $\boldsymbol O,\boldsymbol O' = \boldsymbol c, \boldsymbol c^\dagger, \widetilde{\boldsymbol c},\widetilde{\boldsymbol c}^\dagger$. As a result, $\mathcal C^{-+}_{\boldsymbol O,\boldsymbol O'}(t) = \mathcal C^{+-}_{\boldsymbol O,\boldsymbol O'}(t) = 0$ for all $t\geq 0$.
   Since $\mathcal C_{\boldsymbol O,\boldsymbol O'}^{--}(t)$ has the same initial value as $\mathcal C_{\boldsymbol O,\boldsymbol O'}(t)$ in the proof of \cref{lem:Liouvillian_unitary_spin_boson}, then using the same argument, the only nonzero components are 
\begin{equation}
    \mathrm{C}_{\boldsymbol{c}, \boldsymbol{c}^{\dagger}}^{--}(t)=-\mathrm{C}_{\widetilde{\boldsymbol{c}}^{\dagger}, \boldsymbol{c}^{\dagger}}^{--}(t)=\mathrm{e}^{(-\mathrm{i} H^--\Gamma^-) t}, \quad \mathrm{C}^{--}_{\boldsymbol{c}^{\dagger}, \widetilde{\boldsymbol{c}}^{\dagger}}(t)=\mathrm{C}^{--}_{\widetilde{\boldsymbol{c}}, \widetilde{\boldsymbol{c}}^{\dagger}}(t)=\mathrm{e}^{\left(\mathrm{i} (H^-)^{\mathrm{T}}-(\Gamma^-)^{\mathrm{T}}\right) t}. \label{eq:C_OO_Fermion_lind_1}
\end{equation}
With a similar argument,for $\mathcal C_{\boldsymbol O,\boldsymbol O'}^{++}$, the only nonzero components are 
\begin{equation}
-\mathcal C_{\boldsymbol{c}, \widetilde{\boldsymbol c}}^{++}(t)=\mathrm{C}_{\widetilde{\boldsymbol{c}}^{\dagger}, \widetilde{\boldsymbol{c}}}^{++}(t)=\mathrm e^{(-\mathrm iH^+-\Gamma^+)t},\quad 
\mathcal C_{\boldsymbol{c}^{\dagger}, \boldsymbol{c}}^{++}(t)=\mathrm{C}_{\widetilde{\boldsymbol{c}}, \boldsymbol{c}}^{++}(t)=\mathrm e^{(\mathrm i(H^+)^{\text T}-(\Gamma^+)^{\text T})t},
 \label{eq:C_OO_Fermion_lind_2}
\end{equation}
Then let us define 
$\mathcal G^>(t) = \mathrm e^{(-\mathrm iH^--\Gamma^-)t}$, $\mathcal G^<(t) = \mathrm e^{(-\mathrm iH^+-\Gamma^+)t}$, and $\Delta^>(t) = \nu^- \mathcal G^>(t)(\nu^-)^\dagger$, $\Delta^<(t) = \nu^+\mathcal G^<(t)(\nu^+)^\dagger$, where $\nu^-,\nu^+$ are 
subblocks of matrix $\nu$ of size $n\times N_1$ and $n\times (N_{\text e}-N_1)$, respectively, i.e. $\nu = \begin{pmatrix}
    \nu^-,
    \nu^+
\end{pmatrix}$. Then we have:
$$
\begin{array}{ll}
C_{4 i-3,4 j-2}(t)=C_{4 i, 4 j-2}(t)=-\Delta_{i j}^{>}(t), & C_{4 i-3,4 j-1}(t)=C_{4 i, 4 j-1}(t)=-\Delta_{i j}^{<}(t) \\
C_{4 i-2,4 j}(t)=-C_{4 i-1,4 j}(t)=\left(\Delta_{ij}^{>}(t)\right)^*, & C_{4 i-1,4 j-3}(t)=-C_{4 i-2,4 j-3}(t)=\left(\Delta_{ij}^{<}(t)\right)^*.
\end{array}
$$

\end{proof}
\begin{proof}[Proof of \cref{lem:BCF_quasi_lind}, in the fermionic case]
    Let us define $\mathcal G^{>}(t)=\mathrm{e}^{\left(-\mathrm{i} H^{-}-\Gamma^{-}\right) t}, \mathcal G^{<}(t)=\mathrm{e}^{\left(-\mathrm{i} H^{+}-\Gamma^{+}\right) t}$, and  
    $$\Delta^{>}(t)=\left(\nu^{-}-\mathrm{i} M^{-}\right) \mathcal G^{>}(t)\left(\left(\nu^{-}\right)^{\dagger}-\mathrm{i}\left(M^{-}\right)^{\dagger}\right), \quad \Delta^{<}(t)=\left(\nu^{+}-\mathrm{i} M^{+}\right) \mathcal G^{<}(t)\left(\nu^{+}-\mathrm{i} M^{+}\right).$$
    Since $\boldsymbol L_{\text e}$ is the same in the Lindblad and quasi-Lindblad case, then the results for  $\mathcal C_{\boldsymbol O, \boldsymbol O'}(t)$ in the above proof still holds (\cref{eq:C_OO_Fermion_lind_1,eq:C_OO_Fermion_lind_2}), i.e.,
    the only nonzero components of $\mathcal C_{\boldsymbol O, \boldsymbol O'}(t)$ are
    $$
    \mathcal C_{\boldsymbol{c},\boldsymbol{c}^\dagger}^{--}(t) =  - \mathcal C^{--}_{\widetilde{\boldsymbol{c}}^\dagger,\boldsymbol{c}^\dagger}(t) = \mathcal G^>(t),\quad
    \mathcal C^{--}_{\boldsymbol{c}^\dagger,\widetilde{\boldsymbol{c}}^\dagger}(t) = \mathcal C^{--}_{\widetilde{\boldsymbol{c}},\widetilde{\boldsymbol{c}}^\dagger}(t) = (\mathcal G^>(t))^*,
    $$
    $$
    -\mathcal C_{\boldsymbol{c},\widetilde{\boldsymbol{c}}}^{++}(t) = \mathcal C^{++}_{\widetilde{\boldsymbol{c}}^\dagger,\widetilde{\boldsymbol{c}}}(t) = \mathcal G^<(t),\quad
    \mathcal C^{++}_{\boldsymbol{c}^\dagger,\boldsymbol{c}}(t) = \mathcal C^{++}_{\widetilde{\boldsymbol{c}}, \boldsymbol{c}}(t) = (\mathcal G^<(t))^*.
    $$
    Therefore,
    $$
    C_{4i-3,4j-2}(t) = \left((-\mathrm i\nu^--M^-)\mathcal C^{--}_{\boldsymbol c,\boldsymbol c^\dagger}(t)((-\mathrm i\nu^-)^\dagger - (M^-)^\dagger)\right)_{ij}= -\Delta_{ij}^>(t),
    $$
    $$
    \begin{aligned}
        C_{4i-3,4j-1}(t) &= \left((-\mathrm i\nu^+ + M^+)\mathcal C^{++}_{\boldsymbol c, \widetilde{\boldsymbol{c}}}(t)(\mathrm i(\nu^+)^\dagger+(M^+)^\dagger)+ 2M^+\mathcal C^{++}_{\widetilde{\boldsymbol c}^\dagger, \widetilde{\boldsymbol{c}}}(t)(\mathrm i(\nu^+)^\dagger+(M^+)^\dagger)\right)_{ij} \\
        & = -\left((\nu^+ -\mathrm i M^+)G^>(t)((\nu^+)^\dagger-\mathrm i(M^+)^\dagger)\right)_{ij} = -\Delta_{ij}^<(t),
    \end{aligned}
    $$
    $$
    C_{4i-2,4j-3}(t) = (-\mathrm i(\nu^+)^* +(M^+)^*)\mathcal C_{\boldsymbol c^\dagger, \boldsymbol c}^{++}(t)(-\mathrm i(\nu^+)^{\text T}+(M^+)^{\text T})_{ij} = -(\Delta_{ij}^<(t))^*,
    $$
    $$
    \begin{aligned}
        C_{4i-2, 4j}(t) &= \left((-\mathrm i(\nu^-)^* -(M^-)^* )\mathcal C_{\boldsymbol c^\dagger, \widetilde{\boldsymbol c}^\dagger}^{--}(t)(\mathrm i(\nu^-)^{\text T}-(M^-)^{\text T})+2(M^-)^*\mathcal C_{\widetilde{\boldsymbol{c}}, \widetilde{\boldsymbol c}^\dagger}^{--}(t)(\mathrm i(\nu^-)^{\text T}-(M^-)^{\text T})\right)_{ij} \\
        & = \left( ( (\nu^-)^* + \mathrm i (M^-)^*)(G^<(t))^* ((\nu^-)^\text{T} + \mathrm i (M^-)^\text{T})\right)_{ij} =(\Delta_{ij}^>(t))^*,
    \end{aligned}
    $$
    $$
    \begin{aligned}
        C_{4i-1,4j-3}(t) &= \left((\mathrm i(\nu^+)^* + (M^+)^*)\mathcal C_{\widetilde{\boldsymbol c}, \boldsymbol c}^{++}(t)(-\mathrm i (\nu^+)^{\text T}+(M^+)^{\text T})-2(M^+)^*\mathcal C_{\boldsymbol c^\dagger, \boldsymbol c}^{++}(t)(-\mathrm i (\nu^+)^{\text T}+(M^+)^{\text T})\right)_{ij} \\&= \left((\mathrm i(\nu^+)^* - (M^+)^*)\mathcal C_{\widetilde{\boldsymbol c}, \boldsymbol c}^{++}(t)(-\mathrm i (\nu^+)^{\text T}+(M^+)^{\text T})\right)_{ij} =  (\Delta_{ij}^<(t))^*, 
    \end{aligned}
    $$
    $$
    C_{4i-1,4j }(t) = \left((\mathrm i(\nu^-)^* - (M^-)^*)\mathcal C_{\widetilde{\boldsymbol c}, \widetilde{\boldsymbol{c}}^\dagger}^{--}(t)(\mathrm i(\nu^-)^{\text T}-(M^-)^{\text T})\right)_{ij} = - (\Delta_{ij}^>(t))^*,
    $$
    $$
    \begin{aligned}
        C_{4i,4j-2}(t) &= \left((\mathrm i\nu^- - M^-)\mathcal C_{\widetilde{\boldsymbol c}^\dagger, \boldsymbol c^\dagger}^{--}(t)(-\mathrm i(\nu^-)^\dagger - (M^-)^\dagger) - 2M^-\mathcal C_{\boldsymbol c, \boldsymbol c^\dagger}^{--}(t)(-\mathrm i(\nu^-)^\dagger - (M^-)^\dagger)\right)_{ij} \\&=\left((-\mathrm i\nu^- - M^-)G^<(t)(-\mathrm i(\nu^-)^\dagger - (M^-)^\dagger) \right)_{ij} 
        = -\Delta_{ij}^>(t) , 
    \end{aligned}
    $$
    and finally,
    $$
    C_{4i,4j-1}(t) = \left( (\mathrm i\nu^+ + M^+)\mathcal C_{\widetilde{\boldsymbol c}^\dagger, \widetilde{\boldsymbol c}}^{++}(t)(\mathrm i (\nu^+)^\dagger + (M^+)^\dagger) \right)_ {ij} = -\Delta_{ij}^<(t).
    $$
\end{proof}  

\end{document}